\documentclass{vldb}

\usepackage{graphicx}
\usepackage{subfigure}
\usepackage{balance}  

\usepackage{amssymb}
\usepackage[lined,boxed,commentsnumbered,ruled]{algorithm2e}
\usepackage{algorithmic}
\usepackage{mathtools}
\usepackage{multirow}
\usepackage{array}

\usepackage{bbding}
\usepackage{makecell}
\usepackage[hyphens]{url}
\usepackage{listings}

\newtheorem{theorem}{Theorem}
\newtheorem{proposition}{Proposition}

\newcommand{\vnorm}[1]{\left|\left|#1\right|\right|_{p}}

\newcommand{\ALID}{\textsf{ALID}}
\newcommand{\PALID}{\textsf{PALID}}

\newcommand{\nop}[1]{}

\newcommand{\rmv}[1]{} 

\usepackage[usenames]{color}

\newcommand{\hlit}[1]{{\color{black} #1}}

\begin{document}
\begin{sloppy}
%

\title{ALID: Scalable Dominant Cluster Detection}

\numberofauthors{3}

\author{
\alignauthor
Lingyang Chu\\
       \affaddr{Institute of Computing Technology, CAS\\ Beijing, China}\\
       \email{lingyang.chu@vipl.ict.ac.cn}
\alignauthor
Shuhui Wang\\
       \affaddr{Institute of Computing Technology, CAS\\ Beijing, China}\\
       \email{wangshuhui@ict.ac.cn}
\alignauthor
Siyuan Liu\\
       \affaddr{Heinz College\\ Carnegie Mellon University\\ Pittsburgh, USA}\\
       \email{siyuan@cmu.edu}
\and  
\alignauthor
Qingming Huang\\
       \affaddr{University of Chinese Academy of Sciences\\ Beijing, China}\\
       \email{qmhuang@jdl.ac.cn}
\alignauthor
Jian Pei\\
       \affaddr{School of Computing Science \\Simon Fraser University \\Vancouver, Canada}\\
       \email{jpei@cs.sfu.ca}
}

\maketitle


\begin{abstract}
Detecting dominant clusters is important in many analytic applications. The state-of-the-art methods find dense subgraphs on the affinity graph as the dominant clusters.  However, the time and space complexity of those methods are dominated by the construction of the affinity graph, which is quadratic with respect to the number of data points, and thus impractical on large data sets.  To tackle the challenge, in this paper, we apply \emph{Evolutionary Game Theory} (EGT) and develop a scalable algorithm, Approximate Localized Infection Immunization Dynamics ({\ALID}).  The major idea is to perform Localized Infection Immunization Dynamics (LID) to find dense subgraph within local range of the affinity graph.
\hlit{
LID is further scaled up with guaranteed high efficiency and detection quality by an estimated Region of Interest (ROI) and a carefully designed Candidate Infective Vertex Search method (CIVS).
}
{\ALID} only constructs small local affinity graphs and has a time complexity of $\mathcal{O}(C(a^*+\delta)n)$ and a space complexity of $\mathcal{O}(a^*(a^*+\delta))$, where $a^*$ is the size of the largest dominant cluster and $C\ll{n}$ and $\delta\ll{n}$ are small constants. We demonstrate by extensive experiments on both synthetic data and real world data that {\ALID} achieves state-of-the-art detection quality with much lower time and space cost on single machine. We also demonstrate the encouraging parallelization performance of {\ALID} by implementing the Parallel ALID ({\PALID}) on \emph{Apache Spark}. {\PALID} processes 50 million SIFT data points in 2.29 hours, achieving a speedup ratio of 7.51 with 8 executors.

\end{abstract}


\section{Introduction}
A dominant cluster is a group of highly similar objects that possesses maximal inner group coherence~\cite{SEA,DS}.
On a massive data set, more often than not the dominant clusters carry useful information and convey important knowledge.
For example, in a big collection of news data (e.g., official news, RSS-feeds and tweet-streams), the dominant clusters may indicate potential real world hot events~\cite{DSM_VLDB,APP_twitter_stand}; in a large repository of interpersonal communication data (e.g., emails and social networks), the dominant clusters may reveal stable social hubs~\cite{APP_socialhub}. Therefore, efficiently and effectively detecting dominant clusters from massive data sets has become an important task in data analytics.

In real applications, dominant cluster detection often faces two challenges.  First, the number of dominant clusters is often unknown. Second, large data sets are often noisy. Unknown number of meaningful dominant clusters are often hidden deeply in an overwhelming amount of background noise~\cite{APP_twitter_stand,APP_socialhub}. For instance, numerous news items about almost every aspect of our daily life are added to the Web everyday.  Most of the news items are interesting to small groups of people, and hardly attract sufficient social attention or become a dominant cluster of a hot event.  As another example, billions of spam messages are sent everyday.  Finding meaningful email threads and activities is a typical application of dominant cluster detection, and is challenging mainly due to the large amount of spam messages as noise~\cite{SEA}.

Traditional partitioning-based clustering methods like \emph{k}-means~\cite{SKM_VLDB,KM_Lloyd,KM_stream} and spectral clustering~\cite{NYS_PAMI04,NYS_muli,SC_tradition,SC_KDD12} are often used in cluster detection. However, such methods are not robust in processing noisy data with an unknown number of dominant clusters~\cite{DS}.
First, these methods typically require a pre-defined number of (dominant) clusters, without prior knowledge on the true number of (dominant) clusters, an improper number of clusters may lead to low detection accuracy.
Second, each data item, including both members of dominant clusters and noise data items, are forced to be assigned to a certain cluster, which inevitably leads to degenerated detection accuracy and subtracted cluster coherence under high background noise.

The affinity-based methods~\cite{SEA,DS,IID}, which detect dominant clusters by finding dense subgraphs on an affinity graph, are effective in detecting an unknown number of dominant clusters from noisy background.
Since the data objects in a dominant cluster are very similar to each other, they naturally form a highly cohesive dense subgraph on the affinity graph. Such high cohesiveness is proven to be a stable characteristic to accurately identify dominant clusters from noisy background without knowing the exact number of clusters~\cite{SEA,IID}.

Nevertheless, efficiency and scalability are the bottlenecks of the affinity-based methods, since the complexity of constructing the affinity matrix from $n$ data items is $\mathcal{O}(n^2)$ in both time and space. Even though the computational costs can be saved by forcing the affinity graph sparse~\cite{PSC}, the enforced sparsity breaks the intrinsic cohesiveness of dense subgraphs and consequently affects the detection quality of dominant clusters.

To tackle the challenge, in this paper, we propose Approximate Localized Infection Immunization Dynamics ({\ALID}), a dominant cluster detection approach that achieves high scalability and retains high detection quality.
The key idea is to avoid constructing the global complete affinity graph. Instead, {\ALID} finds a dense subgraph within an accurately estimated local Region of Interest (ROI). The ROI is guaranteed by the \emph{law of triangle inequality} to completely cover a dominant cluster, and thus fully preserves the high intrinsic cohesiveness of the corresponding dense subgraph and ensures the detection quality. Moreover, since dominant clusters generally exist in small local ranges, {\ALID} only searches a small local affinity subgraph within the ROI. Therefore, {\ALID} only constructs small local affinity graphs and largely avoids the expensive construction of the global affinity graph. Consequently, the original $\mathcal{O}(n^2)$ time and space complexity of the affinity graph construction is significantly reduced to $\mathcal{O}(C(a^*+\delta)n)$ in time and $\mathcal{O}(a^*(a^*+\delta))$ in space, where $a^*$ is the size of the largest dominant cluster and $C\ll{n}$, $\delta\ll{n}$ are small constants.

We make the following major contributions.
First, we propose LID to detect dense subgraphs on the local affinity graph within a ROI. LID localizes the Infection Immunization Dynamics~\cite{IID} to efficiently seek dense subgraphs on a small local affinity graph. It only computes a few columns of the local affinity matrix to detect a dense subgraph without sacrificing detection quality. This significantly reduces the time and space complexity.

\hlit{
Second, we estimate a Region of Interest (ROI) and propose a novel Candidate Infective Vertex Search method (CIVS) to significantly improve the scalability of LID and ensure high detection quality.
The estimated ROI is guaranteed to accurately identify the local range of the ``\emph{true}'' dense subgraph (i.e., dominant cluster), which ensures the detection quality of {\ALID}. The CIVS method is proposed to quickly retrieve the data items within the ROI, where all data items are efficiently indexed by Locality Sensitive Hashing (LSH)~\cite{LSH}. Demonstrated by extensive experiments on synthetic data and real world data, {\ALID} achieves substantially better scalability than the other affinity-based methods without sacrificing the detection quality.

}

Third, we carefully design a parallelized solution on top of the MapReduce framework~\cite{MapReduce} to further improve the scalability of {\ALID}. The promising parallelization performance of Parallel ALID ({\PALID}) is demonstrated by the experiments on \emph{Apache Spark} (http://spark.apache.org/). {\PALID} can efficiently process 50 million SIFT data~\cite{SIFT} in 2.29 hours and achieve a speedup ratio of $7.51$ with 8 executors.

For the rest of the paper, we review related work in Section~\ref{sec:related-work}. We revisit the problem of dense subgraph finding in Section~\ref{sec:dense-subgraph}. We present the {\ALID} method in Section~\ref{sec:ALID}. Section~\ref{sec:exp} reports the experimental results. Section~\ref{sec:con} concludes the paper.
We also provide a full version of this work in~\cite{appendix_pdf} with an appendix of proof materials and look-up tables.

\section{Related Work}\label{sec:related-work}
For dominant cluster detection, the affinity-based methods~\cite{SEA,DS,IID} that find dense subgraphs on affinity graphs are more resistant against background noise than the canonical partitioning-based methods like \emph{k}-means~\cite{SKM_VLDB,KM_Lloyd} and spectral clustering~\cite{NYS_PAMI04,NYS_muli,SC_KDD12}.
The dense subgraph seeking problem is well investigated in literature~\cite{SIGMOD_PANG_annotationGraphAPP,DSM_VLDB,SIGMOD_CSV_visualizing_subgraph}. Motzkin \textit{et al.}~\cite{motzkin_straus} proved that seeking dense subgraphs on an un-weighted graph can be formulated as a quadratic optimization problem on the simplex.
This method was extended to weighted graphs by the dominant set method (DS)~\cite{DS}, which solves a standard quadratic optimization problem (StQP) by replicator dynamics (RD)~\cite{Weibull}.

Bul\`o \textit{et al.}~\cite{IID} showed that, given the full affinity matrix of an affinity graph with $n$ vertices, the time complexity for each RD iteration is $\mathcal{O}(n^2)$, which hinders its application on large data sets. Thus, Bul\`o \textit{et al.}~\cite{IID} proposed the infection immunization dynamics (IID)~\cite{IID} to solve the StQP problem in $\mathcal{O}(n)$ time and space. However, the overall time and space complexity of IID is still $\mathcal{O}(n^2)$, since each iteration of IID needs the full affinity matrix, which costs quadratic time and space to compute and store.

Since most dense subgraphs exist in local ranges of an affinity graph, running RD on the entire graph is inefficient~\cite{SEA,GS}. Therefore, Liu \textit{et al.}~\cite{SEA} proposed the shrinking and expansion algorithm (SEA) to effectively prevent unnecessary time and space cost by restricting all RD iterations on small subgraphs. Both the time and space complexities of SEA are linear with respect to the number of graph edges~\cite{SEA}.  The scalability of SEA is sensitive to the sparse degree of the affinity graph.

Affinity propagation (AP)~\cite{AP} is another noise resistant solution to detect unknown number of dominant clusters. It finds the dominant clusters by passing real valued messages along graph edges, which is very time consuming when there are many vertices and edges.

\hlit{
Mean shift \cite{MS_PAMI_most_cited} differs from the affinity-based methods by directly seeking clusters in the feature space. It assumes that the discrete data items are sampled from a pre-defined density distribution and detect clusters by iteratively seeking the maxima of the density distribution. However, the detection quality of mean shift is sensitive to the type and bandwidth of the pre-defined density distribution.
}

The above affinity-based methods are able to achieve high detection quality when the affinity matrix is already materialized. However, the scalability of those methods with respect to large data sets is limited by the $\mathcal{O}(n^2)$ time and space complexities of the affinity matrix computation.
To the best of our knowledge, {\ALID} developed in this paper is the first attempt to achieve high scalability and retain good noise resistance.

\section{Dense Subgraph Finding Revisit}\label{sec:dense-subgraph}
In this section, we revisit the dense subgraph finding problem from the perspective of Evolutionary Game Theory (EGT)~\cite{Weibull} and discuss the common scalability bottleneck of the infection immunization dynamics (IID)~\cite{IID} and other affinity-based methods.

Consider a global affinity graph, denoted by $G=(V,I,A)$, where $V=\{v_i\in{R^d} \mid i\in{I=[1,n]}\}$ is the set of vertexes, and each vertex $v_i$ uniquely corresponds to a $d$-dimensional data point in $R^d$ space, $R$ being the set of real numbers, $I=[1,n]$ is the range of indices of all vertexes.  $A$ is the affinity matrix, where each entry $a_{ij}$ of $A$ represents the affinity between $v_i$ and $v_j$, that is,
\begin{equation}
\label{Eqn:exp_kernel_affinity}
a_{ij}=\left\{
\begin{array}{c l}
e^{-k\vnorm{v_i-v_j}} \ & i\neq{j} \\
0 \ & i=j
\end{array}\right.
\end{equation}
where $\vnorm{\cdot}$ represents the $L_p$-norm ($p\geq 1$) and $k>0$ is the scaling factor of the \emph{Laplacian Kernel}.

Given an affinity graph $G$ of $n$ vertexes, each graph vertex $v_i$ can be further referenced by an $n$-dimensional index vector $s_i=[\underbrace{0 \cdots 0}_{i-1} \; 1 \; \underbrace{0 \cdots 0}_{n-i}]^T$\rmv{, which is the $i$-th column of the \emph{identity matrix} with size $n$}. In other words, both $v_i$ and $s_i$ refer to the same $i$-th vertex in the graph.

\begin{figure}[t]
\centering
\includegraphics[width=82mm]{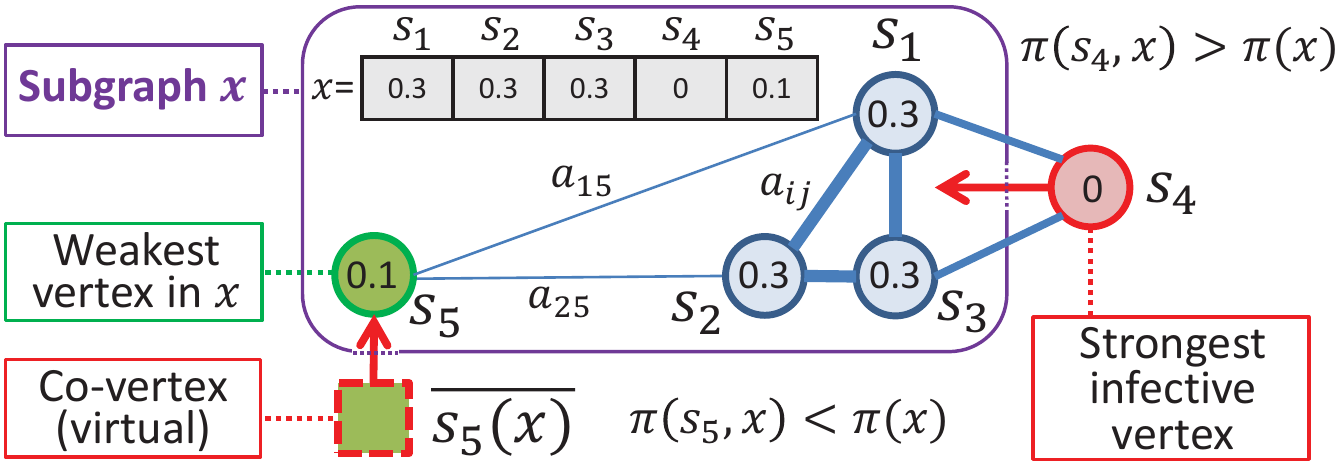}
\caption{An affinity graph $G$ with 5 vertexes $\{s_1, s_2, s_3, s_4, s_5\}$. Thicker edge means larger affinity. The purple rounded rectangle identifies a subgraph $x$ of $G$. The subgraph $x$ is composed of vertexes $\{s_1, s_2, s_3, s_5\}$ and is represented by $x=[0.3, 0.3, 0.3, 0, 0.1]^T, x\in\triangle^5$. $s_4$ is an infective vertex against $x$, since $\pi(s_4,x)>\pi(x)$. $s_5$ is a weak vertex in $x$, since $\pi(s_5,x)<\pi(x)$. The co-vertex $\overline{s_5(x)}$ \rmv{is a virtual vertex that} represents an infective subgraph composed of $\{s_1, s_2, s_3\}$.}
\label{Fig:intuition_example}
\end{figure}

A subgraph can be modeled by a subset of vertexes as well as their probabilistic memberships.
Take Figure~\ref{Fig:intuition_example} as an example.  We assign $L_1$ normalized non-negative weights to a subgraph $x$ containing vertexes $\{s_1, s_2, s_3, s_5\}$ and represent the subgraph by $x=0.3\cdot s_1+0.3\cdot s_2+0.3\cdot s_3+0\cdot s_4+0.1\cdot s_5$. Alternatively, $x$ can be regarded as an $n$-dimensional vector storing all the vertex weights, where the $i$-th dimension $x_i$ is the weight of vertex $s_i$.
Intuitively, $x_i$ embodies the probability that vertex $s_i$ belongs to subgraph $x$. Thus $x_i=0$ indicates that $s_i$ does not belong to $x$.  For example, vertex $s_4$ in Figure~\ref{Fig:intuition_example} does not belong to $x$.
In general, a subgraph can be represented by an $n$-dimensional vector $x\in\triangle^n$ in the standard simplex $\triangle^n=\{x\in{R^n}\mid\sum_{i\in I}{x_i}=1,x_i\geq0\}$.

The average affinity between two subgraphs $x, y\in{\triangle^n}$ is measured by the weighted affinity sum between all their member vertexes.  That is,
\begin{equation}
\pi(y,x)=y^TAx=\sum_i\sum_j{x_iy_ja_{ij}}
\end{equation}

As a special case, when $y=s_i$, that is, $y$ is subgraph of a single vertex, $\pi(s_i,x)=(s_i)^TAx$ represents the average affinity between vertex $s_i$ and subgraph $x$.  Moreover, the average affinity between subgraph $x$ and itself is $\pi(x)=\pi(x,x)=x^TAx$, which measures the internal connection strength between all vertexes of subgraph $x$. Liu and Yan~\cite{GS} indicated that such internal connection strength is a robust measurement of the intrinsic cohesiveness of $x$. Thus, $\pi(x)$ is also called the \textbf{density} of subgraph $x$.

\rmv{A dense subgraph~\cite{GS} is essentially a subset of vertexes in a graph that are coherent and robust against disturbances. Examples of dense subgraphs include communities in a social network or a spam farm in World Wide Web.}
\hlit{
Bul{\`o} \textit{et al.}~\cite{IID} indicated that a dense subgraph is a subgraph of local maximum density $\pi(x)$~\cite{GS,DS}, and every local maximum argument $x^*$ of $\pi(x)$ uniquely corresponds to a dense subgraph. Thus, the dense subgraph seeking problem can be reduced to the following standard quadratic optimization problem (StQP):
\begin{equation}
\label{Eqn:StQP}
\begin{array}{c}
\text{Maximize} \qquad \pi(x)=x^TAx=\sum_i x_i\pi(s_i,x) \\
\text{s.t.} \qquad x\in{\triangle^n}
\end{array}
\end{equation}
}
which can be solved by the Infection Immunization Dynamics (IID)~\cite{IID}. IID finds the dense subgraph $x^*$ by iteratively increasing the graph density $\pi(x)$ in the following two steps.
\begin{description}
  \item [Infection]: for a vertex $s_i$ whose average affinity $\pi(s_i,x)$ is larger than the \rmv{current} graph density $\pi(x)$, such as $\pi(s_4,x)>\pi(x)$ in Figure~\ref{Fig:intuition_example}, increase the its weight $x_i$.
  \item [Immunization]: for a vertex $s_i$ whose average affinity $\pi(s_i,x)$ is smaller than the \rmv{current} graph density $\pi(x)$, such as $\pi(s_5,x)<\pi(x)$ in Figure~\ref{Fig:intuition_example}, decrease its weight $x_i$.
\end{description}

\rmv{
The rationale of IID can be revealed by rewriting $\pi(x)$ to
\begin{equation}
\pi(x)=x^TAx=\sum_i x_i\pi(s_i,x)
\end{equation}
where $x_i$ is the weight of vertex $s_i$ and $\pi(s_i,x)$ is the average affinity between $s_i$ and $x$.
}

For the sake of clarity, we write $\pi(y-x,x)=\pi(y,x)-\pi(x)$ and define the relationship between subgraphs $x$ and $y$ as follow. If $\pi(y-x,x)>0$, then $y$ is said to be \emph{infective} against $x$.  Otherwise, $x$ is said to be \emph{immune} against $y$.  The set of infective subgraphs against $x$ is defined as:
\begin{equation}
\label{Eqn:gamma_x}
\gamma(x)=\{y\in{\triangle^n} \mid \pi(y-x,x)>0\}
\end{equation}

IID performs \textbf{infection} and \textbf{immunization} using the following \textbf{invasion model}.
\begin{equation}
\label{Eqn:invasion_model}
z=(1-\varepsilon)x+\varepsilon y
\end{equation}
The new subgraph $z\in\triangle^n$ is obtained by invading $x\in\triangle^n$ with a  subgraph $\varepsilon y$, where $\varepsilon\in{[0,1]}$ and $y\in\triangle^n$. This transfers an amount of $\varepsilon$ weight from the vertexes in subgraph $x$ to the vertexes in subgraph $y$, while the vertex weights in $z$ still sum to 1. In other words, the weights of vertexes in subgraph $x$ are decreased and the weights of vertexes in subgraph $y$ are increased.

For each iteration, IID selects the optimal graph vertex $s_i$ that maximizes the absolute value of $|\pi(s_i-x,x)|$ by function $s_i=M(x)$, where
\begin{equation}
\label{Eqn:best_mutant_strategy}
\begin{array}{l}
M(x)=\arg\max\limits_{s_i\in(\mathcal{C}_1 \cup \mathcal{C}_2)}\;|\pi(s_i-x,x)| \\
\begin{array}{l}
\mathcal{C}_1=\{s_i \mid \pi(s_i-x,x)>0)\}\\
\mathcal{C}_2=\{s_i \mid \pi(s_i-x,x)<0,x_i>0\}
\end{array}
\end{array}
\end{equation}

If $s_i=M(x)\in\mathcal{C}_1$, then $\pi(s_i,x)>\pi(x)$ and $s_i$ is the strongest infective vertex. In this case, an \textbf{infection} is performed by the invasion model (Equation~\ref{Eqn:invasion_model}) with $y=s_i$. This increases $\pi(x)$ by increasing the weight of infective vertex $s_i$. For example, in Figure~\ref{Fig:intuition_example}, invading $x$ with a weight of $\varepsilon$ of the strongest infective vertex $s_4$ transfers a weight of $\varepsilon$ from $\{s_1, s_2, s_3, s_5\}$ to $s_4$ and increases $\pi(x)$.

If $s_i=M(x)\in\mathcal{C}_2$, then $\pi(s_i,x)<\pi(x)$. Thus $x$ is immune against $s_i$ and $s_i$ is the weakest vertex in subgraph $x$. In such case, an \textbf{immunization} is performed by the invasion model (Equation~\ref{Eqn:invasion_model}) with $y=\overline{s_i(x)}$:
\begin{equation}
\label{Eqn:co_strategy}
\overline{s_i(x)}=\frac{x_i}{x_i-1}(s_i-x)+x
\end{equation}
Here, $\overline{s_i(x)}$ is named the \emph{co-vertex} of $s_i$, and represents an infective subgraph composed of all the vertexes in subgraph $x$ except $s_i$. Thus, invading $x$ with $y=\overline{s_i(x)}$ by the invasion model (Equation~\ref{Eqn:invasion_model}) reduces the weight of vertex $s_i$ and increases the weights of the other vertexes in subgraph $x$. For example, in Figure~\ref{Fig:intuition_example}, The co-vertex $\overline{s_5(x)}$ represents a subgraph composed of $\{s_1, s_2, s_3\}$. Invading subgraph $x$ with a weight of $\varepsilon$ of $\overline{s_5(x)}$ transfers a weight of $\varepsilon$ from $s_5$ to $\{s_1, s_2, s_3\}$ and increases $\pi(x)$.

Formally, the infective vertex (or co-vertex) $y$ of the invasion model (Equation~\ref{Eqn:invasion_model}) can be selected by $y=S(x)$:
\begin{equation}
\label{Eqn:selection_function}
S(x){=}\left\{
\begin{array}{l l}
s_i \ & \text{if  } s_i=M(x)\in\mathcal{C}_1 \\
\overline{s_i(x)} \ & \text{if  } s_i=M(x)\in\mathcal{C}_2
\end{array}\right.
\end{equation}

Bul{\`o} \textit{et al.}~\cite{IID} showed the following.

\begin{theorem}[\cite{IID}]
\label{theorem:nash_gamma_equal}
The following three statements are equivalent for $x\in\triangle^n$.
1) $x$ is immune against all vertexes $s_i\in\triangle^n$; 2) $\gamma(x)=\emptyset$; and 3) $x$ is a dense subgraph with local maximum $\pi(x)$.
\end{theorem}

According to Theorem~\ref{theorem:nash_gamma_equal}, IID searches for the global dense subgraph $x^*$ by iteratively shrinking $\gamma(x)$ until $\gamma(x)=\emptyset$.

Moreover, Bul{\`o} \textit{et al.}~\cite{Bulo50} showed the following.

\begin{theorem}
\label{theorem:infection_immune}
Let $y\in{\gamma(x)}$ and $z=(1-\varepsilon)x+\varepsilon{y}$, where $\varepsilon=\varepsilon_y(x)$ is defined as follows.
\begin{equation}
\label{Eqn:kesi_share}
\varepsilon_y(x)=\left\{
\begin{array}{c c}
\min\left[-\frac{\pi(y-x,x)}{\pi(y-x)},1\right] \ & \mbox{if} \; \pi(y-x)<0 \\[2mm]
1 \ & \mbox{otherwise}
\end{array}\right.
\end{equation}
Then, $y\not\in{\gamma(z)}$ and $\pi(z)>\pi(x)$.
\end{theorem}

According to Theorem~\ref{theorem:infection_immune}, any infective subgraph $y\in\gamma(x)$ can be excluded from $\gamma(z)$ by invading $x$ with weight $\varepsilon=\varepsilon_y(x)$ of $y$. This monotonously reduces the volume of $\gamma(z)$ and guarantees the convergence of IID.

IID needs the full affinity matrix to find the dense subgraph by solving the StQP problem (Equation~\ref{Eqn:StQP}). Therefore, its scalability is largely limited due to the $\mathcal{O}(n^2)$ time and space complexity in computing the complete affinity matrix $A$. Although forcing $A$ sparse can reduce such expensive time and space cost to some extent~\cite{PSC}, the enforced sparsity breaks the intrinsic cohesiveness of dense subgraphs, thus inevitably impairs the detection quality. Both the scalability of DS~\cite{DS} and SEA~\cite{SEA} are limited due to the same reason, since they need the complete affinity matrix $A$ as well.


\section{The ALID Approach}
\label{sec:ALID}
In this section, we introduce our {\ALID} method.  The major idea of {\ALID} is to confine the computation of infection and immunization in small local ranges within the Region of Interest (ROI), so that only small submatrices of the affinity matrix need to be computed. As a result, {\ALID} largely avoids the affinity matrix computation and significantly reduces both the time and space complexity.

The framework of {\ALID} is an iterative execution of the following three steps, as summarized in Figure~\ref{Fig:flowchart_ALID}.

\begin{figure}[h]
\centering
\includegraphics[height=47mm]{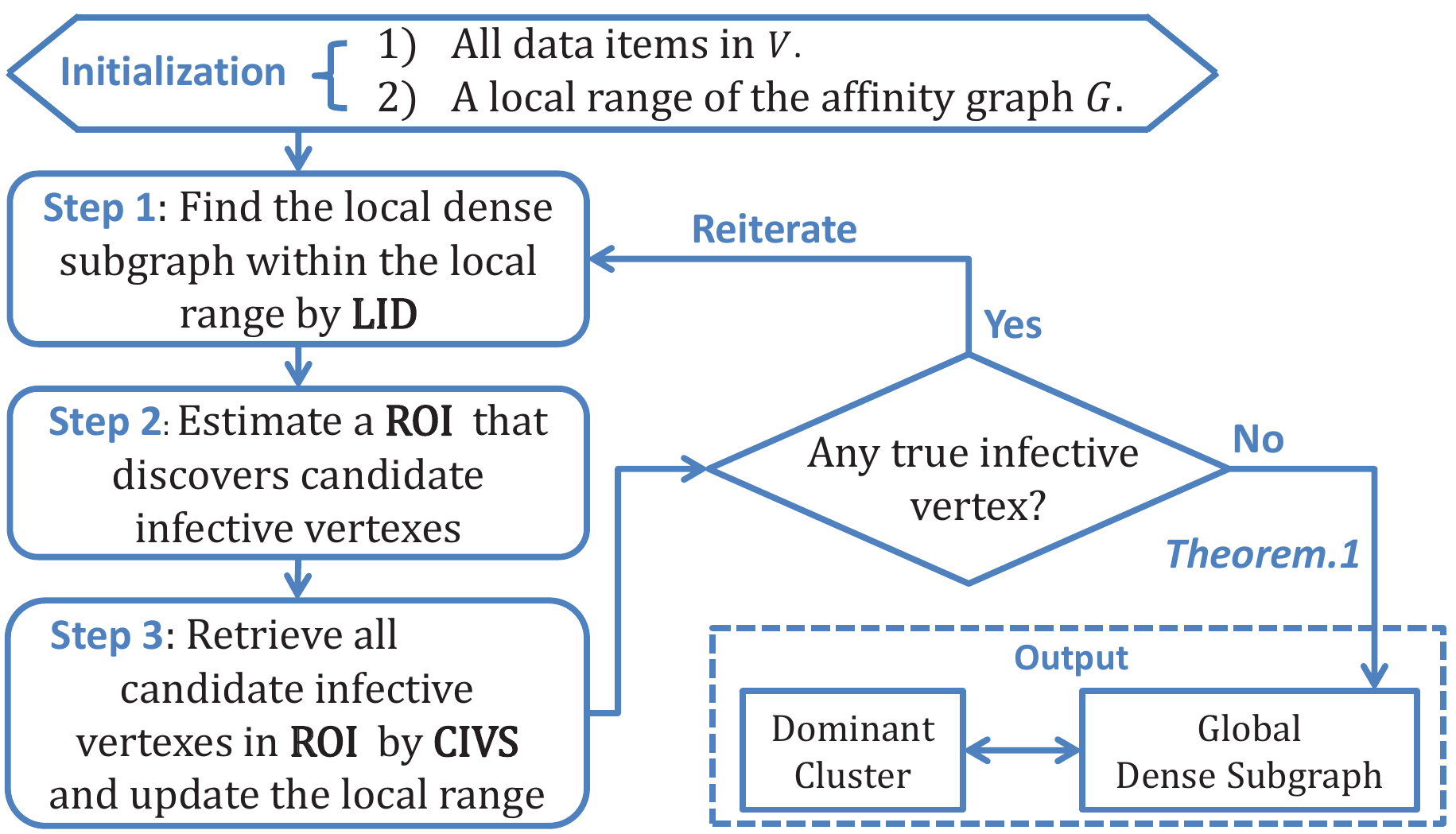}
\caption{{\ALID} framework overview.}
\label{Fig:flowchart_ALID}
\vspace{-3mm}
\end{figure}

\begin{description}
\item [Step 1] \emph{Finding local dense subgraph by Localized Infection Immunization Dynamics (LID)}. We propose a new algorithm LID that confines all infection immunization iterations in a small local range to find the local dense subgraph. LID avoids computing the full affinity matrix by selectively computing a few columns of the local affinity matrix.

\nop{
As illustrated in Figure~\ref{Fig:intuition_example}, given a local range, LID is able to find the local dense subgraph $\hat{x}$ composed of vertexes $\{s_1, s_2, s_3\}$. $\hat{x}$ does not contain vertex $s_4$, since $s_4$ is not covered by the local range.
}

\item [Step 2] \emph{Estimating a Region of Interest (ROI).} The local dense subgraph found in the LID step may or may not be a global dense subgraph, since there may still be global infective vertexes that are not covered by the current local range, as indicated by Theorem~\ref{theorem:nash_gamma_equal}.
    \nop{
    For example, the local dense subgraph $\hat{x}$ within the local range in Figure~\ref{Fig:intuition_example} is not a global dense subgraph, since vertex $s_4$ is still infective against $\hat{x}$.
    }
    Therefore, we estimate a ROI to identify the candidate infective vertexes in the global range, so that the current local range can be further updated to cover the global dense subgraph.

\item [Step 3] \emph{Candidate Infective Vertex Search (CIVS).} CIVS efficiently retrieves the candidate infective vertexes within a ROI and uses them to update the local range for the next iteration.
\end{description}

The iterative process of {\ALID} terminates when the local dense subgraph found in the last round of iteration is immune against all vertexes in the global range. According to Theorem~\ref{theorem:nash_gamma_equal}, such a subgraph is a global dense subgraph that identifies a true dominant cluster.
We explain the details of the three steps in the first three subsections as follow.

\subsection{Localized Infection Immunization Dynamics (Step 1)}
\label{Section:stage1}
The key idea of Localized Infection Immunization Dynamics (LID) is that the dense subgraph on a small local range of the affinity graph can be detected  by selectively computing only a few columns of the corresponding local affinity submatrix. Denote by $\beta\subset{I}$ the local range of the affinity graph, which is the index set of a local group of graph vertexes. LID finds the local dense subgraph $\hat{x}$ that maximizes $\pi(x)$ in the local range $\beta$ by localizing the infection immunization process on the selectively computed submatrix $A_{\beta\alpha}$ (see Figure~\ref{Fig:LID_matrix}). Here, $\alpha\triangleq\{i\in\beta \mid x_i>0\}$ is the \emph{support} of the subgraph $x\in\triangle_\beta^n$, where $\triangle_\beta^n=\{x\in \triangle^n \mid \sum_{i\in \beta}x_i=1\}$ represents the set of all possible subgraphs within local range $\beta$. In this way, the computation of the full matrix $A_{\beta\beta}$ can be effectively avoided.  Consequently, LID is more efficient than directly running IID in the local range.

\begin{figure}[h]
\centering
\includegraphics[width=38mm]{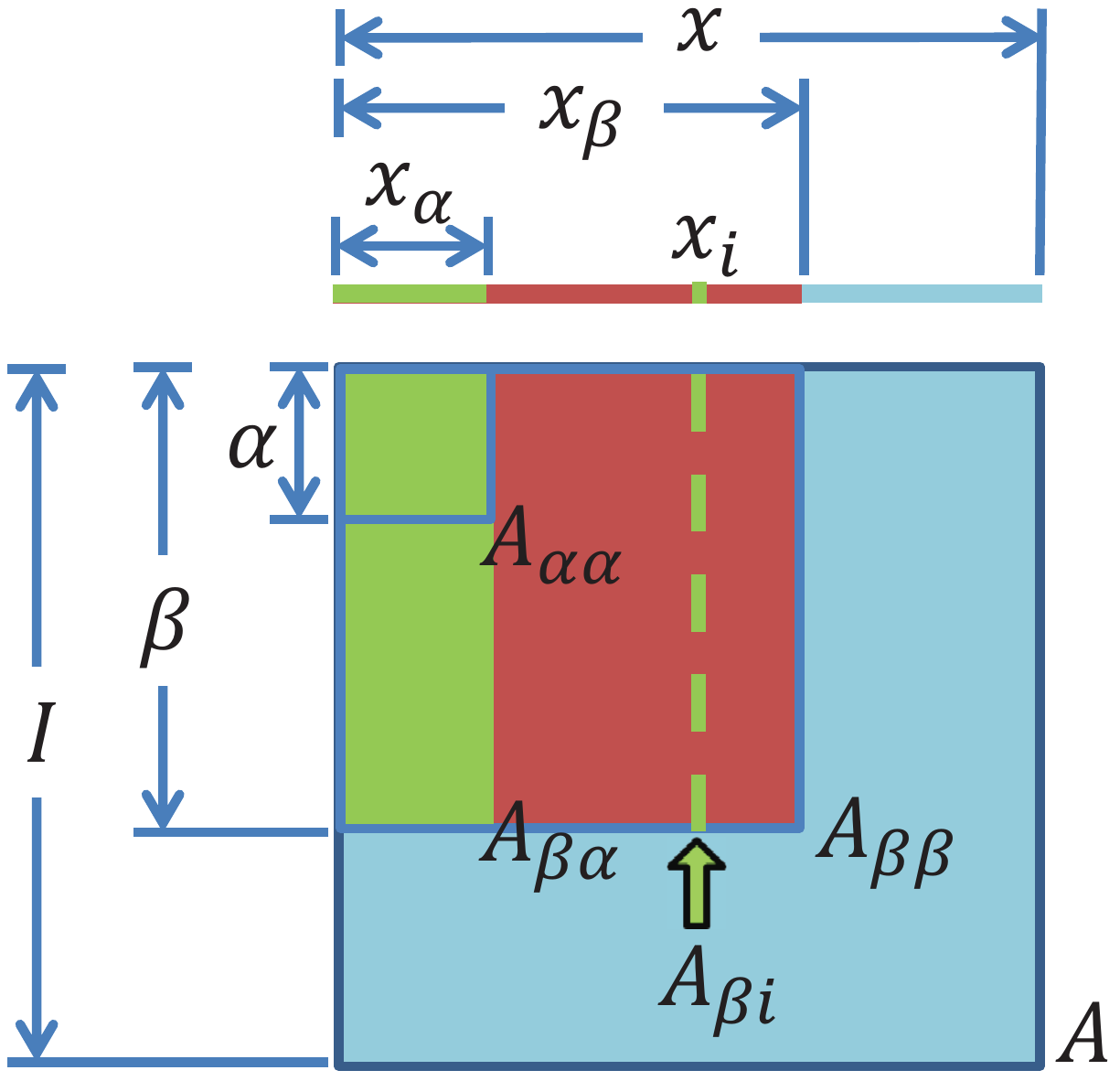}
\vspace{-2mm}
\caption{The affinity matrix with ordered data item indexes. $A$ is the global affinity matrix, $A_{\beta\beta}$ is the affinity matrix in the local range $\beta$, and $A_{\beta\alpha}$ is the group of columns in $A_{\beta\beta}$ corresponding to the support $\alpha$. The green dashed line denotes the new column $A_{\beta i}$ : $i\in{(\beta-\alpha)}$. Only the green parts ($A_{\beta\alpha}$ and $A_{\beta{i}}$) are involved in the LID iteration.}
\label{Fig:LID_matrix}
\vspace{-3mm}
\end{figure}

Algorithm~\ref{alg:the_core_iteration} shows the steps in a LID iteration, which takes $\left[x^{(t)}, (A_{\beta\alpha}x_\alpha)^{(t)}\right]$ as input and obtains the output $\left[x^{(t+1)}, (A_{\beta\alpha}x_\alpha)^{(t+1)}\right]$ in linear time and space with respect to the size of $\beta$. The superscripts $(t)$ and $(t+1)$ indicate the number of iterations. For the interest of simplicity, we omit $(t)$ and illustrate the details of Algorithm~\ref{alg:the_core_iteration} as follows.

First, we select the infective vertex (or co-vertex) $S(x)\in\triangle_\beta^n$ in the local range $\beta$ by Equation~\ref{Eqn:best_mutant_strategy} and \ref{Eqn:selection_function}. Recall that $x\in\triangle_\beta^n$, the component $\pi(s_i-x,x)$ can be computed as
\begin{equation}
\label{Eqn:local_payoff_diff}
\begin{array}{l l}
\pi(s_i-x,x) & =(s_i-x_\beta)^T A_{\beta\beta} x_\beta \\
                   & =(A_{\beta\alpha}x_\alpha)_i-\sum\limits_{i\in\alpha}{(A_{\beta\alpha}x_\alpha)_ix_i}
\end{array}
\end{equation}
where only the graph vertexes $s_i\in\triangle_\beta^n$ are considered.

Second, we compute the invasion share $\varepsilon_y(x)$ by Equation~\ref{Eqn:kesi_share}, whose value depends on $y=S(x)$ in two cases:
\begin{description}
  \item [Case 1]: \emph{Infection} $(y=S(x)=s_i)$. The key components $\pi(y-x,x)$ and $\pi(y-x)$ in Equation~\ref{Eqn:kesi_share} are computed by Equations~\ref{Eqn:local_payoff_diff} and~\ref{Eqn:proof_linearity_case_ei_kesi}, respectively.
      \begin{equation}
      \label{Eqn:proof_linearity_case_ei_kesi}
      \begin{array}{l l}
      \pi(s_i-x) & =(s_i-x_\beta)^T A_{\beta\beta} (s_i-x_\beta) \\[1mm]
                       & =-2(A_{\beta\alpha}x_\alpha)_i+\sum\limits_{i\in\alpha}{(A_{\beta\alpha}x_\alpha)_ix_i}
      \end{array}
      \end{equation}
  \item [Case 2]: \emph{Immunization} $(y=S(x)=\overline{s_i(x)})$. $\varepsilon_y(x)$ is computed by plugging Equation~\ref{Eqn:proof_linearity_case_barei_kesi} into Equation~\ref{Eqn:kesi_share}.
      \begin{equation}
      \label{Eqn:proof_linearity_case_barei_kesi}
      \left\{
      \begin{array}{l}
      \pi(\overline{s_i(x)}-x,x)=\frac{x_i}{x_i-1}\pi(s_i-x,x) \\[1mm]
      \pi(\overline{s_i(x)}-x)=(\frac{x_i}{x_i-1})^2\pi(s_i-x)
      \end{array}\right.
      \end{equation}

\end{description}

Last, the new subgraph $x^{(t+1)}$ is obtained by Equation~\ref{Eqn:next_iterate_on_beta} and $(A_{\beta\alpha}x_\alpha)^{(t+1)}$ is computed by Equation~\ref{Eqn:next_iterate_Ax} in linear time and space for the next iteration.
\begin{equation}
\label{Eqn:next_iterate_on_beta}
x^{(t+1)} = (1-\varepsilon_y(x))x+\varepsilon_y(x)S(x)
\end{equation}
\begin{equation}
\label{Eqn:next_iterate_Ax}
\begin{array}{l}
(A_{\beta\alpha}x_\alpha)^{(t+1)} = (Ax^{(t+1)})_\beta = \\[2mm]
\left\{
\begin{array}{l l}
\varepsilon_y(x)[A_{\beta i}{-}A_{\beta\alpha}x_\alpha]{+}A_{\beta\alpha}x_\alpha & \, y{=}s_i \\[1mm]
(\frac{x_i}{x_i-1})\varepsilon_y(x)[A_{\beta i}{-}A_{\beta\alpha}x_\alpha]{+}A_{\beta\alpha}x_\alpha & \, y{=}\overline{s_i(x)}
\end{array}\right.
\end{array}
\end{equation}

\begin{algorithm}[t]
\label{alg:the_core_iteration}
\caption{A single period of LID iteration}

\KwIn{$x^{(t)}$, $(A_{\beta\alpha}x_\alpha)^{(t)}$}
\KwOut{$x^{(t+1)}$, $(A_{\beta\alpha}x_\alpha)^{(t+1)}$}
\BlankLine

\begin{algorithmic}[1]
\STATE Select the infective vertex $y={S(x)}$ by Equation~\ref{Eqn:selection_function}

\STATE Calculate the invasion share $\varepsilon_y(x)$ by Equation~\ref{Eqn:kesi_share}

\STATE Update $x^{(t+1)}$ by the invasion model of Equation~\ref{Eqn:next_iterate_on_beta}

\STATE Update $(A_{\beta\alpha}x_\alpha)^{(t+1)}$ by Equation~\ref{Eqn:next_iterate_Ax}

\RETURN $x^{(t+1)}$, $(A_{\beta\alpha}x_\alpha)^{(t+1)}$
\end{algorithmic}
\end{algorithm}

\rmv{
\begin{algorithm}[t]
\label{alg:find_local_Nash}
\caption{\rmv{The entire LID iteration}}

\KwIn{$x$, $A_{\beta\alpha}x_\alpha$}
\KwOut{$\hat{x}$, $A_{\beta\alpha}\hat{x}_{\alpha}$}
\BlankLine

\begin{algorithmic}[1]
\REPEAT

\STATE  $\left[x^{(t)},(A_{\beta\alpha}x_\alpha)^{(t)}\right] \rightarrow \left[x^{(t+1)},(A_{\beta\alpha}x_\alpha)^{(t+1)}\right]$ \\ by {Algorithm~\ref{alg:the_core_iteration}}.
\STATE Index update: $t\leftarrow{t+1}$.
\UNTIL{$\gamma_\beta(x^{(t)})=\emptyset$, or $t>T$.}

\RETURN $\hat{x}=x^{(t)}$, $A_{\beta\alpha}\hat{x}_{\alpha}=(A_{\beta\alpha}x_\alpha)^{(t)}$
\end{algorithmic}
\end{algorithm}
}

Each LID iteration in Algorithm~\ref{alg:the_core_iteration} is guaranteed by {Theorem~\ref{theorem:infection_immune}} to shrink the size of the local infective subgraph set $\gamma_\beta(x)=\{y\in\triangle_\beta^n \mid \pi(y-x,x)>0\}$.
Thus, we obtain the local dense subgraph $\hat{x}\in\triangle_\beta^n$ in the local range $\beta$ by repeating Algorithm~\ref{alg:the_core_iteration} to shrink $\gamma_\beta(x)$ until $\gamma_\beta(x)=\emptyset$.
According to {Theorem~\ref{theorem:nash_gamma_equal}}, $\gamma_\beta(\hat x)=\emptyset$ indicates that $\hat x$ is immune against all vertexes $s_i\in\triangle_\beta^n$, thus $\hat x$ is a local dense subgraph with local maximum $\pi(\hat x)$ in the local range $\beta$. In practice, we stop the LID iteration when $\pi(x)$ is stable or the total number of iterations exceeds an upper limit $T$. Let $b=|\beta|$ be the size of $\beta$, the time and space complexities of LID method are $\mathcal{O}(Tb)$ and $\mathcal{O}(b)$, respectively, not including the time and space cost in computing the affinity matrix $A$. The initialization of $\left[x,A_{\beta\alpha}x_\alpha\right]$ will be discussed in Section~\ref{Section:Stage3}.

All {\ALID} iterations are restricted on the dynamically computed submatrix $A_{\beta\alpha}$, where the new matrix column $A_{\beta i}$ only needs to be computed and stored when $i\in(\beta-\alpha)$ (see Figure~\ref{Fig:LID_matrix}). We will discuss how LID effectively reduces the original $\mathcal{O}(n^2)$ time and space complexity of affinity matrix computation in Section~\ref{Section:Complexity_Analysis}.

\subsection{Estimating ROI (Step 2)}
\label{sec:Estimate_ROI}
A local dense subgraph $\hat{x}\in\triangle_\beta^n$ may not be a global dense subgraph in $\triangle^n$, since $\beta$ may not fully cover the true dense subgraph in the global range $I$. Therefore, there may still be graph vertexes in the complementary range $U=I-\beta$, which are infective against the local dense subgraph $\hat{x}$. Thus, the current local range $\beta$ should be updated to include such infective vertexes in $U$, so that the true dense subgraph with maximum graph density can be detected by LID.

A natural way to find the global dense subgraph $x^*$ is to keep invading $\hat{x}$ using the infective vertexes in $U$ until no infective vertex exists in the global range $I$.
However, fully scanning $U$ for infective vertexes leads to an overall time complexity of $\mathcal{O}(n^2)$ in detecting all dominant clusters, since $U$ contains all the remaining vertexes in $I$ with an overwhelming proportion of irrelevant vertexes.
To tackle this problem, we estimate a Region of Interest (ROI) from $\hat{x}$ to include all the infective vertexes and exclude most of the irrelevant ones. Only the limited amount of vertexes inside the ROI are used to update $\beta$, which largely reduces the amount of vertexes to be considered and effectively reduces the time and space complexity of {\ALID}.

Before estimating the ROI, we first construct a \textbf{double-deck hyperball} $H(D,R_{in},R_{out})$ from $\hat{x}$, where $D\in{R^d}$ is the ball center and $R_{in},R_{out}$ are the radiuses of the inner and outer balls, respectively, which are defined as follows.
\begin{equation}
\label{Eqn:define_hyperball}
\left\{
\begin{array}{l}
D=\sum\limits_{i\in\alpha}{v_i\hat x_i},\;where\;v_i\in{V}\;are\;the\;data\;items. \\[3mm]
R_{in}=\frac{1}{k}\ln(\frac{\lambda_{in}}{\pi(\hat{x})}),\;where\;\lambda_{in}{=}\sum\limits_{i\in\alpha}{\hat x_ie^{-k\vnorm{v_i-D}}} \\[1mm]
R_{out}=\frac{1}{k}\ln(\frac{\lambda_{out}}{\pi(\hat{x})}),\;where\;\lambda_{out}{=}\sum\limits_{i\in\alpha}{\hat x_ie^{k\vnorm{v_i-D}}}
\end{array}\right.
\end{equation}
where $k$ is the positive scaling factor in Equation~\ref{Eqn:exp_kernel_affinity}. Proposition~\ref{prop:triangle_inequality_hyperball} gives two important properties of the double-deck hyperball $H(D,R_{in},R_{out})$.

According to the two properties in Proposition~\ref{prop:triangle_inequality_hyperball}, the surfaces of the inner ball $H(D,R_{in})$ and the outer ball $H(D,R_{out})$ are two boundaries, which guarantee that every data entry inside the inner ball corresponds to an infective vertex and the ones outside the outer ball are non-infective.

\begin{proposition}
\label{prop:triangle_inequality_hyperball}
Given the local dense subgraph $\hat{x}$, the double-deck hyperball $H(D,R_{in},R_{out})$ has the following properties:
\begin{enumerate}
  \item $\forall{j}\in{I}$ and $\vnorm{v_j-D}<R_{in}$, $\pi(s_j-\hat{x},\hat{x})>0$; and
  \item $\forall{j}\in{I}$ and $\vnorm{v_j-D}>R_{out}$, $\pi(s_j-\hat{x},\hat{x})<0$.
\end{enumerate}
\end{proposition}

Proposition~\ref{prop:triangle_inequality_hyperball} is proved in Section~\ref{sec:double_hyperball_proof_appendix} of appendix in ~\cite{appendix_pdf}.

\hlit{
The \textbf{ROI} is defined as a growing hyperball $H_c(D,R)$, whose surface starts from the inner ball and gradually approaches the outer ball as the {\ALID} iteration continues. The radius $R$ is defined as
\begin{equation}
\label{Eqn:ROI_final_def}
R=R_{in}+\theta(c)(R_{out}-R_{in}) \\[1mm]
\end{equation}
where $\theta(c)=\frac{1}{1+e^{(4-c/2)}}$ is a shifted logistic function to control the growing speed and $c$ is the current number of {\ALID} iteration. When $c$ grows large, we have $\theta(c)\approx 1$ and $R\approx R_{out}$, thus, the ROI is guaranteed to coincide with the outer ball as $c$ grows. Since the outer ball $H(D,R_{out})$ is guaranteed by \hlit{Proposition~\ref{prop:triangle_inequality_hyperball}} to contain all infective vertexes in the global range $I$, the finally found local dense subgraph $\hat{x}$ within the ROI is guaranteed to be immune against all vertexes in $I$, thus $\hat{x}$ is a global dense subgraph according to Theorem~\ref{theorem:nash_gamma_equal}.
Moreover, starting the ROI from the small inner ball can effectively reduce the number of vertexes to be scanned in the first several {\ALID} iterations.
}

\subsection{Candidate Infective Vertex Search (Step 3)}
\label{Section:Stage3}
The hyperball of ROI identifies a small local region in the $d$-dimensional space $R^d$. The data items $v_i\in{V}$ inside the ROI correspond to the candidate graph vertexes, which are probably infective against the current local dense subgraph $\hat{x}$ and may further increase the graph density $\pi(\hat{x})$. Therefore, we carefully design the Candidate Infective Vertex Searching (CIVS) method to efficiently retrieve such data items inside the ROI and use them to update the local range $\beta^{(c)}$. The variable $c$ is the current number of {\ALID} iteration.

Retrieving the data items inside the ROI $H_c(D,R)$ is equivalent to a \emph{fixed-radius near neighbor problem} in the $d$-dimensional space $R^d$. For the interest of scalability, we adopt the Locality Sensitive Hashing (LSH) method~\cite{LSH} to solve this problem.
However, LSH can only retrieve data items within a small Locality Sensitive Region (LSR) of a query (Figure~\ref{Fig:hive}).
\hlit{
Thus, when using the hyperball center $D$ as a single LSH query (Figure~\ref{Fig:hive} (a)), the corresponding LSR may fail to find all data items within the ROI, which would prevent {\ALID} from converging to the optimal result.
}

To solve this problem, we propose CIVS to guarantee that {\ALID} converges to the optimal result. The convergence of {\ALID} is proved in Section~\ref{sec:convergence_proof_appendix} of the appendix in~\cite{appendix_pdf}. CIVS applies multiple LSH queries using all the supporting data items of $\hat{x}^{(c)}$ (i.e., $\{v_i\in{V} \mid \hat{x}_i^{(c)}>0\}$). As shown in Figure~\ref{Fig:hive} (b), the benefit of CIVS is that multiple LSRs effectively cover most of the ROI, thus most of the data items within the ROI can be retrieved. Specifically, CIVS first collect all the new data items that is retrieved by any of the supporting data items of $\hat{x}^{(c)}$, using LSH. Then, we retrieve at most $\delta$ new data items within the ROI that are the nearest to the ball center $D$. The index set of the retrieved data items is denoted by $\psi=\{i \mid i\in{(I-\alpha^{(c)})}, v_i\in H_c(D,R)\},|\psi|\leq{\delta}$.

\begin{figure}[t]
\centering
\includegraphics[width=60mm]{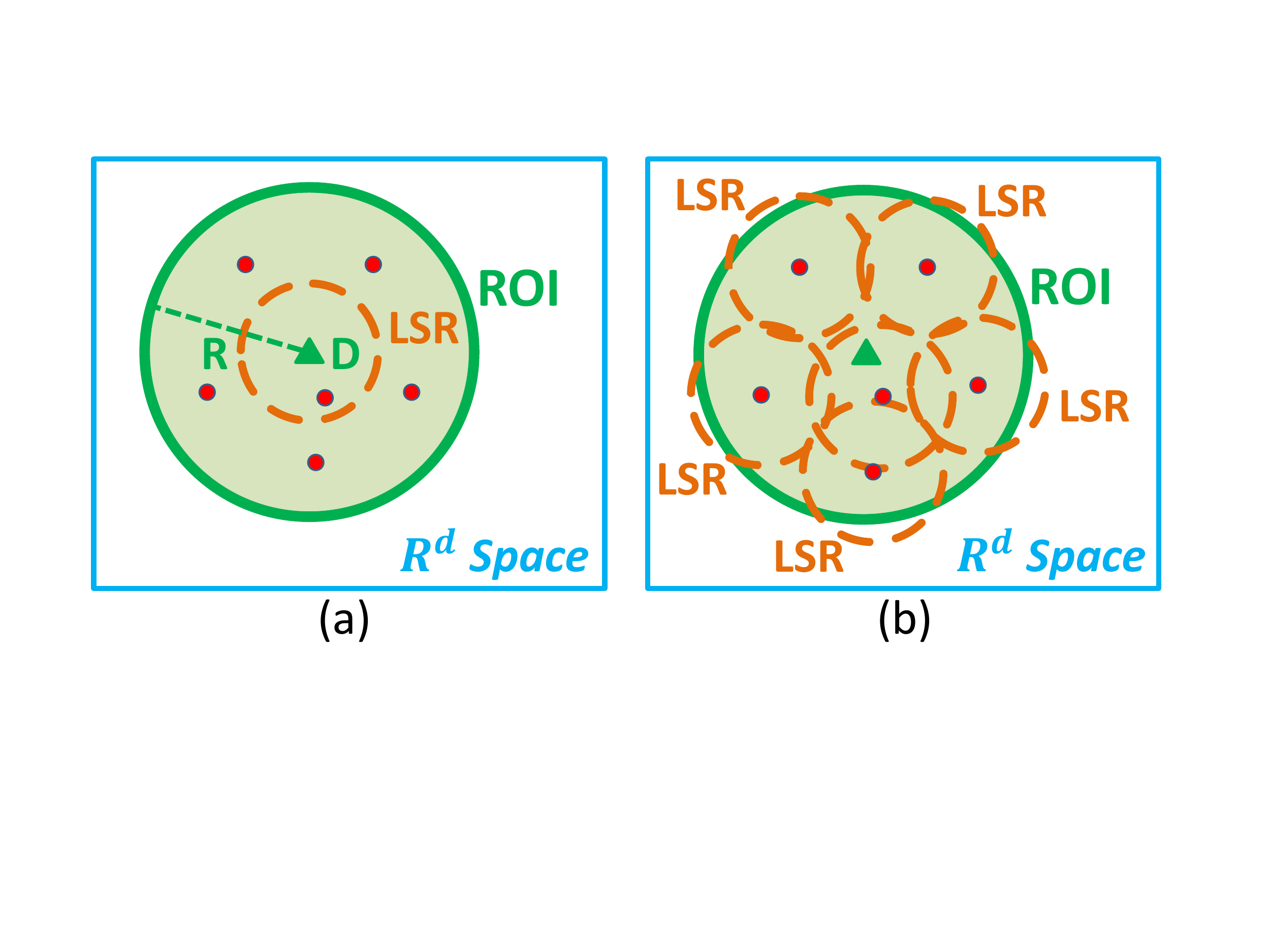}
\vspace{-3mm}
\caption{Each dashed circle is a locality sensitive region (LSR). The red points show the supporting data items of the current local dense subgraph $\hat{x}^{(c)}$. (a) A single LSR cannot cover the ROI. \hlit{(b) CIVS covers most of the ROI by multiple LSRs.}}
\label{Fig:hive}
\vspace{-3mm}
\end{figure}

\rmv{
\begin{enumerate}
  \item Find all the new data items that collide with any of the supporting data items of $\hat{x}^{(c)}$ by LSH.
  \item Retrieve at most $\delta$ new data items within the ROI that are the nearest to the ball center $D$. The index set of the retrieved data items is denoted by $\psi=\{i \mid i\in{(I-\alpha^{(c)})}, v_i\in H_c(D,R)\},|\psi|\leq{\delta}$.
\end{enumerate}
}

\hlit{
The retrieved data $\{v_i \mid i\in\psi\}$ are used to perform update: $\left[\hat{x}^{(c)},(A_{\beta\alpha}\hat{x}_\alpha)^{(c)}\right] \rightarrow \left[x^{(c+1)},(A_{\beta\alpha}x_\alpha)^{(c+1)}\right]$ as follow.
\begin{equation}
\label{Eqn:update_local_opt_result}
\left\{
\begin{array}{c}
  x^{(c+1)}=\hat{x}^{(c)}\\

  (A_{\beta\alpha}x_\alpha)^{(c+1)}=
  \begin{bmatrix}
  (A_{\alpha\alpha}\hat{x}_\alpha)^{(c)} \\
  (A_{\psi\alpha}\hat{x}_\alpha)^{(c)}
  \end{bmatrix}
\end{array}
\right.
\end{equation}
where $(A_{\psi\alpha}\hat{x}_\alpha)^{(c)}=A_{\psi\alpha}\hat{x}_\alpha^{(c)}$
}

\rmv{
\begin{enumerate}
  \item Compute submatrix $A_{\psi\alpha}$ and $(A_{\psi\alpha}\hat{x}_\alpha)^{(c)}=A_{\psi\alpha}\hat{x}_\alpha^{(c)}$.
  \item Update: $\left[\hat{x}^{(c)},(A_{\beta\alpha}\hat{x}_\alpha)^{(c)}\right] \rightarrow \left[x^{(c+1)},(A_{\beta\alpha}x_\alpha)^{(c+1)}\right]$ by
      \begin{equation}
      \label{Eqn:update_local_opt_result}
      \left\{
      \begin{array}{c}
          x^{(c+1)}=\hat{x}^{(c)}\\

          (A_{\beta\alpha}x_\alpha)^{(c+1)}=
          \begin{bmatrix}
          (A_{\alpha\alpha}\hat{x}_\alpha)^{(c)} \\
          (A_{\psi\alpha}\hat{x}_\alpha)^{(c)}
          \end{bmatrix}
      \end{array}
      \right.
      \end{equation}
      which updates the local range by $\beta^{(c+1)}{=}\alpha^{(c)} \cup \psi$, where $\psi$ involves new infective vertexes against $x^{(c+1)}$.
\end{enumerate}
}

\hlit{
Equation~\ref{Eqn:update_local_opt_result} updates the local range by $\beta^{(c+1)}{=}\alpha^{(c)} \cup \psi$, where $\psi$ involves new infective vertexes against $x^{(c+1)}$.
Then, we can re-run LID (i.e., Step 1) with the initialization of $[x^{(c+1)},(A_{\beta\alpha}x_\alpha)^{(c+1)}]$ to find the local dense subgraph $\hat{x}^{(c+1)}$ in the new range $\beta^{(c+1)}$.
Since $\hat{x}^{(c+1)}$ is guaranteed by {Theorem~\ref{theorem:nash_gamma_equal}} to be immune against all vertexes in $\psi\subset{U}$, the number of infective vertexes in the global range $I$ is further reduced.
}

The time and space complexity for building the hash tables are linear with respect to $n$.  Specifically, the time complexity to build $l$ hash tables by $\mu$ hash functions is $\mathcal{O}(ndl\mu)$. The space complexity consists of $\mathcal{O}(nd)$ space for all the $d$ dimensional data items, $\mathcal{O}(nl)$ space for an inverted list that maps each data item to their buckets and $\mathcal{O}(nl)$ space for $l$ hash tables~\cite{LSH}. Since all possible LSH queries are built into the hash tables, we check the inverted list to retrieve neighbor data items and do not store the hash keys.

\subsection{Summarization of ALID}
\label{Section:Summarization_of_ALID}
The entire iteration of {\ALID} is summarized in {Algorithm~\ref{alg:ALID_final}}. The LID in Step 1 makes the local dense subgraph immune against all vertexes within a local range of the ROI. The ROI and CIVS in Step 2 and Step 3 update the local range by the new infective vertexes retrieved from global range. In this way, the number of infective vertexes in global range is guaranteed to be iteratively reduced to zero. Then, according to {Theorem~\ref{theorem:nash_gamma_equal}}, the last found local dense subgraph is a global one that identifies a dominant cluster~\cite{Weibull}. {Algorithm~\ref{alg:ALID_final}} stops when a global dense subgraph is found or the total number of iterations exceeds an upper limit $C$. Since Algorithm~\ref{alg:ALID_final} is initialized with $A_{\beta\alpha}x_\alpha=0$, which cannot be used to compute the radius of ROI $H_{c=1}(D,R)$ (Equation~\ref{Eqn:define_hyperball} and ~\ref{Eqn:ROI_final_def}), thus we set $R=0.4$ for the first iteration $c=1$.

In order to fairly compare with the other affinity-based methods, {\ALID} adopts the same \emph{peeling method} as DS~\cite{DS} and IID~\cite{IID} do to detect all dominant clusters. The \emph{peeling method} peels off the detected cluster and reiterates on the remaining data items to find another one until all data items are peeled off. Then, the clusters with large values of $\pi(x)$ (e.g., $\pi(x)\geq 0.75$) are selected as the final result.

\begin{algorithm}[t]
\label{alg:ALID_final}
\caption{The entire {\ALID} iteration}
\SetKwRepeat{doWhile}{do}{while}
\KwIn{An initial vertex index $i\in{I}$}
\KwOut{A global dense subgraph $x^*$ in global range $I$}
\BlankLine

\begin{algorithmic}[1]

\STATE \textbf{Set}
$\alpha=\beta=i$, $x=s_i$, $A_{\beta\alpha}x_\alpha=a_{ii}=0$, $c=1$

\REPEAT
    \STATE \textbf{Step 1}: $\left[x^{(c)},(A_{\beta\alpha}x_\alpha)^{(c)}\right] \rightarrow \left[\hat{x}^{(c)},(A_{\beta\alpha}\hat{x}_\alpha)^{(c)}\right]$ \\
    Find the local dense subgraph $\hat{x}^{(c)}$ by LID in Step 1 

    \STATE \textbf{Step 2}: $\hat{x}^{(c)} \rightarrow H_c(D,R)$ \\
    Estimate ROI from the local dense subgraph $\hat{x}^{(c)}$ 

    \STATE \textbf{Step 3}: $\left[\hat{x}^{(c)},(A_{\beta\alpha}\hat{x}_\alpha)^{(c)}\right] \rightarrow \left[x^{(c+1)},(A_{\beta\alpha}x_\alpha)^{(c+1)}\right]$ \\
     Apply CIVS to retrieve candidate vertexes within the ROI and update the local range $\beta^{(c)}$ by Equation~\ref{Eqn:update_local_opt_result} for the next iteration 

    \STATE Index update: $c\leftarrow{c+1}$

\UNTIL{$\hat{x}^{(c)}$ is a global dense subgraph, or $c>C$}

\RETURN $x^*=\hat{x}^{(c)}$
\end{algorithmic}
\end{algorithm}

\subsection{Complexity Analysis}
\label{Section:Complexity_Analysis}
The time and space complexities of {\ALID} mainly consist of three parts:
\begin{description}
\item 1) The time and space complexities of LID in Step 1 are $\mathcal{O}(Tb)$ and $\mathcal{O}(b)$, respectively, where $b=|\beta|<n$ is the size of the local range $\beta$ and $T$ is the constant limit of the number of LID iterations.
    \vspace{-1mm}
\item 2) The time and space complexities for the hash tables of CIVS are $\mathcal{O}(ndl\mu)$ and $\mathcal{O}(n(2l+d))$, respectively, where $d,l,\mu$ are constant LSH parameters.
    \vspace{-1mm}
\item 3) The time and space complexities for the affinity matrix $A$ are $\mathcal{O}(C(a^*+\delta)n)$ and $\mathcal{O}(a^*(a^*+\delta))$, respectively, which are analyzed in detail as follows.
\end{description}

Since all {\ALID} iterations are restricted by $A_{\beta\alpha}$, the time and space complexities for the affinity matrix are determined by the size of $A_{\beta\alpha}$.
Let $(A_{\beta\alpha})^c_i$ be the submatrix computed in the $c$-th iteration in {Algorithm~\ref{alg:ALID_final}} when detecting the $i$-th cluster.  Denote by $a^c_i$ and $b^c_i$, respectively, the column and row sizes of $(A_{\beta\alpha})^c_i$.
Since the maximum number of iterations of Algorithm~\ref{alg:ALID_final} is $C$, the overall time cost for detecting the $i$-th cluster is $Time(i)<\sum^C_{c=1}{a^c_ib^c_i}$, which is a \emph{loose bound} since many matrix entries of different $(A_{\beta\alpha})^c_i$ are duplicate and only computed once. Then, we can derive
\begin{equation}
\label{Eqn:time_i}
Time(i)<{Ca^C_i(a^C_i+\delta)}
\end{equation}
from the following observations:
\begin{itemize}
  \item $a^c_i\leq{a^C_i}$, since more and more matrix columns (i.e., $A_{\beta i}:i\in(\beta-\alpha)$ in Figure~\ref{Fig:LID_matrix}) are computed.
  \item $b^c_i\leq(a^c_i+\delta)\leq{(a^C_i+\delta)}$, since the size of $\beta$ is strictly limited by the ROI, where at most $\delta$ data items can be retrieved by CIVS.
\end{itemize}

\hlit{
Since $Time(i)$ is the time cost of affinity matrix computation for detecting the $i$-th dominant cluster, then the overall cost in time of detecting all dominant clusters is $\sum_i Time(i)$. We can derive from Equation~\ref{Eqn:time_i} that
\begin{equation}
\label{Eqn:time_i_sum}
\sum_i Time(i)<\sum_i {Ca^C_i(a^C_i+\delta)}
\end{equation}

Recall that {\ALID} adopts the \emph{peeling method} (see Section~\ref{Section:Summarization_of_ALID}), which peels off one detected cluster and reiterates on the remaining data items to find another one until all the $n$ data items are peeled off. Therefore, we can derive
\begin{equation}
\label{Eqn:sum_ac_i}
\sum_i{a^C_i}=n
\end{equation}
from the fact that $a^C_i$ is the size of the $i$-th detected cluster.

Define $a^*=\max_i\{a^C_i\}$ and $b^*=\max_i\{b^C_i\}$, where $b^*\leq(a^*+\delta)$ due to the restriction of ROI.
Then, we can derive from Equation~\ref{Eqn:time_i_sum} and Equation~\ref{Eqn:sum_ac_i} that the overall cost in time of computing the affinity matrix is
\begin{equation}
\label{Eqn:time_i_final}
\sum_i{Time(i)}<\sum_i{Ca^C_i(a^*+\delta)}=C(a^*+\delta)n
\end{equation}

The maximum cost in space is $a^*b^*\leq a^*(a^*+\delta)$, since all submatrices $(A_{\beta\alpha})^c_i$ are released when the $i$-th cluster is peeled off. As a result, the time and space complexities for the affinity matrix of {\ALID} are $\mathcal{O}(C(a^*+\delta)n)$ and $\mathcal{O}(a^*(a^*+\delta))$, respectively.
}

Recall that $a^*=\max_i\{a^C_i\}$ is the size of the largest (single) dominant cluster. We summarize in Table~\ref{Table:complexity_bounds} the three typical cases how $a^*$ affects the time and space complexities of the affinity matrix.
\hlit{
The data items belonging to the largest dominant cluster with size $a^*$ is referred to as ``positive data'', data items that do not belong to any dominant cluster are regarded as ``noise data'' and the size of the entire data set is denoted by $n$.
}

First, for clean data source, the amount of positive data is in constant proportion of the entire \hlit{data set}. Thus, we have $a^*=\omega{n}$ and $\omega\leq{1}$ is the constant proportion.
In this case, {\ALID} reduces the original $\mathcal{O}(n^2)$ time and space complexities of the affinity matrix to $\mathcal{O}(C(\omega n^2 + \delta n))$ and $\mathcal{O}(\omega^2n^2+\delta\omega n)$, respectively.

\begin{table}[t]
  \centering
  \caption{The complexity of the affinity matrix}
  \label{Table:complexity_bounds}
  \begin{tabular}{|c|c|c|}
    \hline
    Typical Cases & Time Complexity & Space Complexity  \\ \hline
    $a^*=\omega{n}\;(\omega\leq{1})$ & $\mathcal{O}(C(\omega n^2 + \delta n))$ & $\mathcal{O}(\omega^2n^2+\delta\omega n)$ \\ \hline
    $a^*=n^\eta\;(\eta<1)$ & $\mathcal{O}(C(n^{1+\eta}+\delta n))$ & $\mathcal{O}(n^{2\eta}+\delta n^\eta)$ \\ \hline
    $a^*\leq{P}$ & $\mathcal{O}(C((P+\delta)n))$ & $\mathcal{O}(P^2+\delta P)$ \\ \hline
  \end{tabular}
\end{table}

Second, for noisy data source (e.g., tweet-streams and user comments) that generates noise data faster than positive data, the growth rate of $a^*$ is slower than $n$. In this case, we have $a^*=n^\eta\;(\eta<1)$, thus {\ALID} reduces the $\mathcal{O}(n^2)$ time and space complexities to $\mathcal{O}(C(n^{1+\eta}+\delta n))$ and $\mathcal{O}(n^{2\eta}+\delta n^\eta)$, respectively.


\hlit{
Third, for noisy data source with size-limited dominant clusters, there is a constant upper bound $P$ for $a^*$ (i.e., $a^*\leq P$). Typical data sources of such kind consist of phone books and email contacts, where the largest number of people in a stable social group (i.e., dominant cluster) is limited by the \emph{Dunbar's number}
\footnote{\emph{Dunbar's number} was proposed by Robin Dunbar, who found that the social group size is upper bounded due to the limited brain size of human. Therefore, the size of a single dominate cluster (i.e., social group) in a social data source is upper bounded by a constant. \url{http://en.wikipedia.org/wiki/Dunbar's_number}} ~\cite{Dunbar_number}.
In such a case, we have $a^*\leq P$, thus, the time and space complexities of the affinity matrix are reduced to $\mathcal{O}(C(P+\delta)n)$ and $\mathcal{O}(P^2+\delta P)$, respectively.
Note that, since the infection immunization process converges quickly in finite steps~\cite{Bulo50}, a small value of $C=10$ is adequate for the convergence of {\ALID}.
}

\hlit{
We conducted extensive experiments on three synthetic data sets to simulate the three typical cases of $a^*=\omega n$, $a^*=n^\eta$ and $a^*\leq{P}$. The experimental results in Figure~\ref{Fig:SA_SYN_NDI} of Section~\ref{Section:SA} are consistent with the time and space complexities summarized in Table~\ref{Table:complexity_bounds}.
}

\rmv{
In summary, the overall time complexity of {\ALID} under the three typical cases of $a^*$ are: $\mathcal{O}(n^2)$ when $a^*=\omega n$, $\mathcal{O}(n^{1+\eta})$ when $a^*=n^\eta$ and $\mathcal{O}(n)$ when $a^*\leq{P}$. The overall space complexity are: $\mathcal{O}(n^2)$ when $a^*=\omega n$, $\mathcal{O}(n^{2\eta})$ when $a^*=n^\eta$ and $\mathcal{O}(n)$ when $a^*\leq{P}$. Comparing with the other affinity-based methods,  {\ALID} is more efficient on clean data and is much more scalable on noisy data.
}

\subsection{Parallel ALID}
{\ALID} is very suitable for parallelization in the MapReduce framework~\cite{MapReduce}, since multiple tasks of {\ALID} can be concurrently run in independent local ranges of the affinity graph.
We introduce the parallel ALID ({\PALID}) in {Algorithm}~\ref{alg:ALID_parallel} and provide an illustrative example in Figure~\ref{Fig:MapReduce}.

\begin{figure}[h]
\hspace{-2mm}
\includegraphics[width=85mm]{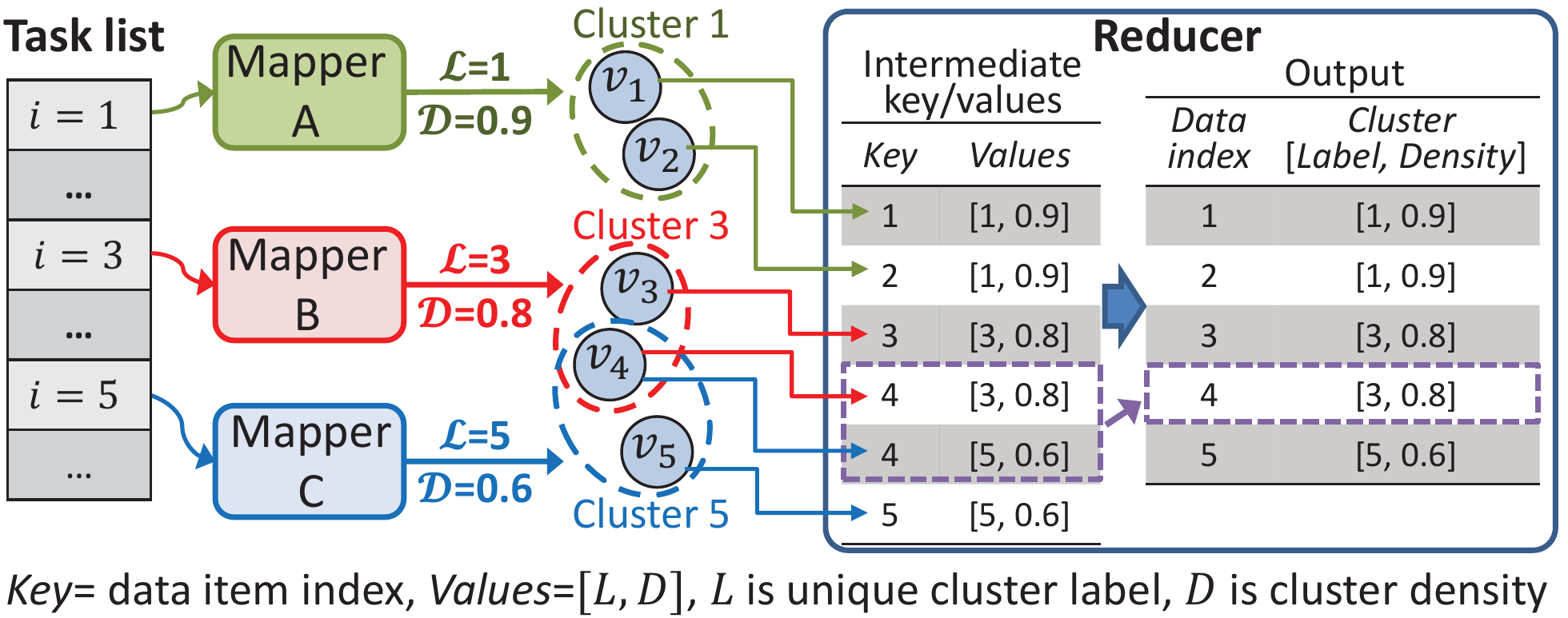}
\vspace{-4mm}
\caption{An illustrative example of {Algorithm}~\ref{alg:ALID_parallel}. Each mapper runs {Algorithm}~\ref{alg:ALID_final} independently with a different initial vertex. The reducer assigns each data item to the cluster with maximum density.}
\label{Fig:MapReduce}
\end{figure}

As shown in Figure~\ref{Fig:MapReduce}, three initial graph vertex indexes $i=\{1,3,5\}$ are assigned to three Mappers (A,B,C). Each Mapper runs {Algorithm}~\ref{alg:ALID_final} to detect a cluster independently. Once a Mapper detects a cluster, it produces a list of intermediate key/value pairs $(Key,Values{=}[\mathcal{L}, \mathcal{D}])$. \emph{Key} is the index of one data item belonging to the cluster, $\mathcal{L}$ is the unique cluster label for the detected cluster and $\mathcal{D}$ is the density of the cluster.
In case of overlapping clusters, such as clusters 3 and 5 in the figure, we simply assign the overlapped data item $v_4$ to cluster 3 of maximum density, which can be easily handled by a reducer.

\hlit{
Since data items belonging to the same dominant cluster are highly similar with each other, such data items are likely to be mapped to the same set of LSH buckets. Therefore, large-sized buckets reveal the potential data items of dominant clusters.
As a result, {\PALID} uniformly samples the initial graph vertexes from every LSH hash bucket that contains more than 5 data items. The sample rate is 20\%.
}

The hash tables and data items are stored in a server database and accessed via the network. The communication overhead is low due to the following reasons: 1) the LSH queries only transport data item indexes, which consumes ignorable bandwidth; 2) each mapper only needs to access a few data items to detect one dominant cluster due to the local property of {\ALID}.

In summary, {\ALID} is highly parallelizeable in the MapReduce framework, which further improves its capability in handling massive data in real world applications.

\begin{algorithm}[t]
\label{alg:ALID_parallel}
\caption{The parallel ALID ({\PALID})}
\SetKwRepeat{doWhile}{do}{while}
\KwIn{$V=\{v_i \mid i\in{I=[1,n]}\}$}
\KwOut{Cluster labels and cluster densities of each data item in the detected clusters}
\textbf{Tasklist}: A list of initial graph vertex indexes


\textbf{Map}(\emph{Key}, \emph{Value})\\
\small{
$\backslash\backslash$\emph{Key}: An initial vertex index $i$ for {Algorithm}~\ref{alg:ALID_final} \\
$\backslash\backslash$\emph{Value}: A unique cluster label $\mathcal{L}$ for the detected cluster
}\\[0.5mm]
\normalsize
\begin{algorithmic}[1]
\STATE Call {Algorithm}~\ref{alg:ALID_final} to find a global dense subgraph $x^*$, which identifies dominant cluster $\mathcal{L}$ with density $\mathcal{D}=\pi(x^*)$
\FORALL {data item index $h$ in $I=[1,n]$}
    \IF {$x_{h}^*>0$}
        \STATE Emit($h$, $[\mathcal{L}, \mathcal{D}]$) $\backslash\backslash$ $v_h$ belongs to cluster $\mathcal{L}$
    \ENDIF
\ENDFOR
\end{algorithmic}


\textbf{Reduce}(\emph{Key}, \emph{Values}) \\
\small{
$\backslash\backslash$\emph{Key}: The index $h$ of a single data item $v_h\in{V}$ \\
$\backslash\backslash$\emph{Values}: A list of $[\mathcal{L}, \mathcal{D}]$ w.r.t. the clusters containing $v_h$
}\\[0.5mm]
\normalsize
\begin{algorithmic}[1]
\STATE Find $[\mathcal{L}^*,\mathcal{D}^*]$ in \emph{Values} with maximum density $\mathcal{D}^*$
\STATE Emit($h$, [$\mathcal{L}^*,\mathcal{D}^*]$) $\backslash\backslash$ Assign $v_h$ to cluster $\mathcal{L}^*$
\end{algorithmic}
\end{algorithm}


\section{Experimental Results}\label{sec:exp}

\hlit{
In this section, we empirically examine and analyze the performances of {\ALID} and {\PALID}. The following state-of-the-art affinity-based methods are analyzed as well:
1) Affinity Propagation (AP)~\cite{AP}; 2) the Shrinking Expansion Algorithm (SEA)~\cite{SEA}; 3) Infection Immunization Dynamics (IID)~\cite{IID}. We use the published source codes of AP \cite{AP_code} and SEA \cite{SEA_code}; as the code for IID is unavailable, we implemented it in MATLAB.
\hlit{All compared methods are carefully tuned to their best performances. The parameter $\delta$ in Step 3 of {\ALID} ({\PALID}) is fixed as $\delta=800$.}
}

The detection quality is evaluated by the Average $\text{F}_1$ score (AVG-F), which is the same criterion as Chen~\textit{et~al.}~\cite{Fscore_cite} used.  AVG-F is obtained by averaging the $\text{F}_1$ scores on all the true dominant clusters. \rmv{The F-score is defined as
\begin{equation}
\label{Eqn:F_score}
\text{F-score} = \frac{Precision*Recall}{Precision + Recall}
\end{equation}
}
A higher $\text{F}_1$ score means a smaller deviation between the detected and the true dominant clusters, thus indicates a better detection quality. Besides, as Chen~\textit{et~al.}~\cite{Fscore_cite} showed, since the data items are partially clustered in this task, traditional evaluation criteria, such as entropy and normalized mutual information, are not appropriate in evaluating the detection quality.

The detection efficiency is measured from two perspectives: 1) the runtime of each method including the time to compute the affinity matrix; 2) the memory overhead, including the memory to store the affinity matrix.

We use a PC computer with a Core i-5 CPU, 12 GB main memory and a 7200 RPM hard drive, running Microsoft Windows 7 operating system.
All experiments using single computer are conducted using \emph{MATLAB}, since the methods to be compared were implemented on the same platform. We will explicitly illustrate the parallel experiments of {\PALID} in Section~\ref{Section:PALID}. {\PALID} was implemented in Java on \emph{Apache Spark} to efficiently process up to 50 million data items.

The following data sets are used: 1) the news articles data set (NART); 2) the near duplicate image data set (NDI); 3) three synthetic data sets; 4) the SIFT-50M data set that consists of 50 million SIFT features \cite{SIFT}. Details about the three synthetic data sets and the SIFT-50M data set are illustrated in Section~\ref{Section:SA} and Section~\ref{Section:PALID}, respectively. Details of NART and NDI are illustrated as follow. For all data sets, the pairwise distance and affinity are calculated using Euclidean distance and Equation~\ref{Eqn:exp_kernel_affinity} ($p=2$), respectively.

\emph{The news articles data set} (NART) is built by crawling 5,301 news articles from \url{news.sina.com.cn}.
It contains 13 real world ``hot'' events happened from May 2012 to June 2012, each of which corresponds to a dominant cluster of news articles.
\hlit{
All 734 news articles of the 13 dominant clusters are manually labeled as ground truth by 3 volunteers without professional background.
}
The remaining 4,567 articles are daily news that do not form any dominant cluster. Each article is represented by a normalized 350-dimensional vector generated by standard Latent Dirichlet Allocation (LDA)~\cite{LDA}.

\emph{The near duplicate image data set} (NDI) contains 109,815 images crawled from \url{images.google.com.hk}. It includes a labeled set of ground truth of 57 dominant clusters and 11,951 near duplicate images, where images with similar contents are grouped
as one dominant cluster. The remaining 97,864 images with diverse contents are regarded as background noise data. Each image is represented by a 256-dimensional GIST feature~\cite{GIST} that describes the global texture of the image content.

\subsection{Sparsity Influence Analysis}
\label{Section:SIA}
As mentioned in Section~\ref{sec:related-work}, the scalability of canonical affinity-based methods (i.e., IID, SEA, AP) is limited by the $\mathcal{O}(n^2)$ time and space complexity to fully compute and store the affinity matrix. Although the computational efficiency can be improved by sparsifying the affinity matrix~\cite{PSC}, the enforced sparsity also breaks the high cohesiveness of dense subgraphs, which inevitably weakens the noise resistance capability and impairs the detection quality.

In this section, we specifically analyze how the sparse degree of sparsified affinity matrix affects the detection quality and runtime of all compared methods.
The sparse degree is defined as the ratio of the number of entries in the matrix taking value $0$ over the total number of entries in the matrix.

\hlit{
All experiment results are obtained on NART and Sub-NDI data sets. Sub-NDI is a subset of the NDI data set, it contains 6 clusters of 1420 ground truth images and 8520 noise images. We use Sub-NDI instead of NDI, since AP cannot deal with the entire NDI data set with 12GB RAM.
}

\hlit{
Chen~\textit{et~al.}~\cite{PSC} provided two approaches to sparsify the affinity matrix: the exact nearest neighbors (ENN) method and the approximate nearest neighbors (ANN) method. The ENN method is  expensive on large data sets, while the ANN method can be efficient by employing LSH~\cite{LSH} and Spill-Tree~\cite{spill_tree}. In our experiments, we sparsify the affinity matrix by LSH due to its efficiency.
}

For AP, IID and SEA, we directly apply LSH to sparsify the affinity matrix, where only the affinities between the nearest neighbors are computed and stored. The same LSH module is utilized by CIVS in {\ALID}. To remove possible uncertainties caused by the LSH approximation, the parameter settings of LSH are kept exactly the same for all the compared methods, including {\ALID}.

Standard LSH projects each data item onto an equally segmented real line. The line segment length $r$ controls the recall of LSH, and thus affects the sparse degree of the affinity matrix. Figure~\ref{Fig:SIA_COMP} shows that the sparse degree of all affinity-based methods, including {\ALID}, decreases when $r$ increases. However, the sparse degree of {\ALID} remains high, since {\ALID} only computes small submatrices corresponding to the vertexes within the ROI.

\begin{figure}[t]
\centering
\mbox{
\hspace{-7mm}
\subfigure[NART]{\includegraphics[width=45mm]{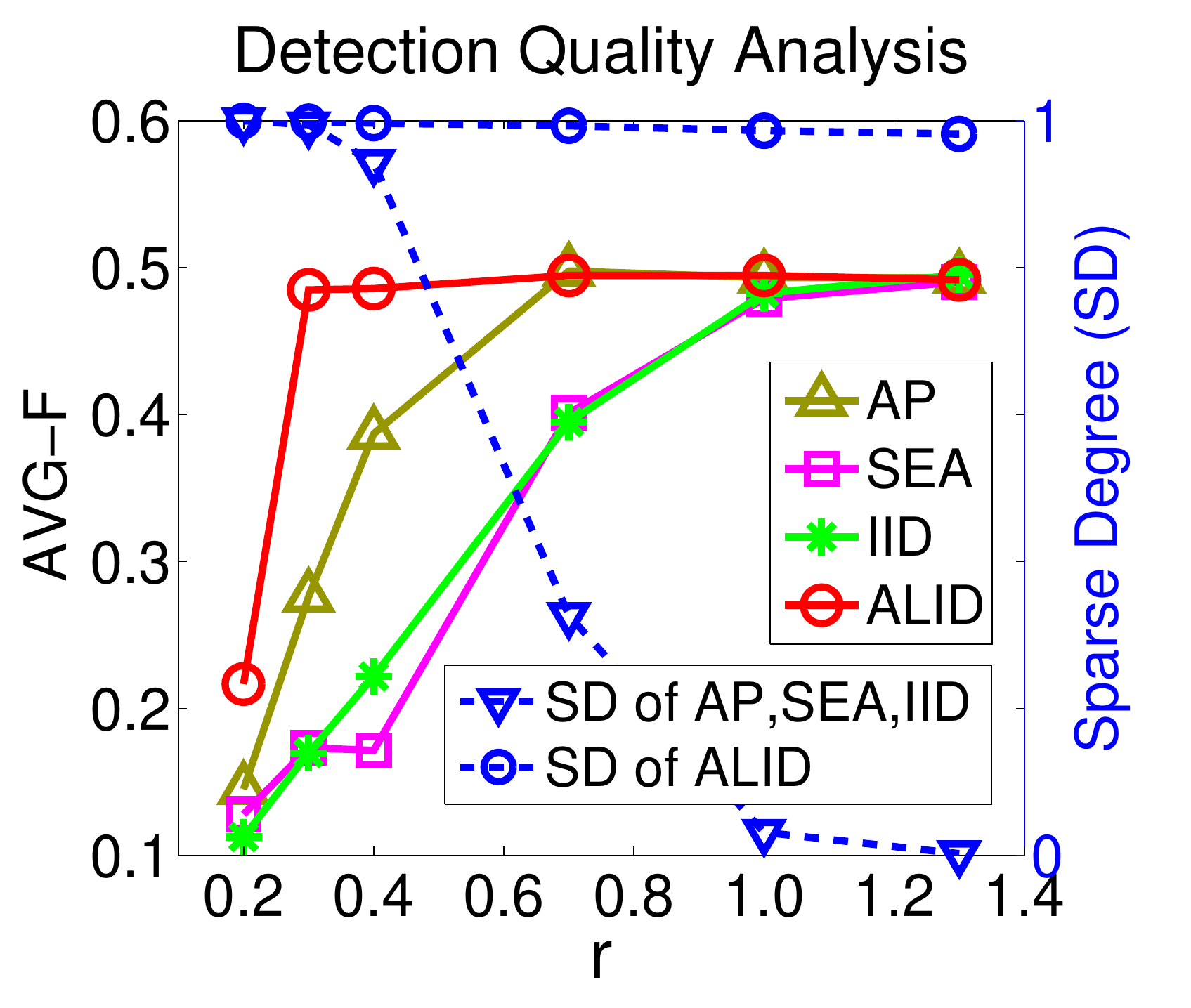}}
\subfigure[Sub-NDI]{\includegraphics[width=45mm]{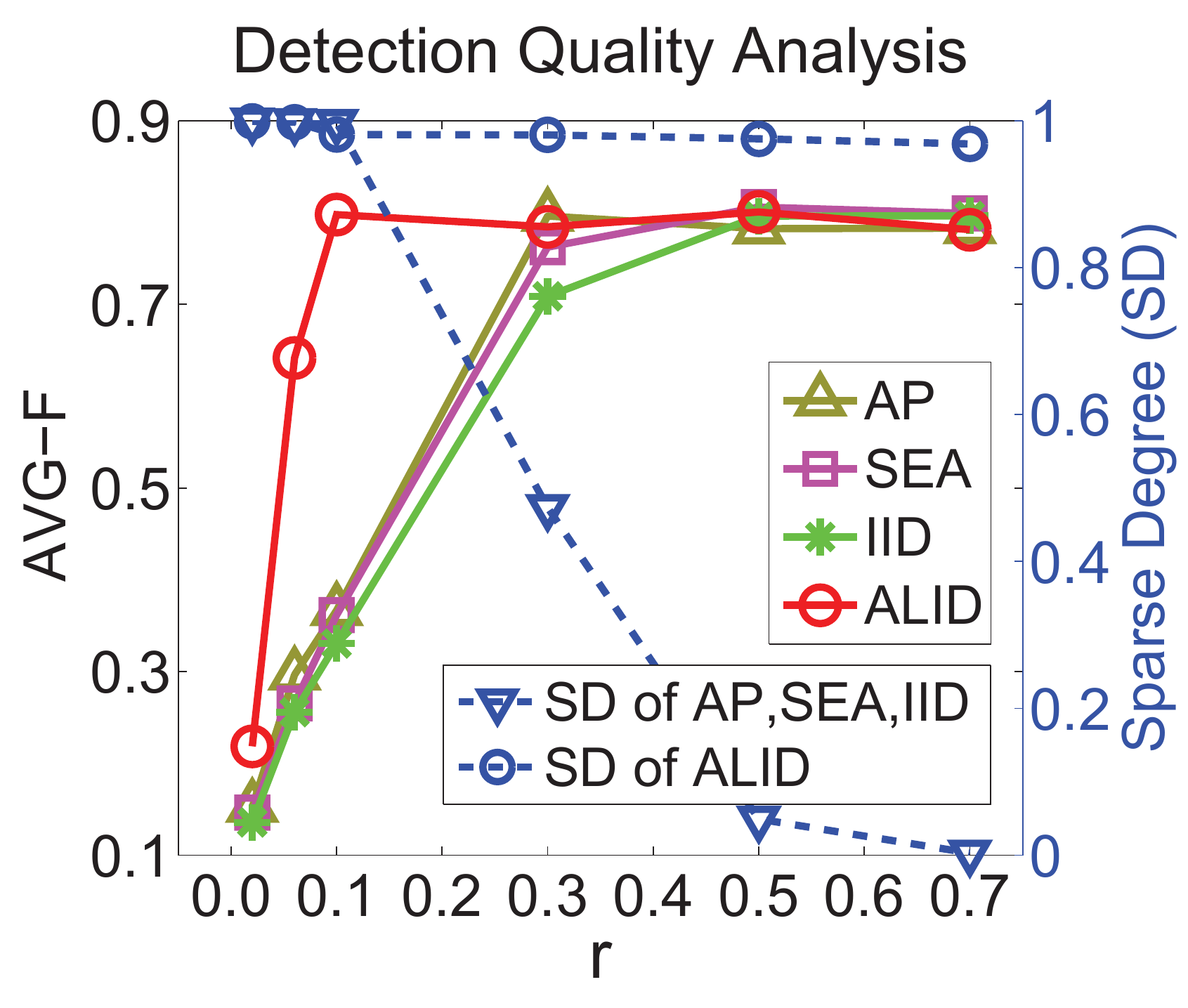}}
}
\mbox{
\hspace{-7mm}
\subfigure[NART]{\includegraphics[width=45mm]{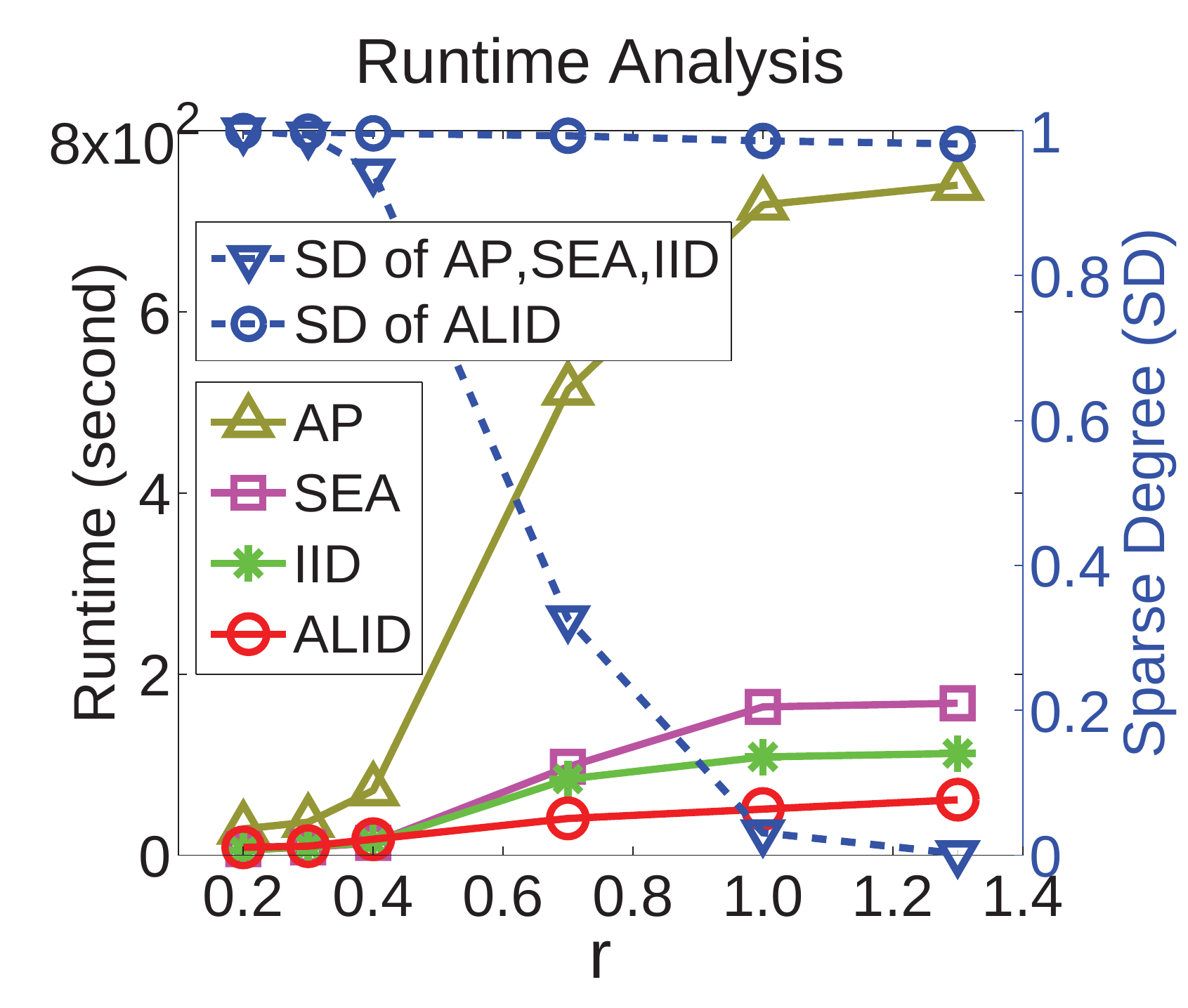}}
\subfigure[Sub-NDI]{\includegraphics[width=45mm]{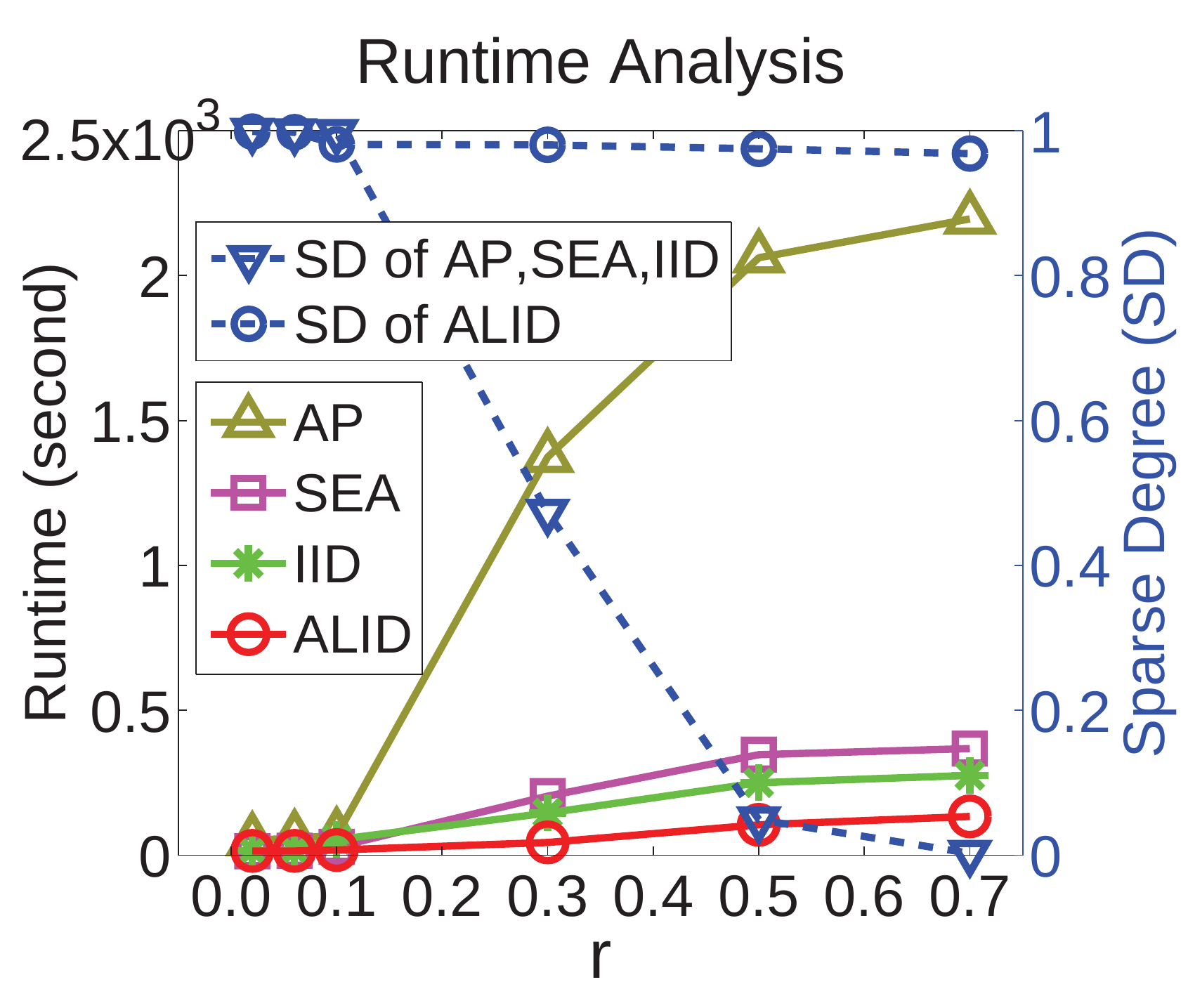}}
}
\vspace{-4mm}
\caption{The results on NART and Sub-NDI. (a)-(b) How sparse degree affects AVG-F. (c)-(d) How sparse degree affects runtime. For LSH, we use 40 projections per hash value and 50 hash tables. $r$ is the length of the equally divided segments of LSH.}
\vspace{-4mm}
\label{Fig:SIA_COMP}
\end{figure}

Figures~\ref{Fig:SIA_COMP}(a) and~\ref{Fig:SIA_COMP}(c) show the experiment results on data set NART. The AVG-F of all methods increases to a stable level as sparse degree decreases. This is because the cohesiveness of dense subgraphs are better retained as sparse degree decreases. For AP, SEA and IID, when sparse degree approaches zero, the original subgraph cohesiveness are maximally preserved by a full affinity matrix, and thus they all approach their best performances. Since most dense subgraphs exist in small local ranges, the relative local submatrices are good enough to retain their cohesiveness.
{\ALID} largely preserves such cohesiveness by accurately estimating the local range of true dense subgraphs and fully computing the relative local affinity matrices.
Consequently, {\ALID} achieves a good AVG-F performance under an extremely high sparse degree of $0.998$ ($r=0.3$), which indicates that the calculation and storage of $99.8\%$ matrix entries are effectively pruned. Such situation is rationale, since the useful matrix entries that correspond to the 13 true clusters of 734 data items in data set NART only take $734^2/(13 \times 5301^2)=0.147\%$ of the entire affinity matrix.

The results in Figure~\ref{Fig:SIA_COMP}(c) demonstrates that sparsifying the affinity matrix reduces the runtime of the affinity-based methods. The runtime of all methods are comparably low when sparse degree is high.  However, when sparse degree decreases, the differences in runtime among the methods become significant. When $r=1.3$, AP is significantly slower than the other methods due to its expensive message passing overheads.  Moreover, SEA is much slower than IID due to the time consuming replicator dynamics~\cite{Weibull}. {\ALID} is the fastest, since it effectively prunes the computation of $99.8\%$ affinity matrix entries. Similar results are also observed on data set Sub-NDI (Figures~\ref{Fig:SIA_COMP}(b) and~\ref{Fig:SIA_COMP}(d)).

\newcommand{\mywidth}{34mm}
\newcommand{\myhspace}{7mm}
\begin{figure*}[t]
\centering
\subfigure[Synthetic $a^*{=}\frac{\omega n}{20}$]{\includegraphics[width=\mywidth]{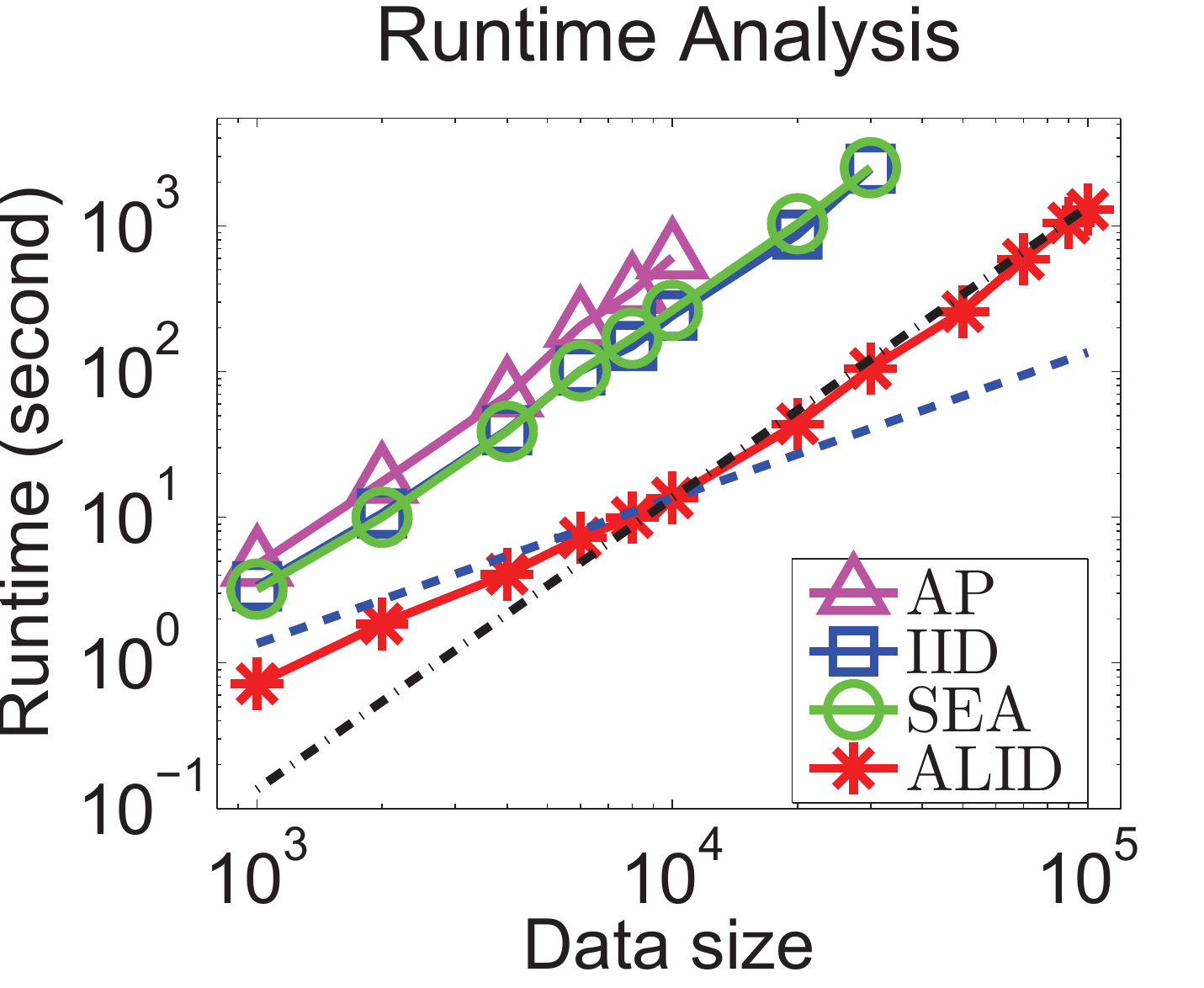}}
\hspace{\myhspace}
\subfigure[Synthetic $a^*{=}\frac{n^\eta}{20}$]{\includegraphics[width=\mywidth]{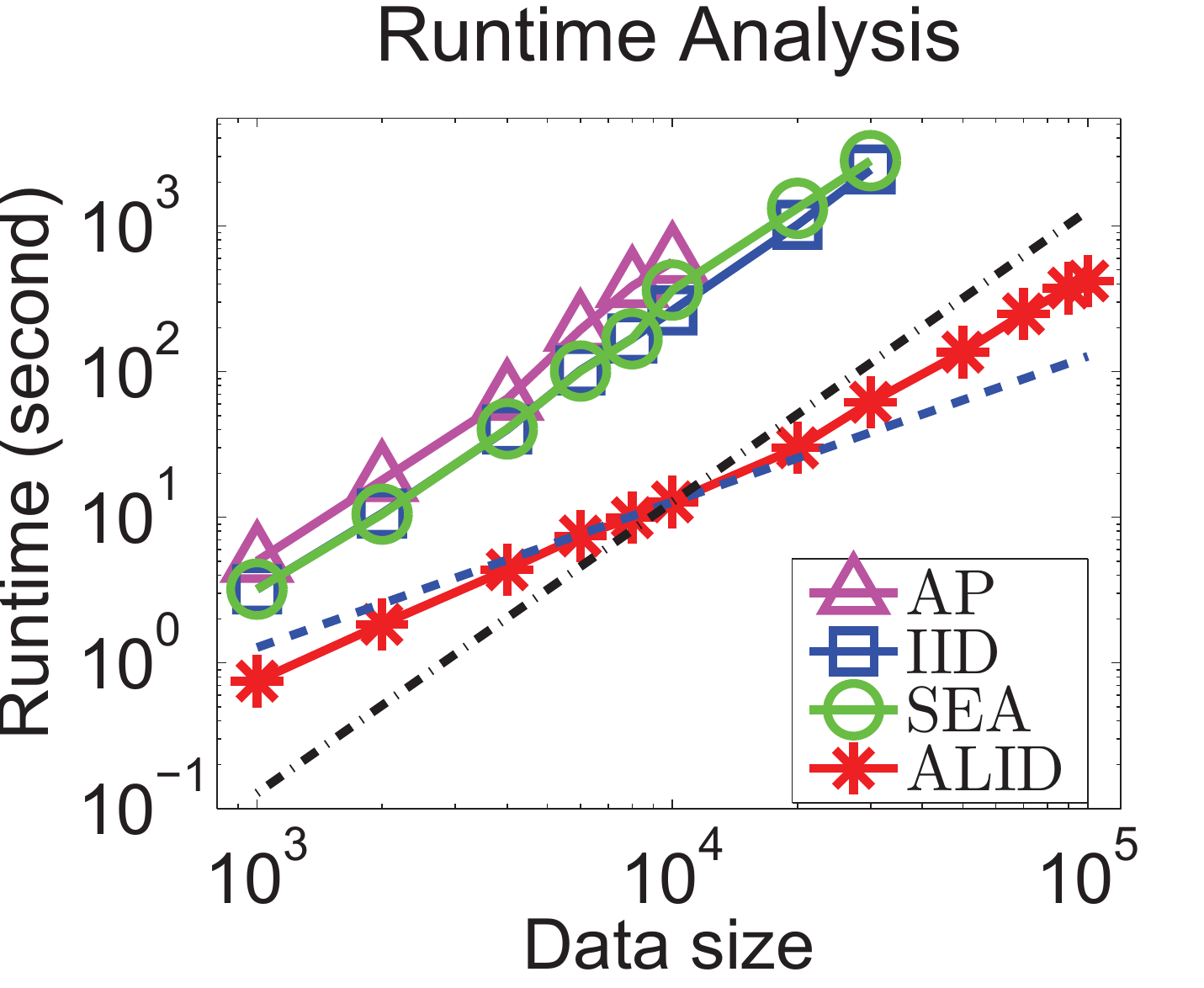}}
\hspace{\myhspace}
\subfigure[Synthetic $a^*{=}\frac{P}{20}$]{\includegraphics[width=\mywidth]{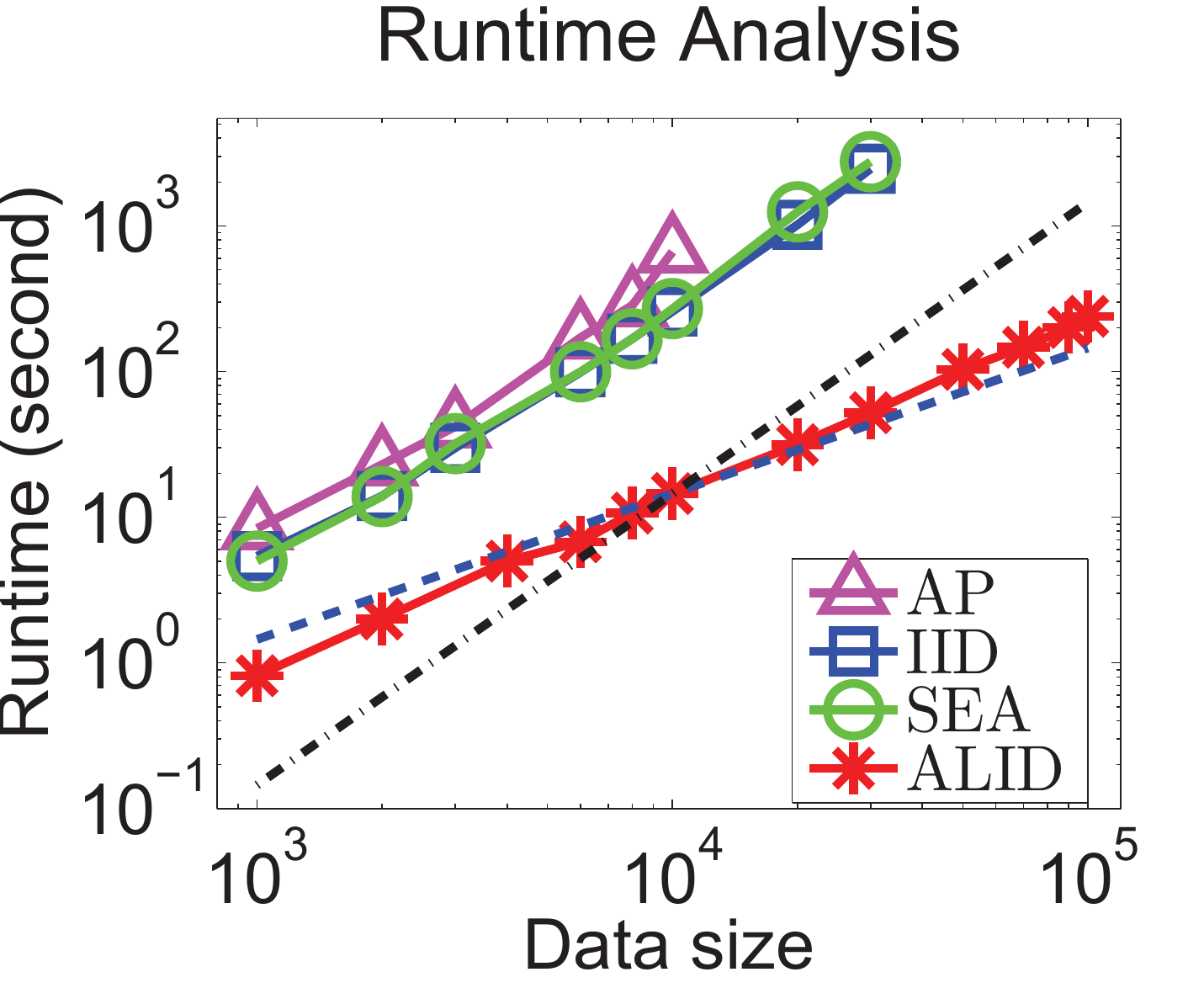}}
\hspace{\myhspace}
\subfigure[NDI data set]{\includegraphics[width=\mywidth]{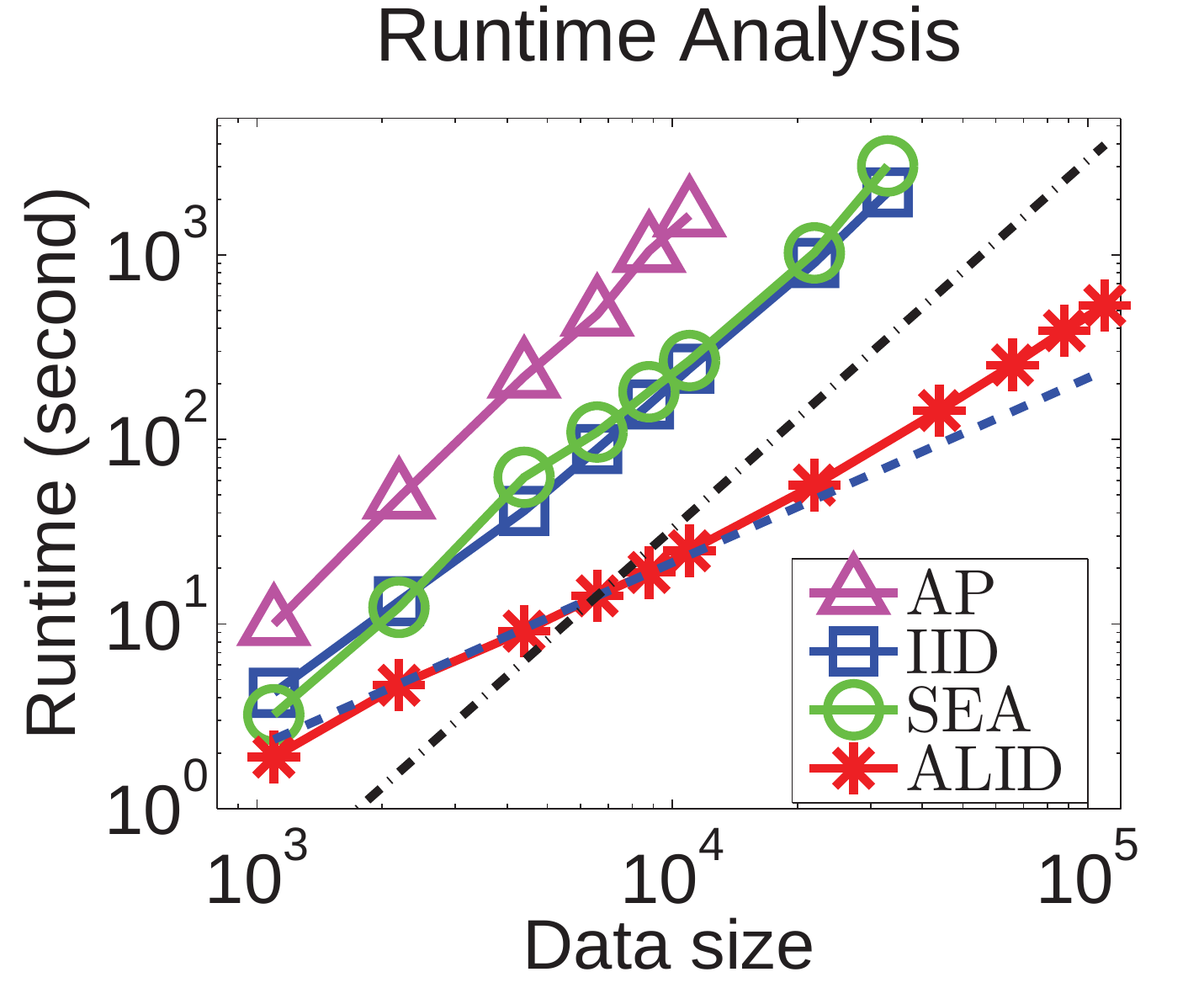}}
\subfigure[Synthetic $a^*{=}\frac{\omega n}{20}$]{\includegraphics[width=\mywidth]{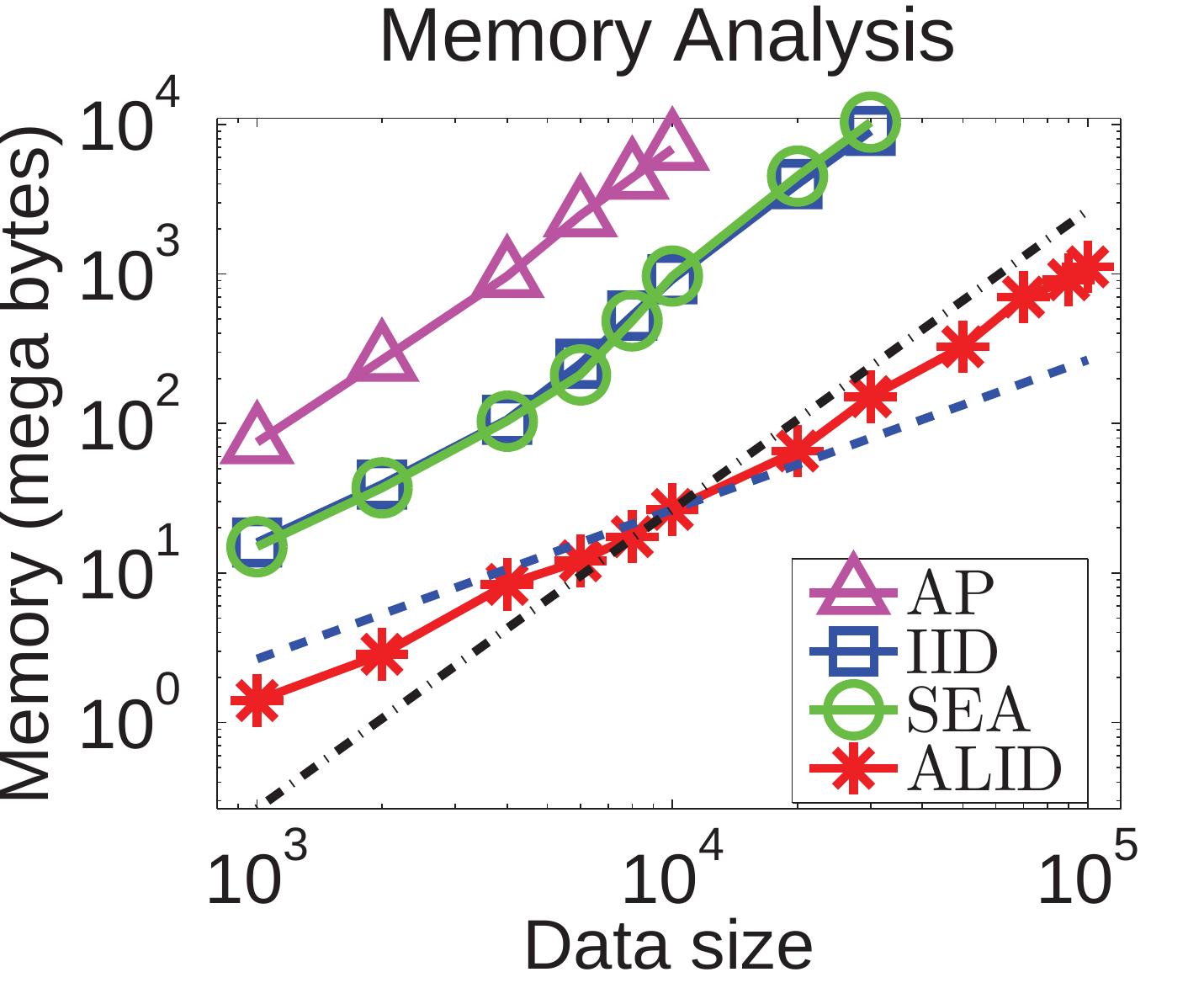}}
\hspace{\myhspace}
\subfigure[Synthetic $a^*{=}\frac{n^\eta}{20}$]{\includegraphics[width=\mywidth]{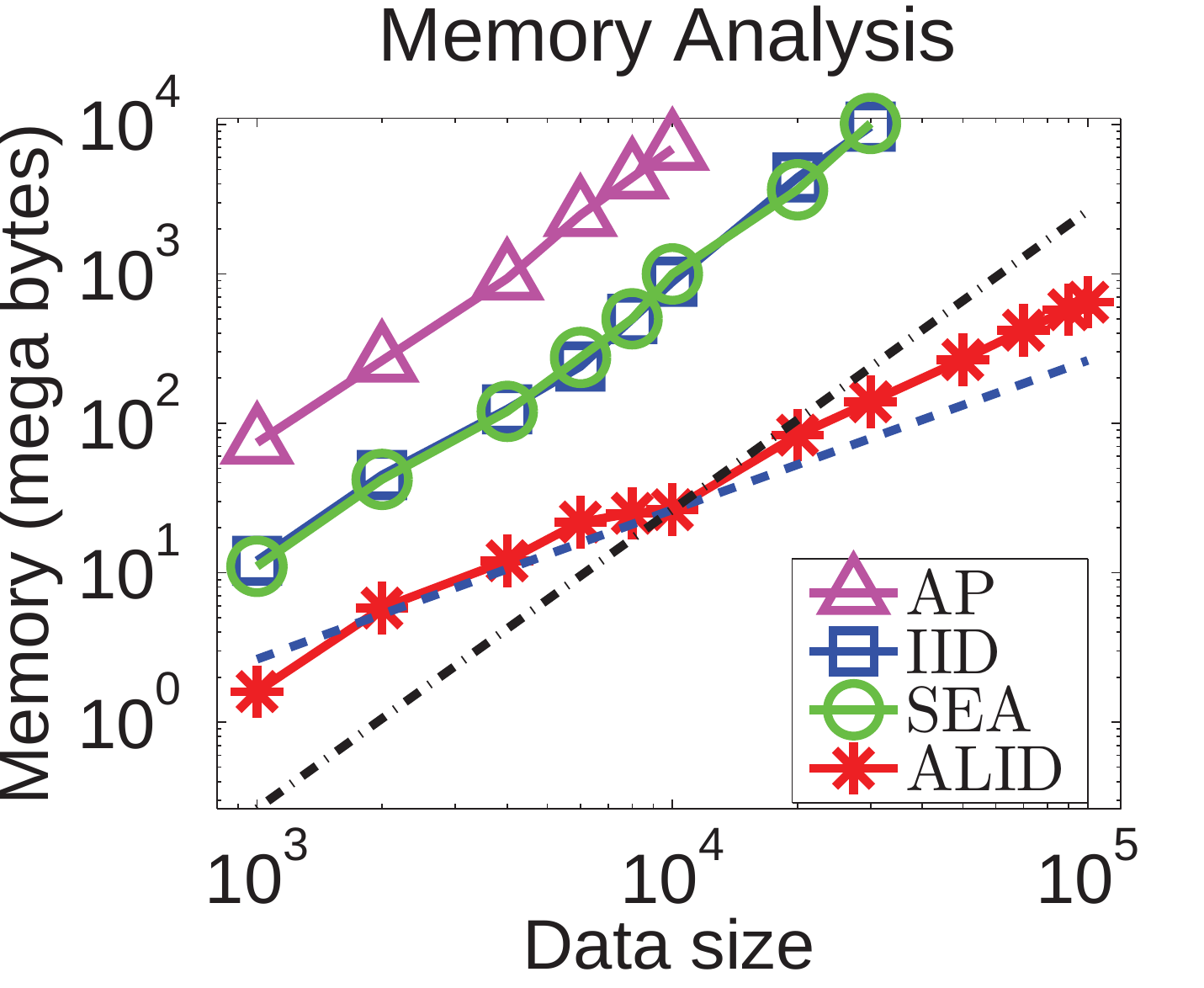}}
\hspace{\myhspace}
\subfigure[Synthetic $a^*{=}\frac{P}{20}$]{\includegraphics[width=\mywidth]{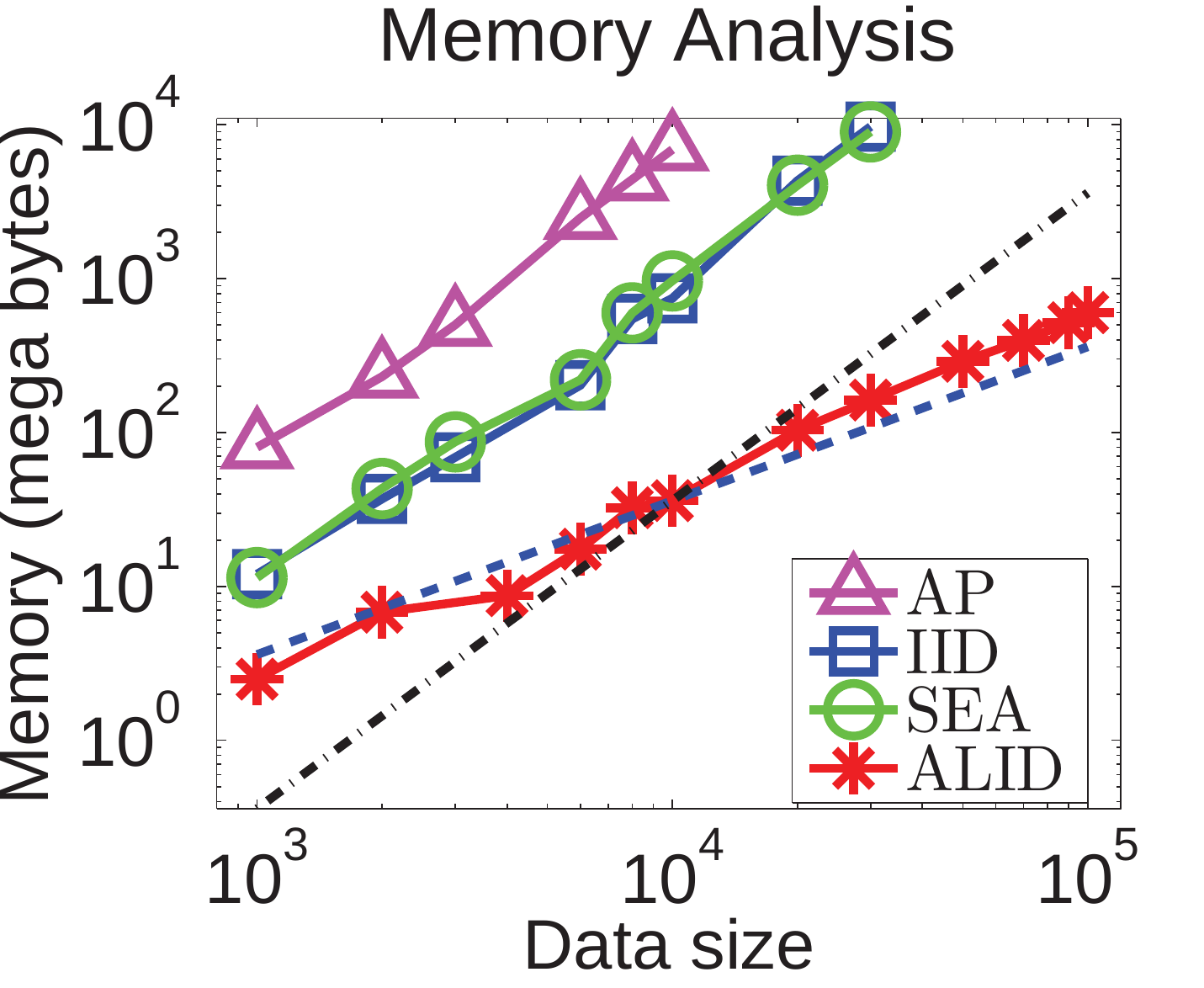}}
\hspace{\myhspace}
\subfigure[NDI data set]{\includegraphics[width=\mywidth]{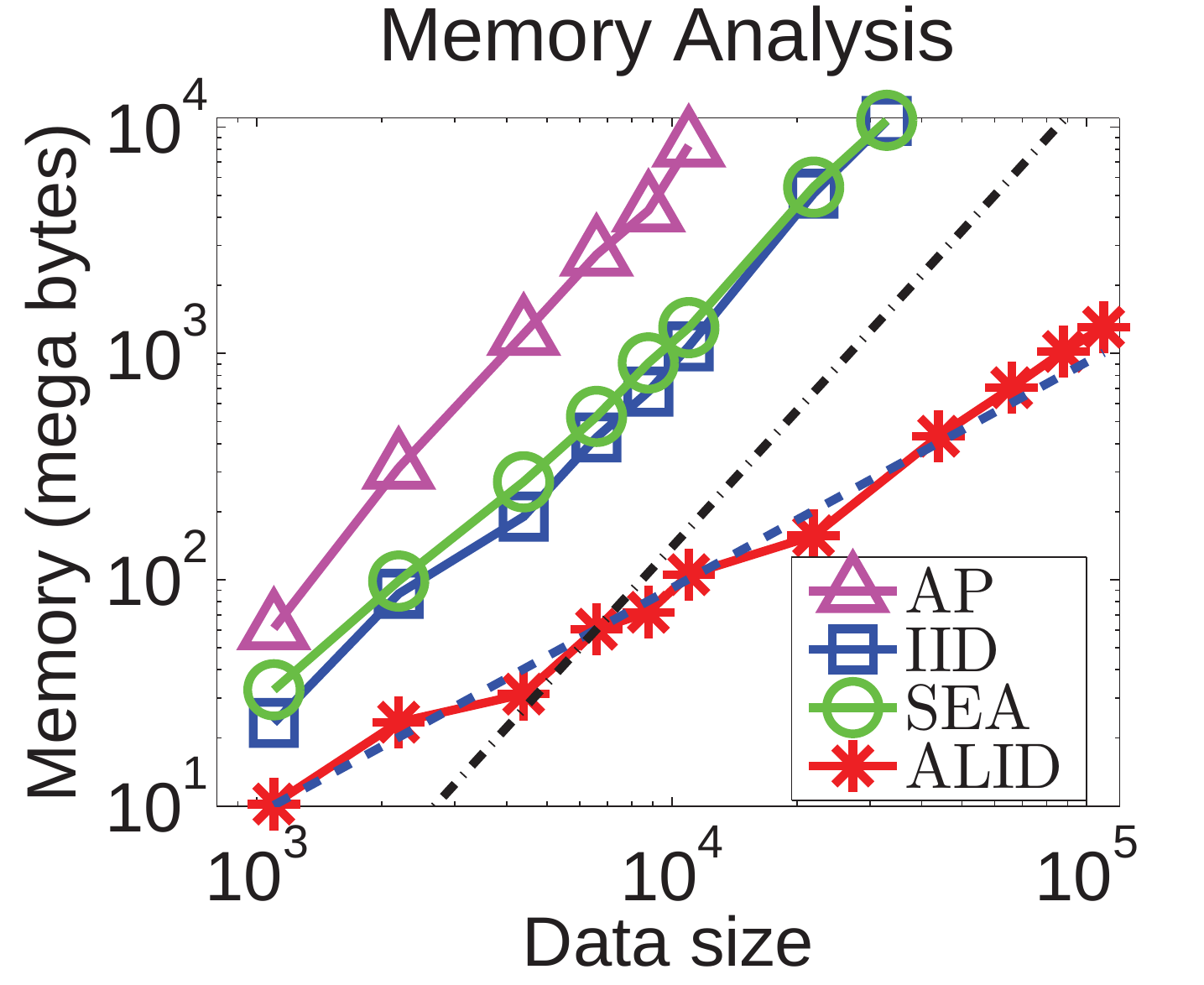}}
\subfigure[Synthetic $a^*{=}\frac{\omega n}{20}$]{\includegraphics[width=\mywidth]{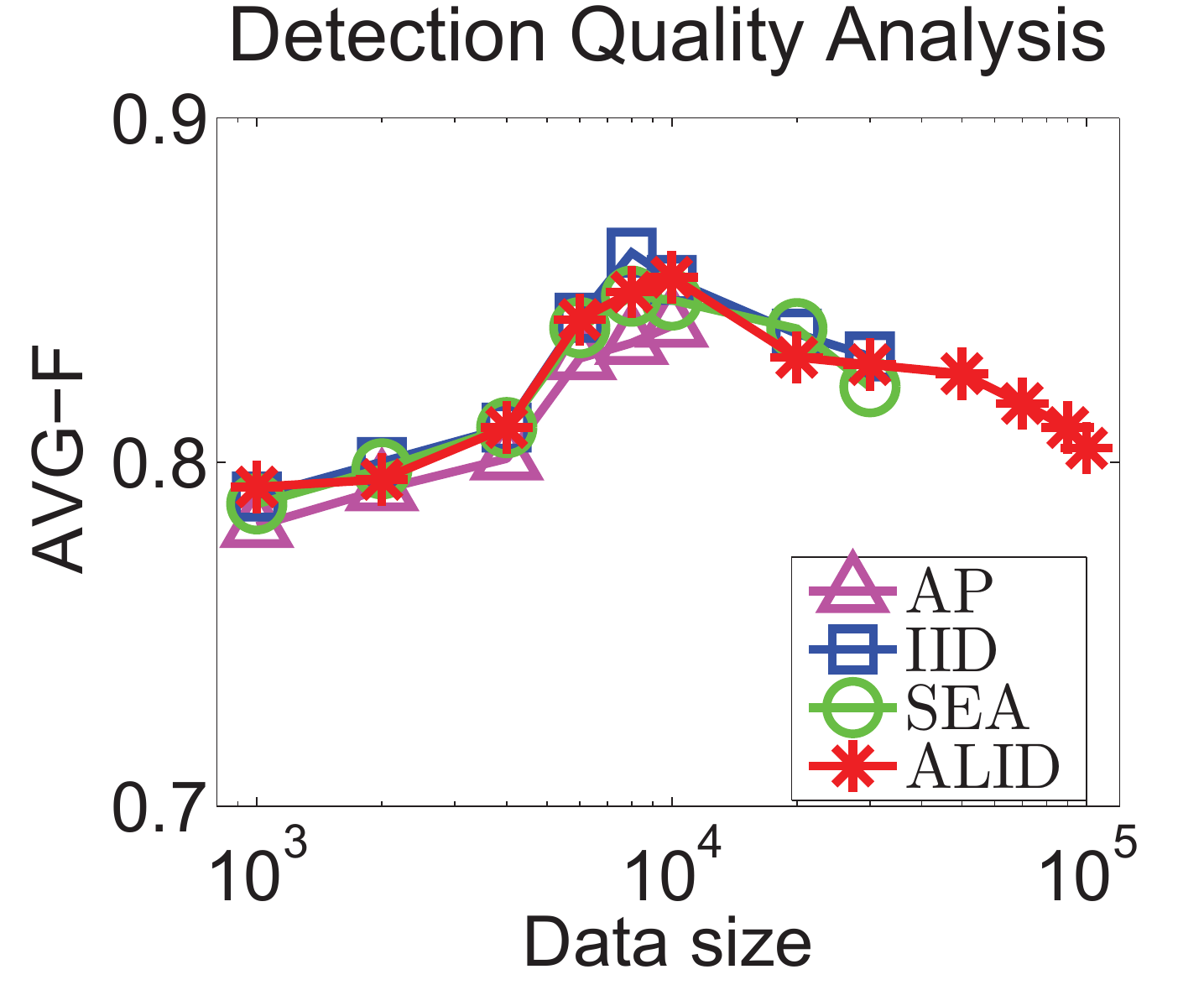}}
\hspace{\myhspace}
\subfigure[Synthetic $a^*{=}\frac{n^\eta}{20}$]{\includegraphics[width=\mywidth]{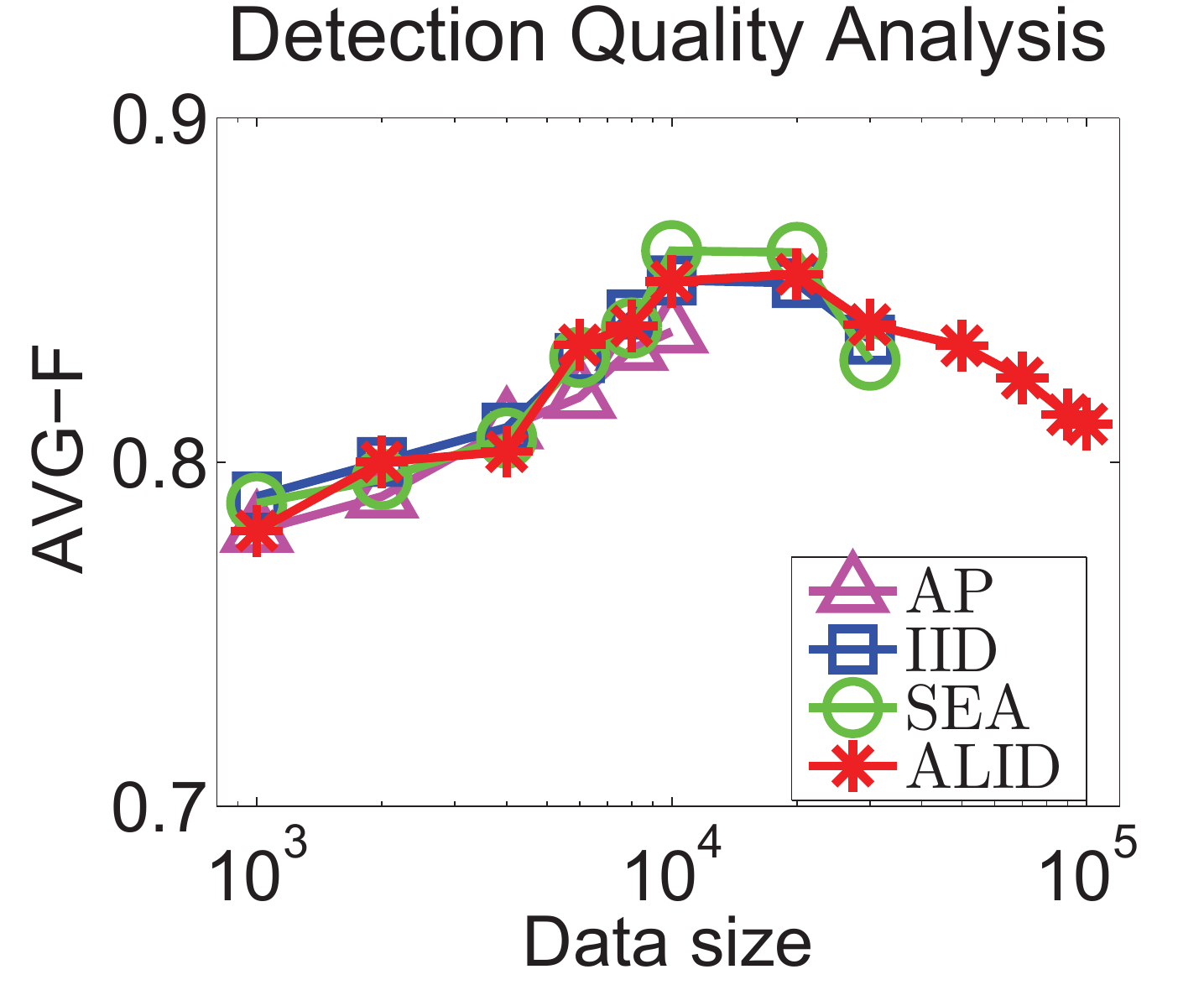}}
\hspace{\myhspace}
\subfigure[Synthetic $a^*{=}\frac{P}{20}$]{\includegraphics[width=\mywidth]{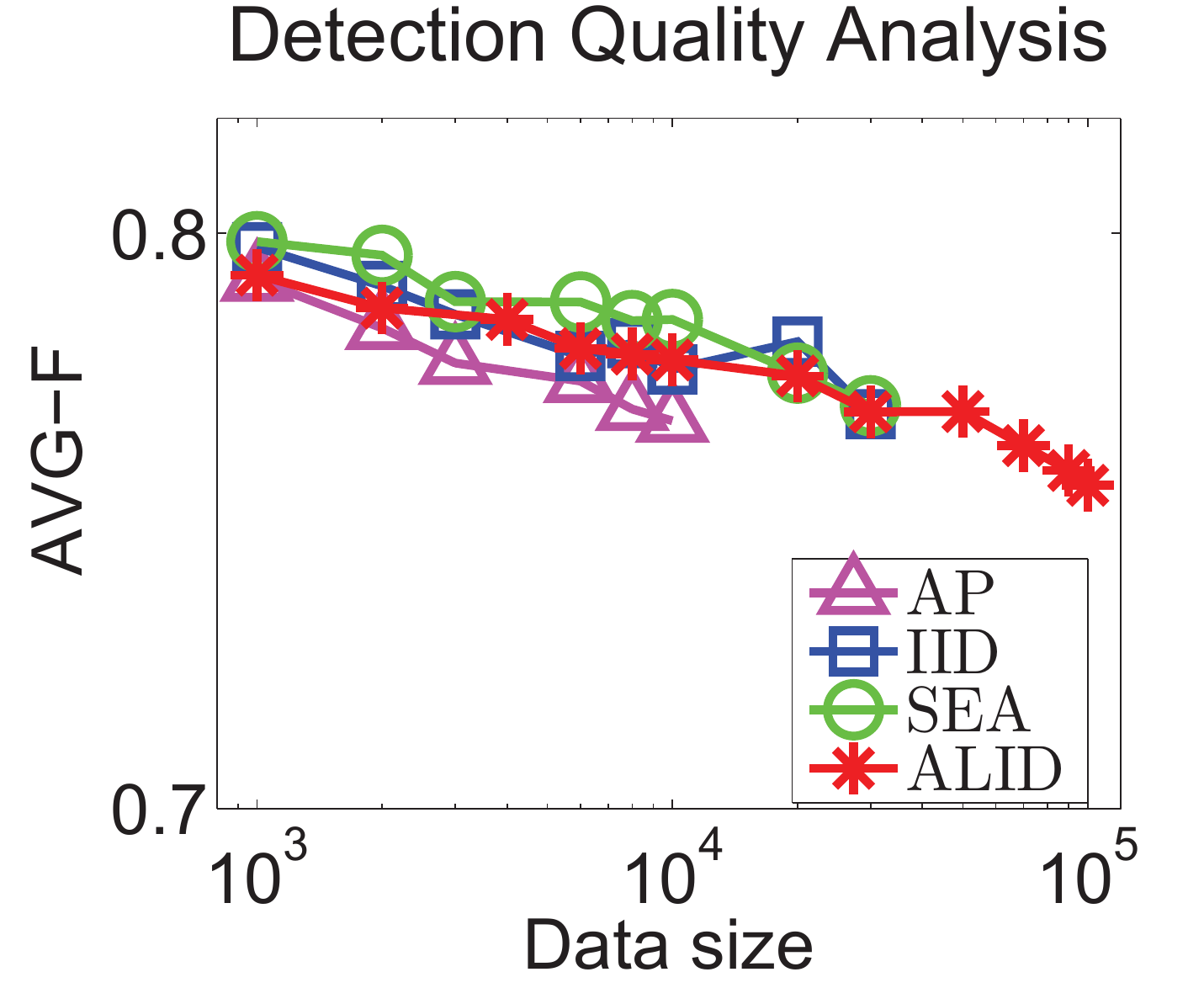}}
\hspace{\myhspace}
\subfigure[NDI data set]{\includegraphics[width=\mywidth]{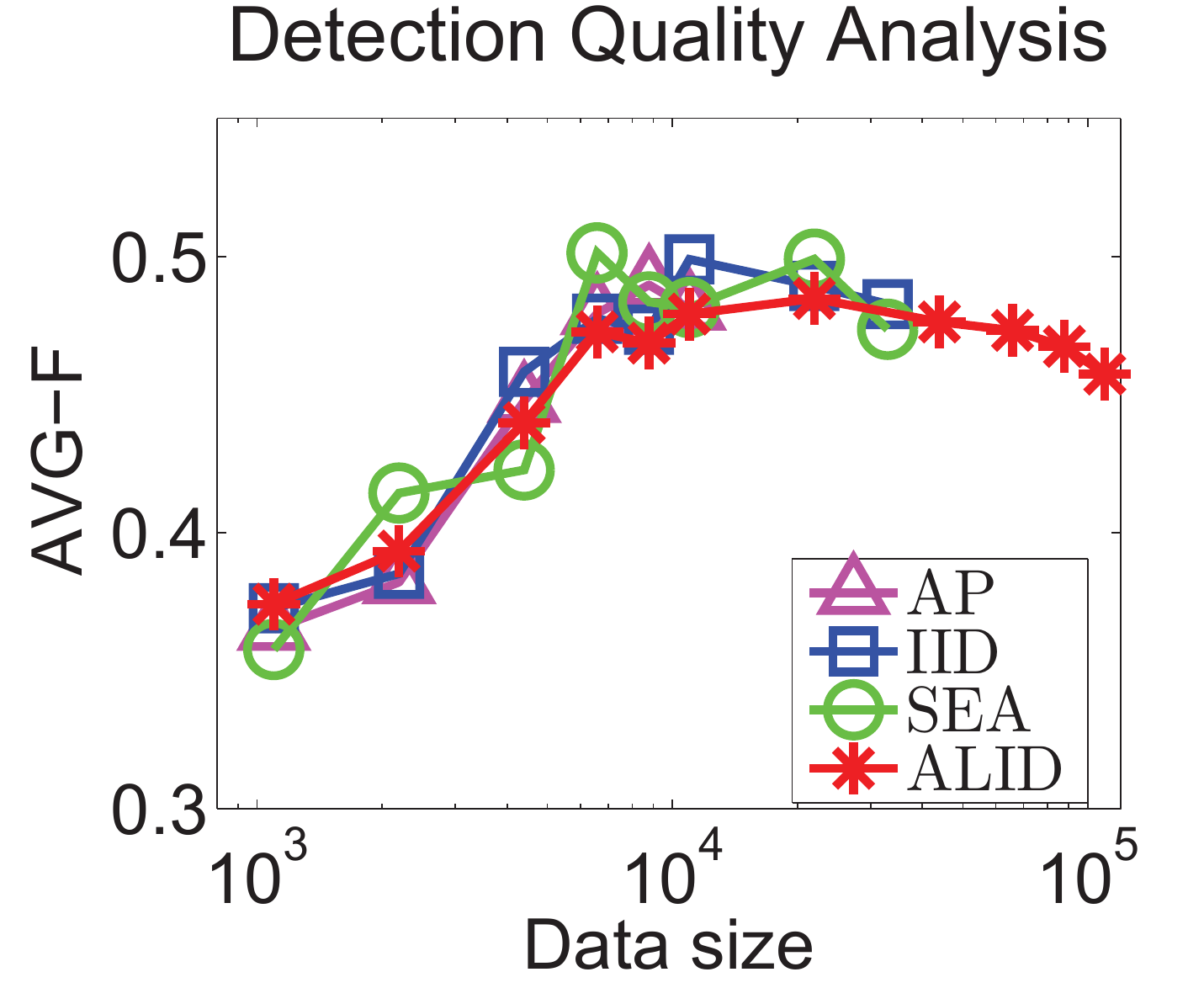}}
\caption{Scalability analysis on synthetic data sets with 3 typical cases of $a^*$ and the real world data set NDI. Parameters for the 3 synthetic data sets are: $\omega=1.0$, $\eta=0.9$ and \hlit{$P=1000$}.
\hlit{
The blue and black dashed lines are reference lines with slope 1 and 2, respectively. Under double logarithmic coordinate system, the slope of performance curves indicate the orders of growth with respect to the size of data set.}
}
\label{Fig:SA_SYN_NDI}
\end{figure*}

\subsection{Scalability Analysis}
\label{Section:SA}
In this section, we analyze the scalability of the affinity-based methods on four data sets, which consist of the real world data set NDI and the other three synthetic data sets.

The \emph{synthetic data sets} are made up by sampling $n$ 100-dimensional data items from 20 different multivariate gaussian distributions as dominant clusters and one uniform distribution as the background noise. To better simulate typical properties of real world data, we make some gaussian distributions partially overlapped by setting their mean vectors close to each other and variate the shapes of all gaussian distributions by different diagonal covariance matrices with elements ranged in $[0,10]$.
Then, we sample $a^*$ data items from each gaussian distribution as ground truth and $(n-20a^*)$ data items from the surrounding uniform distribution as noise.
\hlit{
Since all the 20 clusters are sampled in equal size, $a^*$ is the largest size of dominant clusters, which is consistent with the definition of $a^*$ in Section~\ref{Section:Complexity_Analysis}.
}

The \emph{synthetic data sets} are mainly used to test the efficiency and scalability of {\ALID}, thus we simulate the three typical cases of $a^*$ analyzed in Table~\ref{Table:complexity_bounds} by controlling the amount of sampled data items with $a^*{=}\frac{\omega n}{20}$, $a^*{=}\frac{n^\eta}{20}$ and $a^*{=}\frac{P}{20}$, where the constant denominator $20$ does not affect the complexity of {\ALID}. We use different values of
\hlit{
$n\in{[1{\times}10^3,1{\times}10^5]}$
}
to generate data sets of different sizes.
For the experiment on the NDI data set, we generate subsets of different sizes by randomly sampling the original NDI data set. The performances are evaluated in AVG-F, runtime and memory overheads of methods SEA, IID, AP and {\ALID}.


We adopt the double logarithmic coordinate system to draw all the performance curves of runtime and memory overheads. In this way, the empirical orders of runtime growth with respect to the data set size (i.e., $n$) can be easily observed from the slopes of the performance curves. Take $runtime=n^2$ as an example, the slope of the runtime curve under double logarithmic coordinate system is $\log(runtime)/\log(n)=2$, which is consistent with the theoretical order of runtime growth. We can observe the empirical order of memory growth in the same way.

As it is shown in Figures~\ref{Fig:SA_SYN_NDI}(a)-\ref{Fig:SA_SYN_NDI}(c), the empirical orders of runtime growth on the synthetic data sets are $\log(runtime)/\log(n)\approx2$ when $a^*{=}\frac{\omega n}{20}$ ($\omega=1.0$), $\log(runtime)/\log(n)\approx1.7$, when $a^*{=}\frac{n^\eta}{20}$ ($\eta=0.9$) and $\log(runtime)/\log(n)\approx1$, when $a^*{=}\frac{P}{20}$, respectively.
This result is consistent with the orders of time complexity (see Table~\ref{Table:complexity_bounds}) analyzed in Section~\ref{Section:Complexity_Analysis}. The results on the real world data set NDI (Figure~\ref{Fig:SA_SYN_NDI}(d)) also demonstrate that the empirical order of runtime growth of {\ALID} is is substantially lower than the other affinity-based methods.

As shown in Figure~\ref{Fig:SA_SYN_NDI}(e)-\ref{Fig:SA_SYN_NDI}(g), the empirical order of memory growth of {\ALID} on the three typical synthetic data sets are consistent with the theoretical space complexity (see Table~\ref{Table:complexity_bounds}) analyzed in Section~\ref{Section:Complexity_Analysis}.
The results on NDI (Figure~\ref{Fig:SA_SYN_NDI}(h)) further demonstrates the superior memory performance of {\ALID}, which consumes about 90\% less memory to process 3.3 times larger data size with a substantially lower empirical order of memory growth than the other affinity-based methods.

The experimental results in Figures~\ref{Fig:SA_SYN_NDI}(i)-\ref{Fig:SA_SYN_NDI}(l) show that {\ALID} achieves comparable AVG-F performance with the other affinity-based methods on all the four data sets. The AVG-F curves in Figures~\ref{Fig:SA_SYN_NDI}(i),\ref{Fig:SA_SYN_NDI}(j),\ref{Fig:SA_SYN_NDI}(l) have similar trends as the data set size increases: the AVG-F grows first due to the increasing dense subgraph cohesiveness caused by the growing amount of ground truth data (i.e., $a^*{=}\frac{\omega n}{20}$ and $a^*{=}\frac{n^\eta}{20}$).
Then, the AVG-F decreases due to the growing amount of ground truth data in overlapping dominant clusters and the increasing influence of noise.
The AVG-F curves in Figure~\ref{Fig:SA_SYN_NDI}(k) monotonously decreases, since the dense subgraph cohesiveness does not increase when the amount of ground truth data is fixed to $a^*{=}\frac{P}{20}$.

In summary, {\ALID} achieves remarkably better scalability than the other affinity-based methods, and at the same time retains high detection quality. The high scalability is mainly achieved by limiting all {\ALID} iterations within the ROI, which largely prevents unnecessary storage and computation of the entire affinity matrix. 

\subsection{Parallel Experiments of PALID}
\label{Section:PALID}
{\PALID} is the parallel implementation of {\ALID}. It was implemented in Java on the parallel platform of \emph{Apache Spark} (\url{http://spark.apache.org/}) on operating system Ubuntu. In this section, we evaluate the parallel performance of {\PALID} on a cluster of 5 PC computers each with an i7 CPU, 32 GB RAM and a 7200 RPM hard drive. We set 1 machine as the master and the other 4 machines as workers. Each worker is assigned with 2 executor processes, where each executor takes a single CPU core. The hash tables and the data items are stored in a MongoDB (\url{http://www.mongodb.org/}) server on the master.  All machines are connected by a 1000 Mbit/s ethernet.
The performance of {\PALID} is evaluated on the SIFT-50M data set, which is an unlabeled data set containing 50 million SIFT features~\cite{SIFT}. The SIFTs are extracted from the IPDID/1000k image data set~\cite{NP_COP} using the VLFeat toolbox~\cite{vlfeat}.

\begin{figure}[h]
\centering
\includegraphics[width=83mm]{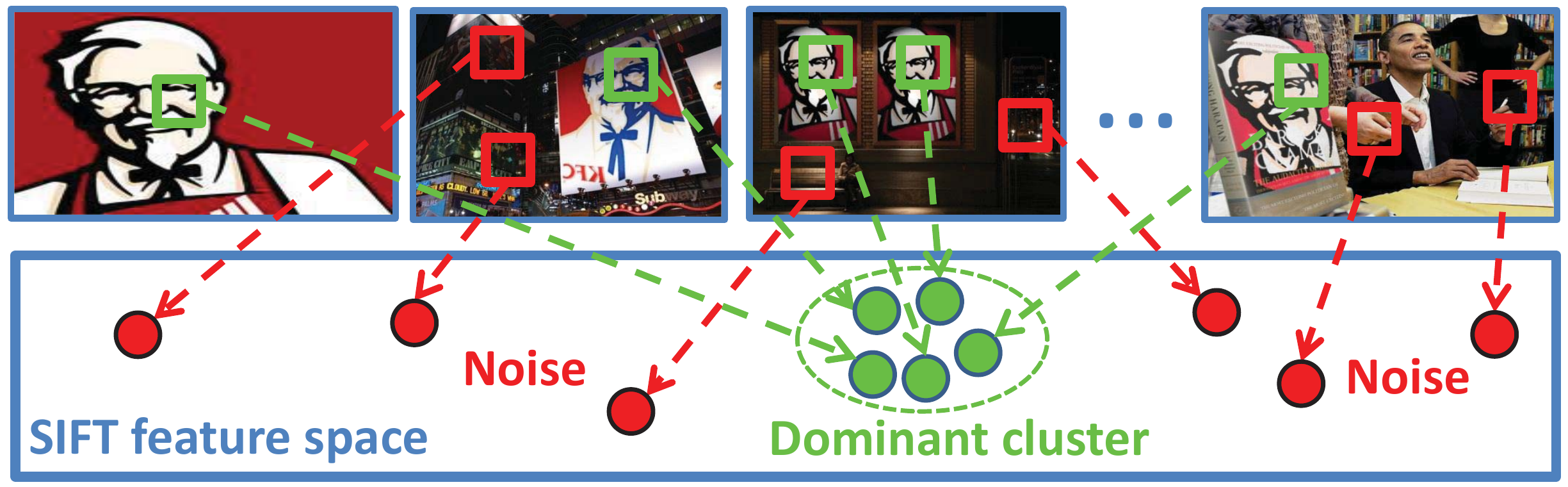}
\vspace{-3mm}
\caption{Illustration of SIFT dominant cluster.}
\label{Fig:SIFT_cluster}
\vspace{-1mm}
\end{figure}

Scale Invariant Feature Transform (SIFT)~\cite{SIFT} is a $L_2$ normalized 128-dimensional vector describing the texture of a small image region. In the field of computer vision, it is a standard procedure to represent similar image regions by highly cohesive SIFT dominant clusters named ``\emph{visual words}''~\cite{BOW}. As shown in Figure~\ref{Fig:SIFT_cluster}, since partial duplicate images always share a common image content (e.g., KFC grandpa), the SIFTs~\cite{SIFT} extracted from similar image regions are highly similar to each other and naturally form a dominant cluster (i.e., visual word). However, the number of visual words is unknown and there are also a large proportion of noisy SIFTs extracted from the random non-duplicate regions (i.e., red points in Figure~\ref{Fig:SIFT_cluster} and Figure~\ref{Fig:IPDID-SIFT_results}), leading to a high degree of background noise.
As a result, the scalability and strong resistance against noise of {\PALID} are very suitable for visual word generation.

Table~\ref{Table:SA_SIFT_50M} shows the parallel performance of {\PALID}. {\PALID} is able to process 50 million SIFT features in $2.29$ hours, achieving a speedup ratio of 7.51 with 8 executors. This demonstrates the promising parallel performance of {\PALID}.

\hlit{
Since there is no parallel solution for the other affinity-based methods, such as IID\cite{IID}, SEA\cite{SEA} and AP\cite{AP}, we fairly compare them with {\ALID} on the SIFT-50M data set, using the same single-machine experimental settings as Section~\ref{Section:SA} used.
Figure~\ref{Fig:SA_IPDID_SIFT} shows the memory and runtime performances of the affinity-based methods on uniformly sampled subsets of SIFT-50M, where all experiments are stopped when the 12GB RAM limit is reached.
We can see that the empirical orders of runtime and memory growth of {\ALID} are significantly lower than the other methods. Especially, {\ALID} consumes 10 GB memory to process 1.29 million SIFTs in $4.4$ hours on a standard PC. In contrast, such a large amount of data is far beyond the capability of the other affinity-based methods, which can at most deal with $0.04$ million SIFTs on the same platform.
}

\begin{table}[t]
  \centering
  \caption{Performance of PALID on SIFT-50M}
  \label{Table:SA_SIFT_50M}
  \begin{tabular}{|c|c|c|c|c|}
    \hline
     Methods       &Executors &    Runtime   &  Speadup Ratio       \\ \hline
     PALID-1Exec   &     1    &     17.2 hours&     1                \\ \hline
     PALID-2Exec   &     2    &     8.96 hours&     1.92             \\ \hline
     PALID-4Exec   &     4    &     4.48 hours&     3.84             \\ \hline
     PALID-8Exec   &     8    &     2.29 hours&     7.51             \\ \hline
  \end{tabular}
\end{table}

\begin{figure}[t]
\centering
\includegraphics[width=43mm]{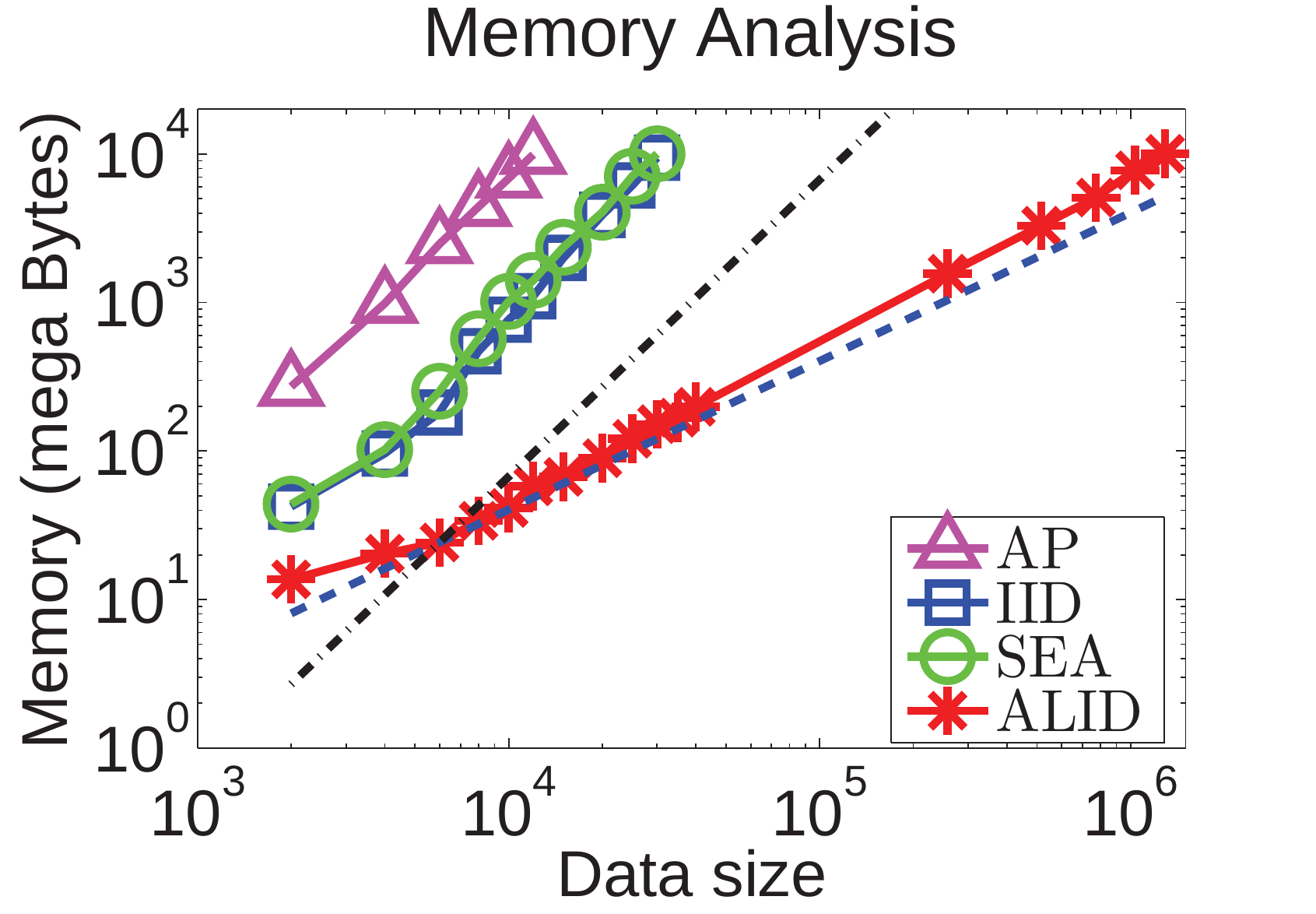}
\hspace{-4mm}
\includegraphics[width=43mm]{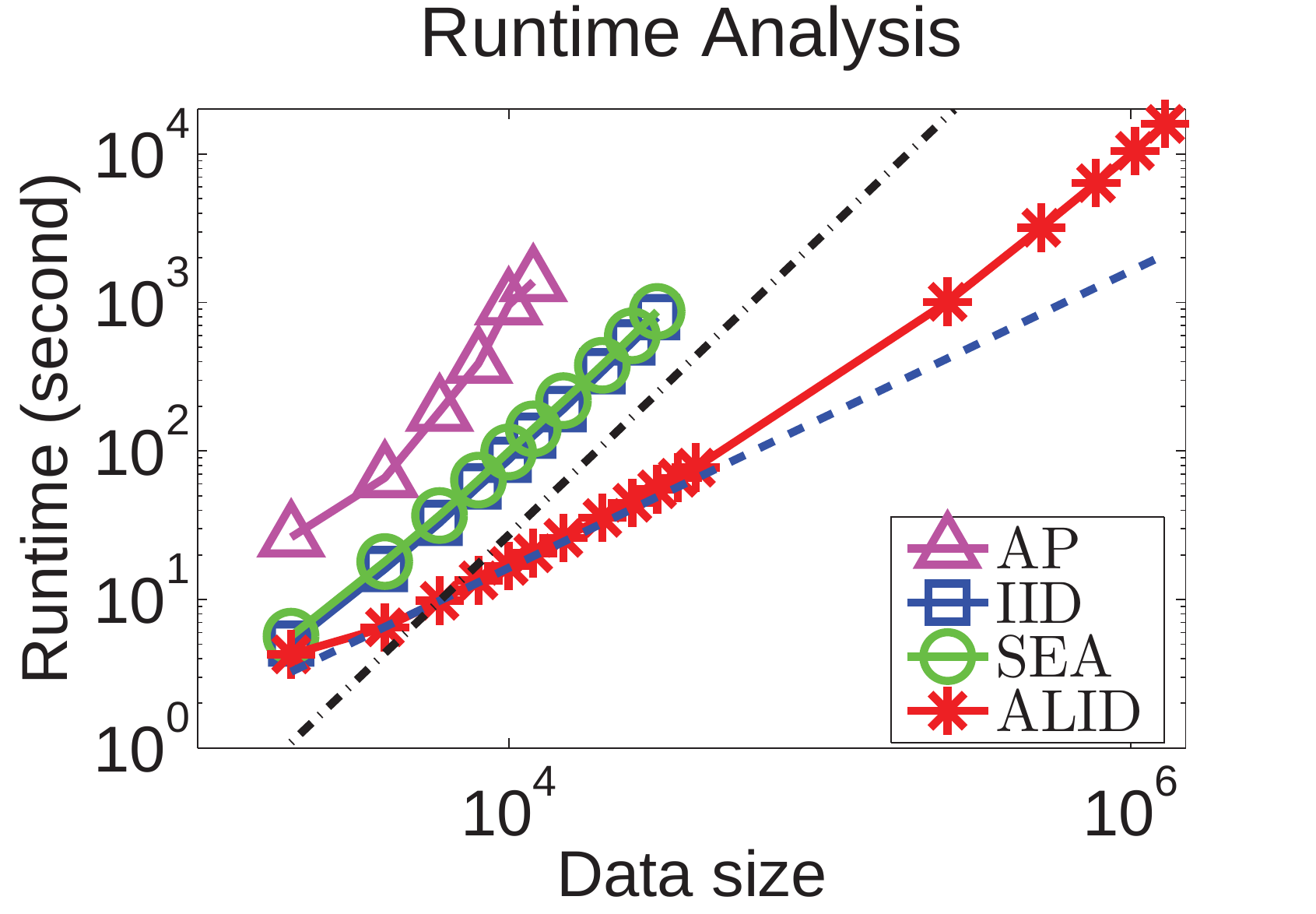}
\vspace{-5mm}
\caption{Scalability analysis on SIFT-50M subset.}
\label{Fig:SA_IPDID_SIFT}
\end{figure}

\newcommand{\mfigwidth}{25mm}
\begin{figure}[h]
\centering
\hspace{-3mm}
\subfigure[Original]{\includegraphics[width=\mfigwidth]{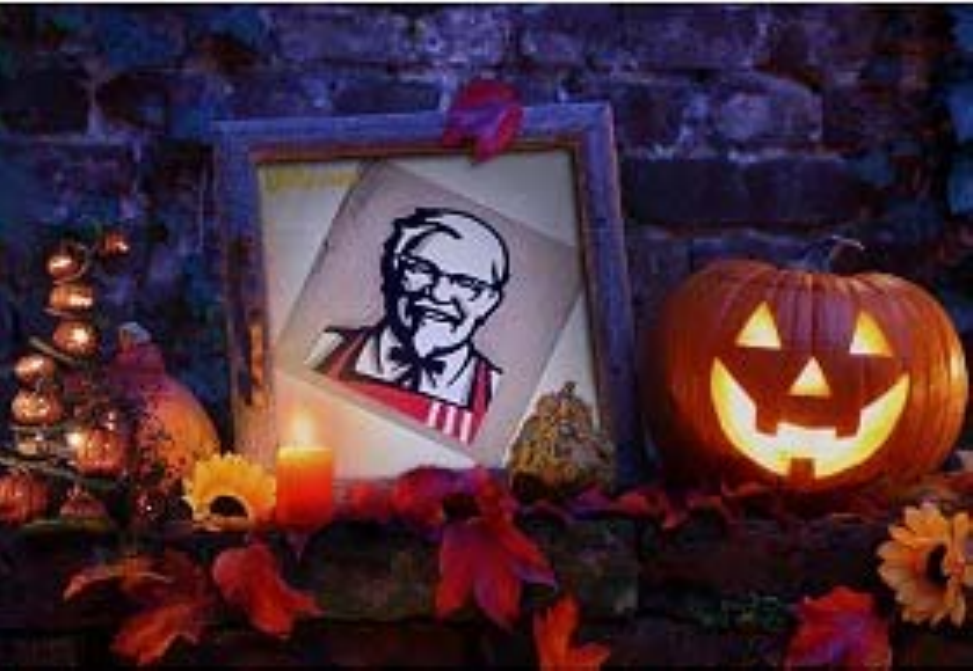}}
\subfigure[{\PALID}]{\includegraphics[width=\mfigwidth]{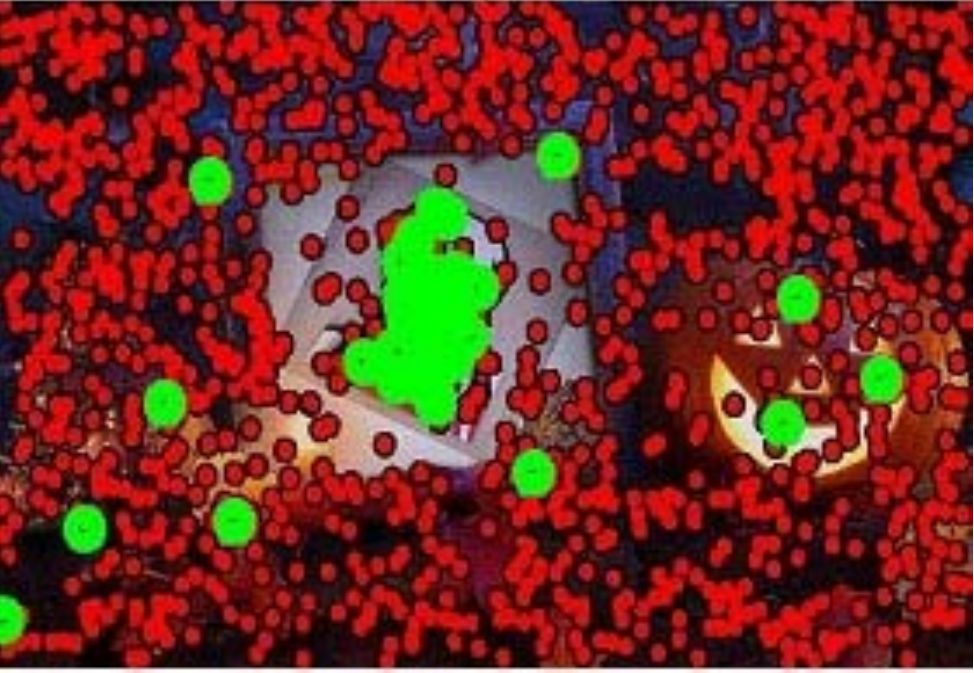}}
\subfigure[{\ALID}]{\includegraphics[width=\mfigwidth]{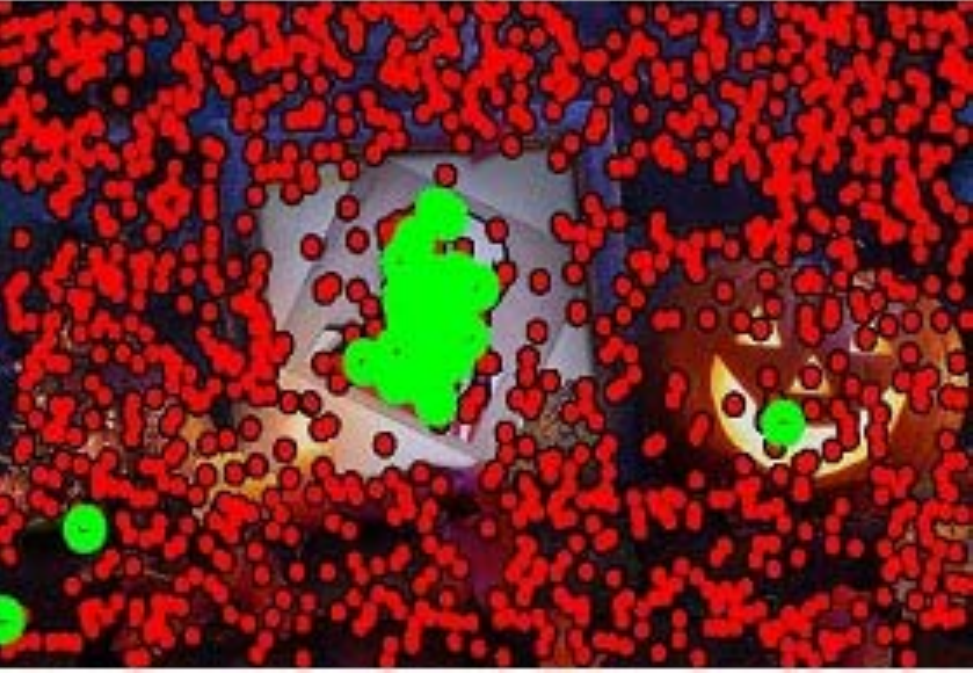}}
\subfigure[IID]{\includegraphics[width=\mfigwidth]{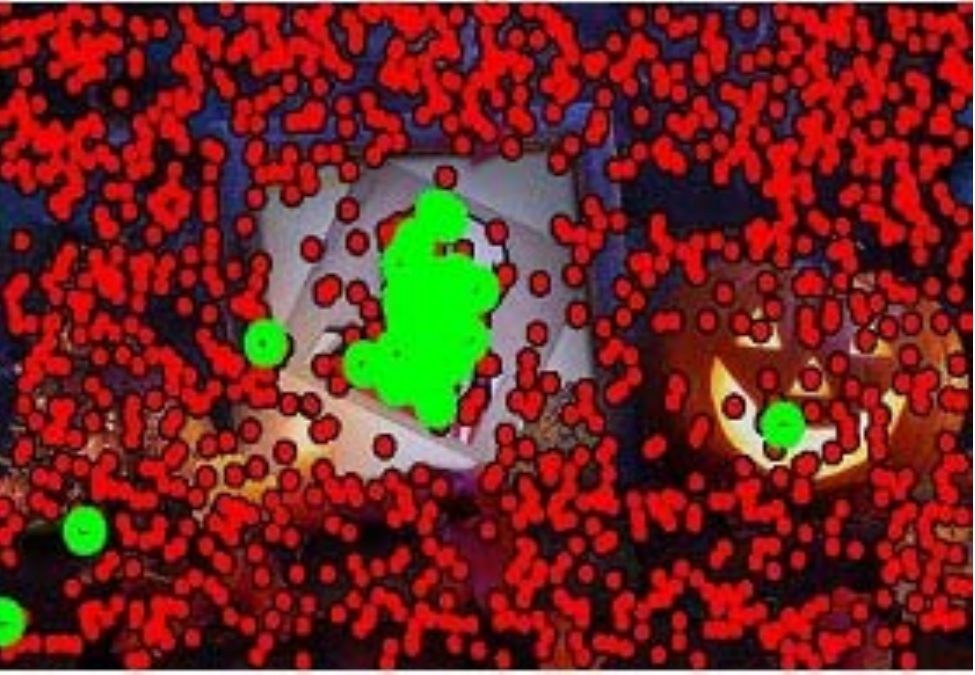}}
\subfigure[SEA]{\includegraphics[width=\mfigwidth]{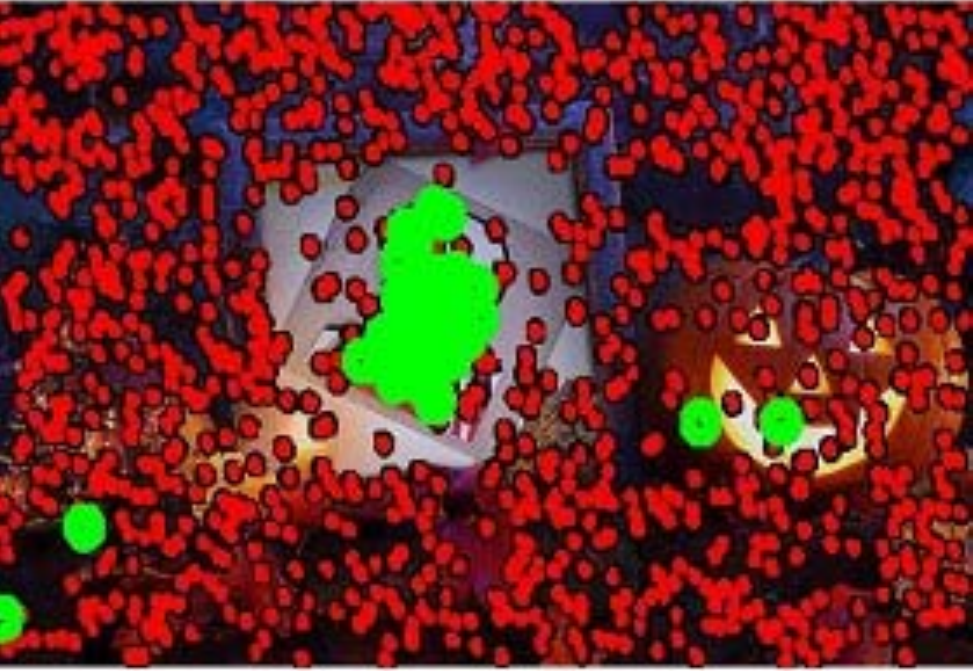}}
\subfigure[AP]{\includegraphics[width=\mfigwidth]{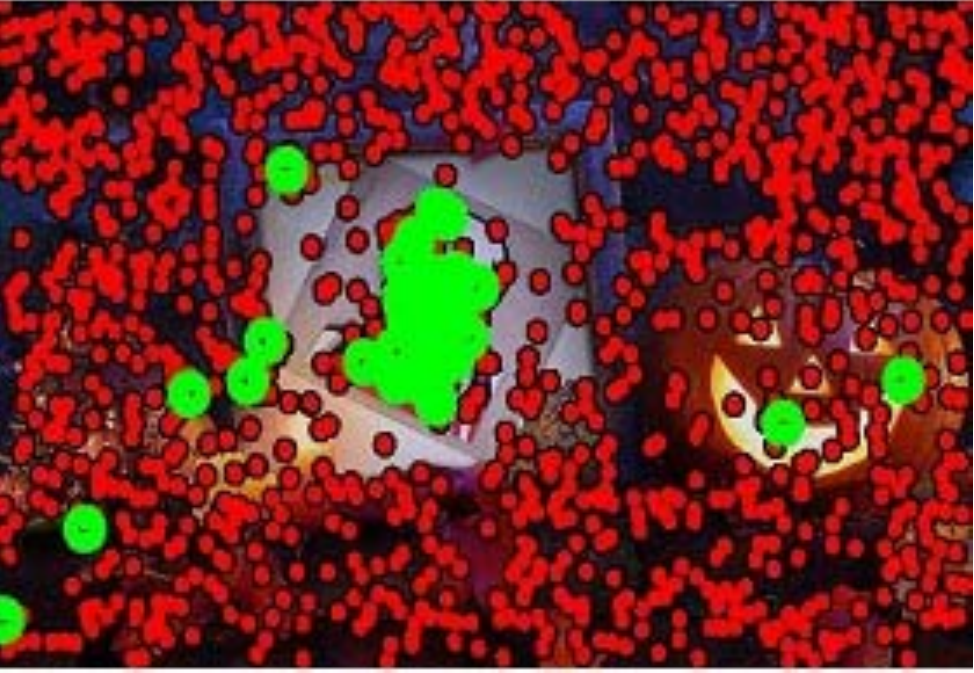}}
\nop{
\subfigure[KM]{\includegraphics[width=\mfigwidth]{FINAL_FIGS/kfc_KM}}
\subfigure[SC-FL]{\includegraphics[width=\mfigwidth]{FINAL_FIGS/kfc_SC-FL}}
\subfigure[SC-NYS]{\includegraphics[width=\mfigwidth]{FINAL_FIGS/kfc_SC-NYS}}
}
\vspace{-3mm}
\caption{Detection quality on (a) ``KFC grandpa'' of the IPDID/1000k data set. The green points are SIFTs from dominant clusters with high densities ($\pi{(x)}>0.75$). The red points in (b)-(f) are noise SIFTs filtered out by the affinity-based methods.
\nop{
The partitioning-based methods of (g) KM, (h) SC-FL and (i) SC-NYS cannot filter out noise SIFTs on background image regions, thus all SIFTs are plotted as green points.}}
\label{Fig:IPDID-SIFT_results}
\end{figure}

Despite that SIFT-50M is too large to be manually labeled, the dominant cluster detection quality can still be qualitatively assessed by the results in Figure~\ref{Fig:IPDID-SIFT_results}. As shown in Figures~\ref{Fig:IPDID-SIFT_results}(b)-\ref{Fig:IPDID-SIFT_results}(f), the affinity-based methods effectively detect most of the SIFT features (i.e., green points) extracted from the similar image regions of ``KFC grandpa'', since such SIFTs can naturally form dense subgraphs with high cohesiveness.
At the same time, the large proportion of noisy SIFTs (i.e., red points) extracted from the random non-duplicate background image regions are filtered out. The results demonstrate the effectiveness of affinity-based methods in resisting overwhelming amount of noisy SIFTs.
Additionally,  we also evaluate the AVG-F performance of {\PALID} on the labeled data sets of NART and NDI. The resulting AVG-F performances are consistent with {\ALID}.

\nop{
In constrast, since the partitioning-based methods, such as KM~\cite{KM_Lloyd}, SC-FL~\cite{SC_tradition} and SC-NYS~\cite{NYS_PAMI04}, force assign every SIFT to a cluster, they are not able to distinguish the dominant cluster from noise data. Therefore, all SIFTs in Figures~\ref{Fig:IPDID-SIFT_results}(g)-\ref{Fig:IPDID-SIFT_results}(i) are plotted as green points.
}


\section{Conclusions}\label{sec:con}
In this paper, we proposed {\ALID}, a scalable and effective dominant cluster detection approach against high background noise.
{\ALID} demonstrates remarkable noise resistance capability, achieves significant scalability improvement over the other affinity-based methods and is highly parallelizable in MapReduce framework.
As future work, we will further extend {\ALID} towards the online version to efficiently process streaming data sources.
\fussy

\bibliographystyle{abbrv}

\balance
\bibliography{references}

%
%

\newpage

\begin{table*}[t]
\centering
\caption{The definitions of symbols and acronyms.}
\label{table:lookup}
\newcommand{\verline}{\vrule width 1.5pt}
\begin{tabular}{|c|p{65mm}!{\verline}c|p{65mm}|}\hline

\textbf{Symbols}    & \textbf{Definitions} & \textbf{Acronyms}  & \textbf{Definitions} \\ \hline

$n$       & The total number of graph vertexes. &
{\ALID}   & The proposed method in this paper. \\ \hline

$G$       & The global affinity graph $G=(V,I,A)$. &
{\PALID}  & The parallel implementation of {\ALID} \\ \hline

$V$       & The set of all graph vertexes $V=\{v_i\in R^d \mid i\in I\}$. $R^d$ is $d$-dimensional real number space. &
LID       & Localized infection immunization dynamics (Step 1 of Algorithm~\ref{alg:ALID_final}). \\ \hline

$I$       & The ``global range'', which is the index set of all graph vertexes $I=[1,n]$. &
ROI       & Region of interest (Step 2 of Algorithm~\ref{alg:ALID_final}). \\ \hline

$A$       & The global affinity matrix of $G$. &
CIVS      & Candidate infective vertex search (Step 3 of Algorithm~\ref{alg:ALID_final}). \\ \hline

$v_i$     & The $d$-dimensional data item that uniquely corresponds to the $i$-th graph vertex.  &
LSH       & The locality sensitive hashing method~\cite{LSH}. \\ \hline

$s_i$     & The $n$-dimensional index vector that uniquely represents the $i$-th graph vertex.  &
DS        & The dominant set method~\cite{DS}. \\ \hline

$x$       & The $n$-dimensional vector to represent any subgraph of $G$. &
RD        & The replicator dynamics method~\cite{Weibull}. \\ \hline

$x_i$     & The $i$-th dimension of $x$. $x_i$ is also the weight of vertex $s_i$ in subgraph $x$. &
StQP      & Standard quadratic optimization problem (Equation~\ref{Eqn:StQP}). \\ \hline

$\gamma_\beta(x)$   &  The set of infective subgraphs that are infective against $x$ and are within local range $\beta$. $\gamma_\beta(x)=\{y\in\triangle_\beta^n \mid \pi(y-x,x)>0\}$.&
IID       & The infection immunization dynamics method~\cite{IID}. \\ \hline

$a^*$     & The number of vertexes in the largest dense subgraph (i.e., dominant cluster). &
SEA       & The shrink and expansion algorithm~\cite{SEA}. \\ \hline

$\alpha$  & The index set of all vertexes in subgraph $x$. &
AP        & The affinity propagation method~\cite{AP}. \\ \hline

$\beta$   & The ``local range'', which is the index set of a local group of vertexes. &
LSR       & The locality sensitive region (Figure~\ref{Fig:hive}). \\ \hline

$A_{\beta\beta}$, $A_{\beta\alpha}$    & Local affinity matrices (see Figure~\ref{Fig:LID_matrix}). &
w.r.t.    & This means ``with respect to''. \\ \hline

$c\leq C$ & $c$ is the current number of iterations of {\ALID}. $C$ is the maximum iteration limit of {\ALID}. &
AVG-F     & Average $\text{F}_1$ score. \\ \hline

$\omega$, $\eta$, $P$    & The constant parameters in Table~\ref{Table:complexity_bounds}. &
NDI       & The near duplicate image data set. \\ \hline

$\delta$  & The maximum number of vertexes that can be retrieved by CIVS (Step 3). &
Sub-NDI   & The subset of NDI data set (Section~\ref{Section:SIA}). \\ \hline

$\triangle^n$, $\triangle_\beta^n$   &
$\triangle^n=\{x\in{R^n}\mid\sum_{i\in I}{x_i}=1,x_i\geq0\}$ is the set of all possible subgraphs in global range $I$.
$\triangle_\beta^n=\{x\in \triangle^n \mid \sum_{i\in \beta}x_i=1\}$ is the set of all possible subgraphs in local range $\beta$. &
SIFT-50M  & The data set of 50 million SIFT features~\cite{SIFT}. \\ \hline

$\hat x$, $\hat x^{(c)}$, $x^*$   & $\hat x$ is the local dense subgraph found by LID. $\hat x^{(c)}$ is the dense subgraph detected by LID in the $c$-th iteration of {\ALID}. $x^*$ is the global dense subgraph (i.e., the output of {\ALID}). &
NART      & The news articles data set. \\ \hline

\end{tabular}
\end{table*}

\appendix

\section{Proof of Proposition~1}
\label{sec:double_hyperball_proof_appendix}

\newtheorem{proposition_appendix}{Proposition}
\begin{proposition_appendix}
\label{prop:triangle_appendix}
Given the local dense subgraph $\hat{x}$, the double-deck hyperball $H(D,R_{in},R_{out})$ has the following properties:
\begin{enumerate}
  \item $\forall{j}\in{I}$ and $\vnorm{v_j-D}<R_{in}$, $\pi(s_j-\hat{x},\hat{x})>0$; and
  \item $\forall{j}\in{I}$ and $\vnorm{v_j-D}>R_{out}$, $\pi(s_j-\hat{x},\hat{x})<0$.
\end{enumerate}
\end{proposition_appendix}

\begin{proof}
We refer to Equantion~\ref{Eqn:define_hyperball} in Section~\ref{sec:Estimate_ROI} as
\begin{equation}\notag
\left\{
\begin{array}{l}
D=\sum\limits_{i\in\alpha}{v_i\hat x_i},\;where\;v_i\in{V}\;are\;the\;data\;items. \\[3mm]
R_{in}=\frac{1}{k}\ln(\frac{\lambda_{in}}{\pi(\hat{x})}),\;where\;\lambda_{in}{=}\sum\limits_{i\in\alpha}{\hat x_ie^{-k\vnorm{v_i-D}}} \\[1mm]
R_{out}=\frac{1}{k}\ln(\frac{\lambda_{out}}{\pi(\hat{x})}),\;where\;\lambda_{out}{=}\sum\limits_{i\in\alpha}{\hat x_ie^{k\vnorm{v_i-D}}}
\end{array}\right.\;\;\,\eqref{Eqn:define_hyperball}
\end{equation}

Define $f_{in}(v_j)$ and $f_{out}(v_j)$ as
\begin{equation}
\label{Eqn:f_vi_inout}
\left\{
\begin{array}{l}
f_{in}(v_j)=\lambda_{in}e^{-k\vnorm{v_j-D}} \\[1mm]
f_{out}(v_j)=\lambda_{out}e^{-k\vnorm{v_j-D}}
\end{array}\right.
\end{equation}
By plugging $\lambda_{in}$ and $\lambda_{out}$ (Equation~\ref{Eqn:define_hyperball}) into Equation~\ref{Eqn:f_vi_inout}, we have:
\begin{equation}
\label{Eqn:f_vi_inout_specific}
\left\{
\begin{array}{l}
f_{in}(v_j)=\sum\limits_{i\in\alpha}\hat x_ie^{-k(\vnorm{v_j-D}+\vnorm{v_i-D})} \\
f_{out}(v_j)=\sum\limits_{i\in\alpha}\hat x_ie^{-k(\vnorm{v_j-D}-\vnorm{v_i-D})}
\end{array}\right.
\end{equation}
Recall that the scaling factor $k>0$ is always positive (Equation~\ref{Eqn:exp_kernel_affinity}). Applying the triangle inequality to Equation~\ref{Eqn:f_vi_inout_specific}, we obtain
\begin{equation}
\label{Eqn:hyperball_final_proof}
\left\{
\begin{array}{l}
\forall{j}\in{I}, \pi(s_j,\hat x)\geq{f_{in}(v_j)} \\
\forall{j}\in{I}, \pi(s_j,\hat x)\leq{f_{out}(v_j)}
\end{array}\right.
\end{equation}
where $\pi(s_j,\hat x)=(A\hat x)_j=\sum\limits_{i\in\alpha}\hat x_ie^{-k\vnorm{v_j-v_i}}$.

For any vertex $v_j$ that satisfies $\vnorm{v_j-D}=R_{in}$, we have
\begin{equation}
f_{in}(v_j)=\lambda_{in}e^{-kR_{in}}=\pi(\hat{x})
\end{equation}
by plugging $R_{in}$ of Equation~\ref{Eqn:define_hyperball} into $f_{in}(v_j)$ of Equation~\ref{Eqn:f_vi_inout}.

Since $f_{in}(v_j)$ in Equation~\ref{Eqn:f_vi_inout} monotonously decreases with respect to $\vnorm{v_j-D}$, we can derive
\begin{equation}
\label{Eqn:hyperball_final_proof_chain_mid}
\forall{j}\in{I} \text{ and } \vnorm{v_j-D}<R_{in}, f_{in}(v_j)>\pi(\hat{x})
\end{equation}
Therefore, we can derive from Equation~\ref{Eqn:hyperball_final_proof_chain_mid} and the first inequation of Equation~\ref{Eqn:hyperball_final_proof} that
\begin{equation}
\forall{j}\in{I} \text{ and } \vnorm{v_j-D}<R_{in}, \pi(s_j,\hat{x})\geq f_{in}(v_j)>\pi(\hat{x})
\end{equation}
which proves the first property of Proposition~\ref{prop:triangle_inequality_hyperball}.

The second property of Proposition~\ref{prop:triangle_inequality_hyperball} can be proved in a similary way by plugging $\vnorm{v_j-D}=R_{out}$ of Equantion~\ref{Eqn:define_hyperball} into $f_{out}(v_j)$ of Equation~\ref{Eqn:f_vi_inout_specific}.
\end{proof}

\section{Proof of Convergence}
\label{sec:convergence_proof_appendix}

\begin{proposition}
\label{prop:convergence}
For a global dense subgraph $x^*$ that contains $M$ vertexes, denote $\hat x^{(c)}$ as the local dense subgraph detected by \emph{Step 1} of Algorithm~\ref{alg:ALID_final} (i.e., {\ALID}) in the $c$-th iteration, then we have the following two properties:
\begin{enumerate}
  \item {\ALID} is guaranteed to converge; and
  \item The statistical expectation of the number of vertexes in $\hat x^{(c)}$ converges to $M$.
\end{enumerate}
\end{proposition}

\begin{proof}
Some useful notations as listed as follow:
\begin{itemize}
  \item $x^*$ is the global dense subgraph containing $M$ vertexes.
  \item $\hat x^{(c)}$ denotes the local dense subgraph detected by LID (i.e., Step 1) in the $c$-th iteration of Algorithm~\ref{alg:ALID_final}.
  \item $a^{(c)}$ denotes the statistical expectation of the number of vertexes in local dense subgraph $\hat x^{(c)}$.
  \item $H_c(D,R)$ is the ROI estimated from $\hat x^{(c)}$ in Step 2.
  \item $m^{(c)}\leq M$ represents the number of vertexes that belong to $x^*$ and are within the ROI $H_c(D,R)$.
  \item $p\in(0,1)$ is the lower bound of recall of the LSH method \cite{LSH} in CIVS (i.e., Step 3). Theoretical proof of the lower bound of recall of LSH can be found in M. Datar's work~\cite{LSH}.
\end{itemize}

First, we prove the convergence of {\ALID}, which is the first property of Proposition~\ref{prop:convergence}.

According to Theorem~\ref{theorem:nash_gamma_equal} and Theorem~\ref{theorem:infection_immune}, given any local range $\beta^{(c)}$, the LID method in Step 1 of Algorithm~\ref{alg:ALID_final} is guaranteed to converge to the local dense subgraph $\hat x^{(c)}\in\triangle_{\beta^{(c)}}^n$ within the local range $\beta^{(c)}$. That is,
\begin{equation}
\label{Eqn:x_c_max_LID}
\hat x^{(c)}=\max\limits_{x\in\triangle_{\beta^{(c)}}^n}{\pi(x)}
\end{equation}
Recall that $\beta^{(c)}$ is updated by the vertexes retrieved by CIVS in Step 3 of Algorithm~\ref{alg:ALID_final}, we discuss how {\ALID} converges in the following cases. For the interest of simplicity, we regard all the vertexes in the updated local range $\beta^{(c)}=a^{(c-1)}\cup \psi$ as retrieved by CIVS.

\begin{description}
\item [Case 1]: CIVS retrieves all the $M$ vertexes of the global dense subgraph $x^*$. In this case, we have $x^*\in\triangle_{\beta^{(c)}}^n$. By the definition of global dense subgraph (See Equation~\ref{Eqn:StQP}), we can drive that
     \begin{equation}
     x^*=\max\limits_{x\in\triangle^n}{\pi(x)}
     \end{equation}
     Since $\triangle_{\beta^{(c)}}^n\subset\triangle^n$ and $x^*\in\triangle_{\beta^{(c)}}^n$, we have
     \begin{equation}
     x^*=\max\limits_{x\in\triangle_{\beta^{(c)}}^n}{\pi(x)}
     \end{equation}
     Therefore, we can derive from Equation~\ref{Eqn:x_c_max_LID} that LID is guaranteed by Theorem~\ref{theorem:nash_gamma_equal} and Theorem~\ref{theorem:infection_immune} to converge to the global dense subgraph $\hat x^{(c)}=x^*$. In other words, {\ALID} converges to the global dense subgraph $x^*$.
\item [Case 2]: CIVS retrieves $N$ vertexes ($0<N<M$) of the global dense subgraph $x^*$. In this case, the $N$ retrieved vertexes of $x^*$ form a subgraph $x'\in\triangle_{\beta^{(c)}}^n$, which is a subgraph of the global dense subgraph $x^*\in\triangle^n$. Considering the fact that all vertexes of $x^*$ are highly similar with each other, we know that the vertexes in $x'$ are highly similar with each other as well. Thus, the graph density $\pi(x')$ is the maximum in the local range $\beta^{(c)}$, that is
    \begin{equation}
    x'=\max\limits_{x\in\triangle_{\beta^{(c)}}^n}{\pi(x)}
    \end{equation}
    Therefore, we can derive from Equation~\ref{Eqn:x_c_max_LID} that LID is guaranteed to converge to $\hat x^{(c)}=x'$ in the $c$-th iteration of {\ALID}. As proved later, the statistical expectation of the number of vertexes in $\hat x^{(c)}$ converges to $M$. In other words, {\ALID} converges and the expected detection quality converges to the optimal result.

\item [Case 3]: CIVS retrieves zero vertex of the global dense subgraph $x^*$. This is an ill-conditioned case, which only happens when the recall of LSH~\cite{LSH} is $p\approx 0$ under improper parameters, such as extremely large number of hash functions and small segment length $r$ ( Figure~\ref{Fig:SIA_COMP}). In this case, LID is still guaranteed to converge to a local dense subgraph $\hat x^{(c)}\in\triangle_{\beta^{(c)}}^n$, which does not contain any vertex of $x^*$. However, such ill-condition can be easily avoided by setting the LSH parameters properly~\cite{LSH}.
\end{description}

In summary, {\ALID} is guaranteed to converge and the detection quality of {\ALID} depends on how many vertexes of the global dense subgraph $x^*$ is retrieved by CIVS.

In the following, we further prove that the statistical expectation of the number of vertexes in $\hat x^{(c)}$ converges to $M$, which is the second property of Proposition~\ref{prop:convergence}.

For Case 1, we have proved that $\hat x^{(c)}=x^*$, thus the second property of Proposition~\ref{prop:convergence} is valid for Case 1.

For Case 2, we analyze the situation in detail as follows. By the definition of the global dense subgraph, all the $M$ vertexes in the global dense subgraph $x^*$ are highly similar with each other, therefore, when using a single graph vertex in the global dense subgraph $x^*$ as the LSH query, the probability to retrieve each of the other vertexes in $x^*$ is lower bounded by $p\in(0,1)$ (i.e., the recall of LSH~\cite{LSH}).

Since CIVS uses all the $a^{(c)}$ different vertexes in the local dense subgraph $\hat x^{(c)}$ as LSH queries, we can derive that the probability to retrieve each of the vertexes in the global dense subgraph $x^*$ is lower bounded by $1-(1-p)^{a^{(c)}}$.

Considering that there are totally $m^{(c)}$ vertexes of $x^*$ inside the ROI $H_c(D,R)$, we can derive that the number of vertexes of $x^*$ retrieved by CIVS is a random variable $b^{(c)}$ that submits to binomial distribution
\begin{equation}
\label{Eqn:b_c}
b^{(c)}\sim B(m^{(c)}, 1-(1-p)^{a^{(c)}})
\end{equation}

\nop{
These $b^{(c)}$ vertexes retrieved by CIVS are used to update the local range $\beta^{(c)}\rightarrow\beta^{(c+1)}$.
}

Here, $b^{(c)}$ is the number of vertexes of $x^*$ in the updated local range $\beta^{(c+1)}=\alpha^{(c)}\cup \psi$.
As proved in Case 2, such vertexes form a subgraph $x'\in\triangle_{\beta^{(c+1)}}^n$ of the global dense subgraph $x^*$, and LID is guaranteed to converge to the local dense subgraph $\hat x^{(c+1)}=x'$. Therefore, we can further derive from Equation~\ref{Eqn:b_c} that the statistical expectation of the number of vertexes in $\hat x^{(c+1)}$ is $a^{(c+1)}$, where
\begin{equation}
\label{Eqn:a_c_series}
a^{(c+1)}=E(b^{(c)})=m^{(c)}\left[1-(1-p)^{a^{(c)}}\right]
\end{equation}

Recall that $m^{(c)}$ is the number of vertexes of $x^*$ within ROI $H_c(D,R)$, since there are totally $M$ vertexes in $x^*$, we have $m^{(c)}\leq M$.

For any iteration $c$, if $m^{(c)}<M$, then there are $(M-m^{(c)})$ vertexes of $x^*$ outside the ROI $H_c(D,R)$ and we can derive
\begin{equation}
\label{Eqn:xc_fact}
\hat x^{(c+1)}\neq x^*
\end{equation}
from the fact that $\hat x^{(c+1)}$ is the local dense subgraph within the ROI $H_c(D,R)$ and $m^{(c)}<M$.

Then, we prove by contradiction that at least one of the $(M-m^{(c)})$ vertexes of $x^*$ is infective against the local dense subgraph $\hat x^{(c+1)}$, as follows.
If none of the $(M-m^{(c)})$ vertexes are infective against $\hat x^{(c+1)}$, then, as indicated by Theorem~\ref{theorem:nash_gamma_equal}, $\hat x^{(c+1)}$ will be the global dense subgraph, which leads to $\hat x^{(c+1)}=x^*$. This is in contradiction with the fact that $\hat x^{(c+1)}\neq x^*$ in Equation~\ref{Eqn:xc_fact}. Therefore, at least one of the $(M-m^{(c)})$ vertexes is infective against $\hat x^{(c+1)}$.

Since the ROI $H_{(c+1)}(D,R)$ is guaranteed by Proposition~\ref{prop:triangle_inequality_hyperball} to obtain all the new infective vertexes against $\hat x^{(c+1)}$, we can derive that $m^{(c+1)}>m^{(c)}$. Therefore, we have $m^{(c)}<m^{(c+1)}\leq M$,
which indicates that the series $\{m^{(c)}\}$ is an increasing series with a reachable upper bound $M$.

Recall that $p\in(0,1)$, we can further derive from Equation~\ref{Eqn:a_c_series} that the series of $\{a^{(c)}\}$ is an increasing series with respect to the iteration times $c$.

Since the total number of vertexes in the global dense subgraph $x^*$ is $M$, we have $a^{(c)}\leq M$. Since $\{a^{(c)}\}$ is an increasing series, we can derive from Equation~\ref{Eqn:a_c_series} that the series $\{a^{(c)}\}$ converges to $M$, and a larger value of $p$ leads to a faster convergence rate. Recall that $a^{(c)}$ is the statistical expectation of the number of vertexes in local dense subgraph $\hat x^{(c)}$, thus, the second statement of Proposition~\ref{prop:convergence} is proved.
\end{proof}

\section{Noise Resistance Analysis}\label{sec:noise-analysis}
The purpose of noise resistance analysis is to compare the noise resisting performance of the affinity-based methods (i.e., ALID, AP, IID, SEA) and the partitioning-based methods, such as: 1) \emph{k}-means (KM) \cite{KM_Lloyd}; 2) spectral clustering using full affinity matrix (SC-FL) \cite{SC_tradition}; 3) spectral clustering using nystrom approximation (SC-NYS) \cite{NYS_PAMI04}. For the partitioning-based methods, we use the source codes published by Chen \emph{et al.} \cite{PSC}. The mean-shift method (MS)~\cite{MS_PAMI_most_cited} is analyzed as well.
 
We compare the AVG-F performance of each method on a series of data sets with increasing noise degree
\begin{equation}
\label{Eqn:noise_degree}
\text{noise degree} = \frac{\# \emph{\textrm{ noise\_data}}}{\# \emph{\textrm{ ground\_truth\_data}}}
\end{equation}
which is the ratio of the number of noise data items over the number of data items in the dominant clusters according to the ground truth.
Such series of data sets are obtained by injecting a certain amount of noise data into the ground truth data set.

For partitioning-based methods KM, SC-FL and SC-NYS, we set the cluster number $K$ in the same way as Liu~\textit{et~al.}~\cite{SEA,GS} did, which counts the noise data as an extra cluster. We use the same settings as in Section~\ref{sec:exp} for the affinity-based methods and adopts the classic gaussian kernel for the mean shift method.  All compared methods are optimally tuned. To avoid the performance degeneration caused by the enforced sparsity of the affinity graph, we use a full affinity matrix for methods AP, SEA and IID to preserve the original cohesiveness of all dense subgraphs. Due to the limited scalability of method AP, we use the two small data sets NART and Sub-NDI.

\begin{figure}[h]
\mbox{
\hspace{-3mm}
\subfigure[NART]{\includegraphics[width=42mm]{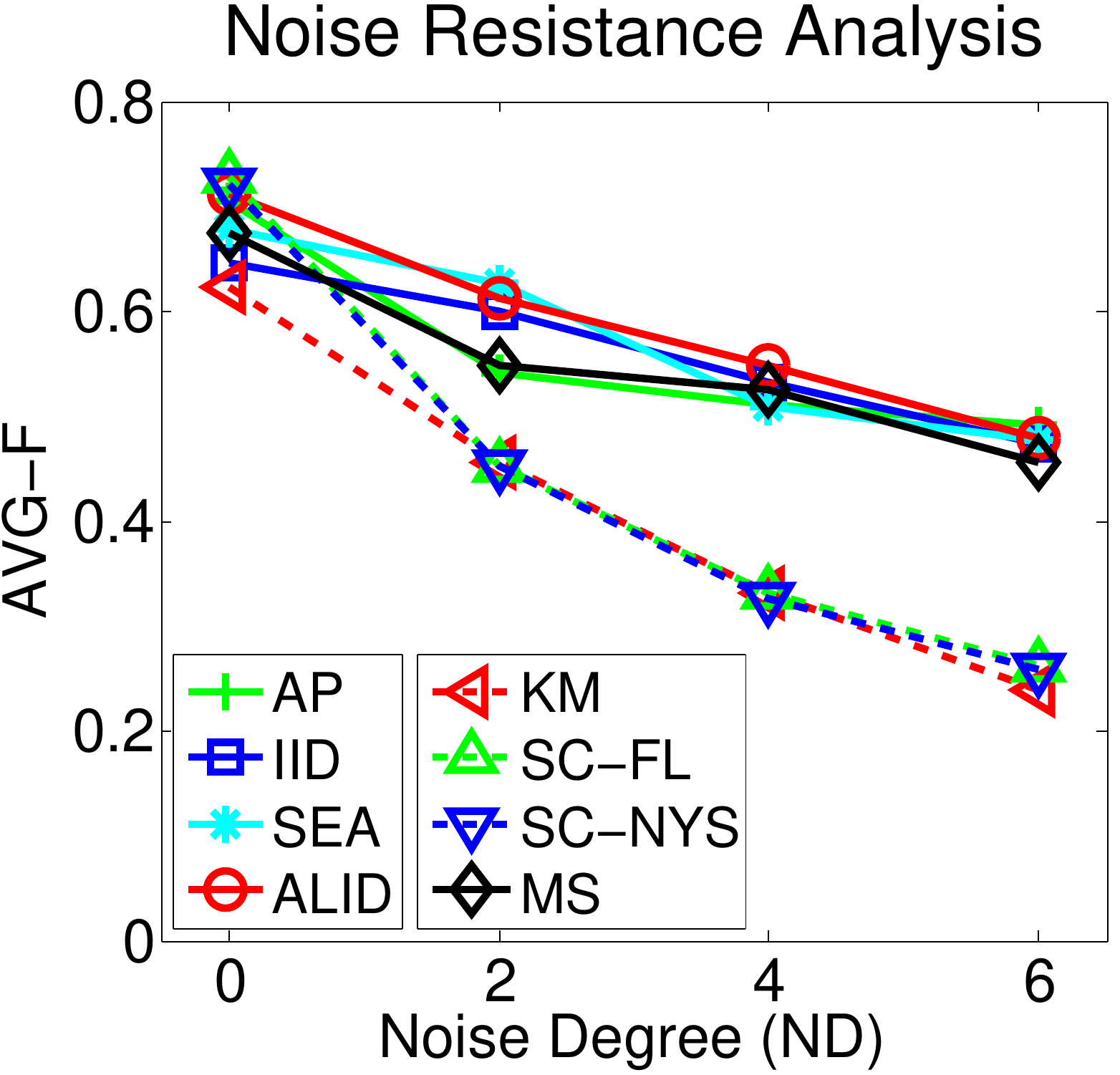}}
\subfigure[Sub-NDI]{\includegraphics[width=42mm]{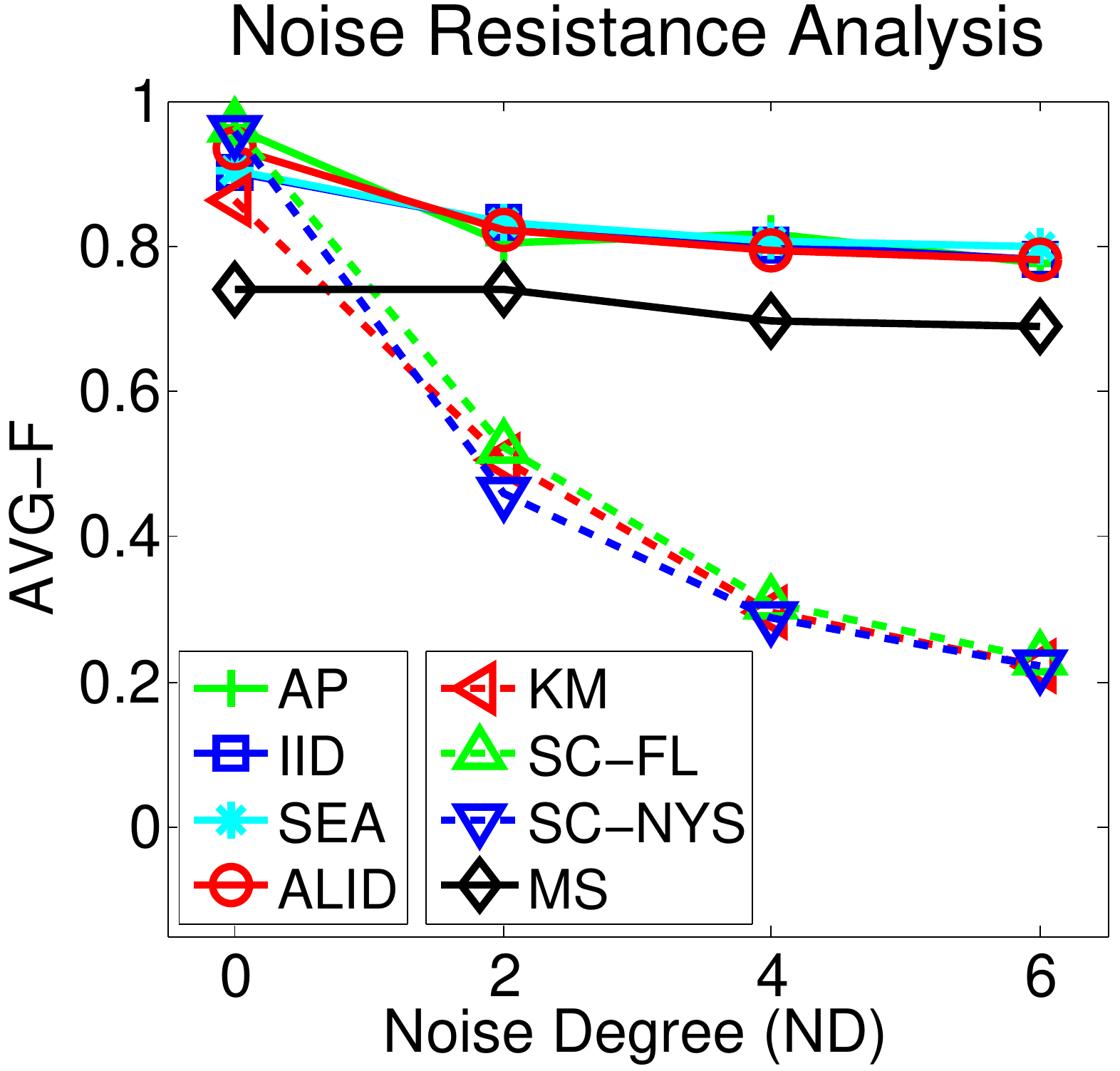}}
}
\caption{The noise resisting performances.  (a) The results on the NART data set. (b) The results on the Sub-NDI data set.}
\label{Fig:NRA}
\end{figure}

As shown in Figure~\ref{Fig:NRA}(a)-(b), the AVG-F performance of the partitioning-based methods decreases much faster than the other methods when noise degree increases. Such phenomenon is caused by the core mechanism of partitioning-based methods, which partitions all data items (including the noise data) into a fixed number of clusters. Thus, such mechanism is ineffective on noisy data. Similar observations is also reported by Pavan~\textit{et~al.}~\cite{DS}. In contrast, the affinity-based methods significantly outperform the partitioning-based methods on noisy data when the noise degree is large. This demonstrates the advantage of the affinity-based methods in resisting high background noise.

It is easy to understand why both SC-FL and SC-NYS achieve good AVG-F performance when ``$\text{noise degree}=0$'', since those methods are fed with the correct number of clusters. However, the number of clusters is most likely unknown in practical applications.

As shown in Figure~\ref{Fig:NRA}(a), MS achieves a comparable AVG-F performance with the affinity-based methods on the NART data set. This is achieved by a proper bandwidth setting of the pre-defined density distribution, which makes the MS method noise-resisting if the bandwidth properly fits the scale of most true clusters.
However, on the Sub-NDI data set (Figure~\ref{Fig:NRA}(b)), MS failed to achieve a comparable AVG-F with the affinity-based methods, since the distribution of the image features in Sub-NDI data set is more sophisticated than the text features in NART data set. Such performance degeneration is probably caused by the strong reliance of MS on the bandwidth and the pre-defined density distribution, which cannot simultaneously fit on various clusters under complex feature distributions.

In summary, {\ALID} and the other affinity-based methods are more robust than the partitioning-based methods in detecting unknown number of dominant clusters from noisy background.

\end{sloppy}
\end{document}